\documentclass[%
prl,
reprint,
superscriptaddress, nofootinbib,amsmath,amssymb,aps,
floatfix
]{revtex4-2}

\usepackage{enumitem}
\usepackage{mathtools}

\usepackage{hyperref,url}
\hypersetup{breaklinks,colorlinks=true,linkcolor=blue,citecolor=blue,urlcolor=blue,filecolor=blue}

\usepackage[dvipsnames]{xcolor}
\usepackage{graphicx}%
\usepackage{xcolor}

\usepackage{booktabs}
\usepackage{makecell} %
\usepackage{dcolumn}%

\usepackage{subcaption}

\usepackage{quantikz}
\usepackage{physics}
\usepackage{amsmath}
\usepackage{bm}%
\usepackage{mathtools}
\usepackage{amsfonts}
\usepackage{amsthm, thmtools} \theoremstyle{definition}

\usepackage[ruled]{algorithm2e}

\usepackage[capitalise]{cleveref}%

\newcommand{\Google}{\affiliation{Google Quantum AI, Mountain View, CA, USA}}
\newcommand{\citen}[1]{Ref.~\citenum{#1}}
\newcommand{\etal}{\textit{et al.}}

\usepackage[english]{babel}

\makeatletter
\def\bbl@set@language#1{%
  \edef\languagename{%
    \ifnum\escapechar=\expandafter`\string#1\@empty
    \else\string#1\@empty\fi}%
  \@ifundefined{babel@language@alias@\languagename}{}{%
    \edef\languagename{\@nameuse{babel@language@alias@\languagename}}%
  }%
  \select@language{\languagename}%
  \expandafter\ifx\csname date\languagename\endcsname\relax\else
    \if@filesw
      \protected@write\@auxout{}{\string\select@language{\languagename}}%
      \bbl@for\bbl@tempa\BabelContentsFiles{%
        \addtocontents{\bbl@tempa}{\xstring\select@language{\languagename}}}%
      \bbl@usehooks{write}{}%
    \fi
  \fi}
\newcommand{\DeclareLanguageAlias}[2]{%
  \global\@namedef{babel@language@alias@#1}{#2}%
}
\makeatother

\DeclareLanguageAlias{en}{english}

\usepackage[colorinlistoftodos, textsize=scriptsize, color=orange!50]{todonotes}

\newcommand{\citeneeded}[1]{[{\color{blue!50}\textbf{CITE}}]}

\theoremstyle{definition}
\newtheorem{theorem}{Theorem}

\newtheorem{corollary}[theorem]{Corollary}

\SetKwFor{UnaryFor}{multiplexing over register}{do}{end}
\SetKw{MiddleFor}{, for}

\crefname{algocf}{alg.}{algs.}
\Crefname{algocf}{Algorithm}{Algorithms}

\newcommand{\bigo}[1]{\mathcal{O}(#1)}

\newcommand{\bigot}[1]{\tilde{\mathcal{O}}\left( #1\right)}

\newcommand{\doubleplus}{\mathbin{{+}\!\!\!{+}}}
\DeclareMathOperator{\prep}{\mathtt{PREPARE}}
\DeclareMathOperator{\sel}{\mathtt{SELECT}}

\begin{document}

\title{Productionizing Quantum Mass Production}

\author{William J. Huggins}
\email{whuggins@google.com}
\Google

\author{Tanuj Khattar}
\Google

\author{Nathan Wiebe}
\affiliation{Department of Physics, University of Toronto, Toronto, ON, Canada}
\affiliation{Department of Computer Science, University of Toronto, Toronto, ON, Canada}
\affiliation{Pacific Northwest National Laboratory, Richland, WA, USA}

\date{\today}

\newcommand{\ceil}[1]{\left\lceil #1 \right\rceil}

\begin{abstract}
	For many practical applications of quantum computing, the most costly steps involve coherently accessing classical data.
	We help address this challenge by applying mass production techniques, which can reduce the cost of applying an operation multiple times in parallel~\cite{Ulig1974-xh, Uhlig1992-tg,Kretschmer2022-uy}.
	We combine these techniques with modern approaches for classical data loading based on ``quantum read-only memory'' (QROM).
        We find that we can polynomially reduce the total number of gates required for data loading, but we find no advantage in cost models that only count the number of non-Clifford gates.
        Furthermore, for realistic cost models and problem sizes, we find that it is possible to reduce the cost of parallel data loading by an order of magnitude or more.
	We present several applications of quantum mass production, including a scheme that uses parallel phase estimation to asymptotically reduce the gate complexity of state-of-the-art algorithms for estimating eigenvalues of the quantum chemical Hamiltonian, including both Clifford and non-Clifford gates, from \(\bigot{N_{orb}^2}\) to \(\bigot{N_{orb}^{\log_2 3}}\), where \(N_{orb}\) denotes the number of orbitals.
	We also show that mass production can be used to reduce the cost of serial calls to the same data loading oracle by precomputing several copies of a novel QROM resource state.
\end{abstract}
\maketitle

\section{Introduction}

The high cost of loading classical data poses an obstacle to the search for quantum advantage in machine learning, data analysis, and many other contexts.
For example, some of the most promising quantum algorithms for the electronic structure problem make heavy use of classical preprocessing and data loading~\cite{von-Burg2021-yq, Lee2021-su, Caesura2025-wl, Low2025-ei}.
Even when the amount of data is modest by classical standards, data loading may be costly when we account for the slow speeds we expect from fault-tolerant quantum computers~\cite{Babbush2021-aq, Beverland2022-dh}.
We apply recently-developed ``quantum mass production'' techniques to help address this problem~\cite{Kretschmer2022-uy}.

These mass production techniques build on the earlier work of Uhlig, who showed how to construct a classical circuit for evaluating an arbitrary \(n\)-bit boolean function in parallel \(r\) times with a circuit complexity of \(\mathcal{O}\left( 2^n / n \right)\), so long as \(r\) is not too large (\(r = 2^{o(n / \log n)}\))~\cite{Ulig1974-xh,Uhlig1992-tg}.
Asymptotically, this matches the scaling required to implement a single arbitrary \(n\)-bit boolean function.
Ref.~\citenum{Kretschmer2022-uy} developed analogues of these results for preparing arbitrary quantum states or synthesizing arbitrary unitaries in parallel.
In both cases, gate complexity of implementing \(r = 2^{o(n / \log n)}\) parallel operations is asymptotically proportional to the complexity of performing the basic task a single time.

We apply these mass production ideas to the task of querying a boolean function in superposition, sometimes referred to as ``quantum read-only memory'' (QROM)~\cite{Babbush2018-tb}.
Given a function \(f: \left\{ 0, 1 \right\}^n \rightarrow \left\{ 0, 1 \right\}^m,\) the task is to implement the oracle 
\begin{equation}
	O_f: \ket{x}\ket{\alpha} \rightarrow \ket{x} \ket{\alpha \oplus f(x)}.
	\label{eq:O_f_def}
\end{equation}
There is a significant body of work on optimizing the implementation of such an \(O_f\)~\cite{Low2018-uu,Berry2019-qo, Zhu2024-pe, Giovannetti2008-un, Giovannetti2007-cg, Arunachalam2015-vr, Jaques2023-fo}.
We review some of these developments in \Cref{app:background_and_qrom}, including the ``SelectSwap'' QROM that we use in this work~\cite{Low2018-uu, Berry2019-qo}.
This approach reduced the number of non-Clifford T or Toffoli gates required, at the expense of requiring additional space.
Historically, it has been believed that non-Clifford gates will be vastly more costly to implement for fault-tolerant quantum computing using the surface code, making this tradeoff appealing~\cite{Bravyi2005-vi, Bravyi2012-mg, Fowler2018-he}.
Continued progress on better techniques for implementing non-Clifford gates challenges this assumption~\cite{Litinski2019-ek, gidney2024magic}, motivating us to consider a cost model that does not neglect the cost of Clifford gates.

We outline our mass production protocol and the intuition behind it below, formalizing asymptotic scaling in \Cref{thm:mass_produced_qrom}.
Informally, we show that the cost of implementing \(r\) queries to \(O_f\) in parallel is asymptotically proportional to the cost of implementing a single query, so long as \(r\) is not too large.
We estimate the actual savings afforded by our construction for mass production using the Qualtran software package~\cite{Harrigan2024-rj}, finding that it is possible to achieve practical benefits at reasonable problem sizes.
Interestingly, this advantage only appears when accounting for the cost of both Clifford and non-Clifford gates.
We go on to discuss applications of our protocol, showing that we can reduce the gate complexity for certain applications of parallel phase estimation and amplitude amplification.
We even find that mass production can be used to reduce the (amortized) cost of serial queries to the same data loading oracle.

\section{Mass-produced QROM queries}
\label{sec:mass_producing_qrom}

Given a function \(f: \left\{ 0,1 \right\}^n \rightarrow \left\{ 0, 1 \right\}^m\), we aim to implement \(O_f^{\otimes r}\) for some \(r\) that is a power of two.
As with Refs.~\citenum{Ulig1974-xh,Uhlig1992-tg,Kretschmer2022-uy}, the technique for accomplishing this task most efficiently relies on a clever strategy for reducing the implementation of \(O_f^{\otimes r}\) to a collection of smaller data loading tasks.
We begin our construction by describing the mass production protocol for implementing two parallel operations.
We then construct the general protocol recursively.
We give a high-level overview of the procedure here and refer the reader to \Cref{app:mass_production_protocol} for more details.

We begin by defining a family of functions \(\left\{ f_{\ell} \right\}\), where \(\ell\) ranges from \(0\) to \(2^{k} - 1\) for a constant number of bits $k<n$. 
For each \(\ell\), we let \(f_{\ell}\) denote the function \(f\) with the first \(k\) bits fixed to (the binary encoding of) \(\ell\).  For example, if $k=1$ then the first bit will be fixed to $\ell$.
Then we define a family of functions,
\begin{align}
	g_0 =     & f_0 \nonumber,
	\\
	g_\ell =  & f_{\ell - 1} \oplus f_\ell \text{ for } 1 \leq \ell \leq 2^k - 1,
	\label{eq:g_def}
	\\
	g_{2^k} = & f_{2^k - 1} \nonumber.
\end{align}
Mass production takes advantage of the following equation to evaluate \(f\) on two inputs while only evaluating each \(g_\ell\) at most once,
\begin{equation}
	f_\ell(z) = \bigoplus_{j=0}^{\ell} g_j(z) = \bigoplus_{j=\ell+1}^{2^k}g_j(z).
	\label{eq:f_ell_g_ell_trick}
\end{equation}

Our aim is to implement a circuit
\begin{multline}
	C: \ket{x_L}\ket{x_R}\ket{\alpha}\ket{y_L}\ket{y_R}\ket{\beta} 
	\\
	\rightarrow \ket{x_L}\ket{x_R}\ket{\alpha \oplus f(x)}\ket{y_L}\ket{y_R}\ket{\beta \oplus f(y)},
	\label{eq:desired_C}
\end{multline}
where \(x_L\) (\(y_L\)) denotes the first \(k\) bits of \(x\) (\(y\)) and \(x_R\) (\(y_R\)) denotes the remaining \(n-k\) bits.
We begin by defining a data loading subroutine for each \(g_\ell\),
\begin{equation}
	G_\ell: \ket{z}\ket{\gamma} \rightarrow \ket{z}\ket{\gamma \oplus g_\ell(z)}.
\end{equation}
Using a series of arithmetic comparisons and controlled swap gates, we can control whether \(G_\ell\) is applied to the registers initially containing \(x_R\) and \(\alpha\) or \(y_R\) and \(\beta\).

We work by applying each \(G_{\ell}\) in sequence.
When \(x_L \leq \ell\), we route the inputs and outputs so that the effect of \(G_\ell\) is to XOR \(g_\ell(x_R)\) into the \(\alpha\) register.
Similarly, when \(y_L > \ell\), we ensure that \(g_\ell(y_R)\) is XORed into the \(\beta\) register.
After this process, the output register initially set to \(\ket{\alpha}\) contains the state \(\ket{\alpha \bigoplus_{j=0}^{x_L} g_j(x_R)}\) and the other output register contains \(\ket{\beta \bigoplus_{j=y_L + 1}^{2^k} g_j(y_R)}\).
Applying the identities in \Cref{eq:f_ell_g_ell_trick}, we find that we have implemented \Cref{eq:desired_C} as desired.

We defer analysis of the cost savings to \Cref{app:base_case_cost_analysis} and sketch the argument here.
The dominant cost of the construction is the \(2^k + 1\) calls to the data loading oracles \(G_{\ell}\).
Working in a cost model where we account for Clifford gates, the cost of implementing each \(G_\ell\) is bounded by \(2^{n-k} m K\) for some constant \(K\).
Summing these costs, we have 
\begin{equation}
	\sum_{\ell=0}^{2^k} \textsc{Cost}\left( G_\ell \right) \leq  2^n m K \frac{2^{k} + 1}{2^k}.
\end{equation}
As \(k\) grows, this cost approaches that of implementing a single call to \(O_f\), which has a cost bounded by \( 2^n m K\).
As we show in \Cref{app:base_case_cost_analysis}, this conclusion breaks down if we neglect the cost of Clifford gates.

The insight that enables us to recursively construct the general protocol is that data loading is present in the two-copy case as a subroutine.
Imagine implementing \(r/2\) copies of the two-query mass production protocol in parallel.
Doing so would involve implementing \(r/2\) queries to each of the data loading oracles \(G_{\ell}\) in parallel.
Instead of doing this, we implement the \(r/2\) parallel queries to each \(G_{\ell}\) using an \(r/2\)-query mass production protocol.
We illustrate this for the \(r=4\) case in \Cref{fig:recursion_cartoon}.
Since the cost of each two-copy scheme is dominated by the \(G_{\ell}\), this replacement will result in an \(r\)-query protocol with a cost comparable to the cost of a single \(r/2\)-query protocol.
We formalize this below in \Cref{thm:mass_produced_qrom}, which we prove in \Cref{app:scaling_thm_proof}.

\begin{figure}
	\begin{quantikz}[wire types={q,q,q,q,q,q,q,q}, row sep={0.40cm,between origins}, column sep=0.34cm, font=\small]
    \lstick{$\ket{x_1}$}      & \gate[2,style={rounded corners, fill=red!25}, disable auto height][.85cm]{G_0} \gategroup[6,steps=1,style={dashed,rounded corners,fill=blue!35, inner sep=0pt},label style={label position=above}, background]{$\bm{C \left(g_0, 2 \right)}$} \gategroup[wires=4, steps=5, style={dotted, draw=black, rounded corners,fill=black, inner sep=0pt, draw opacity=.6, fill opacity=.05}]{}  & \gate[4,style={rounded corners}, disable auto height][.85cm]{A_0} & \ldots & \gate[2,style={rounded corners, fill=red!25}, disable auto height][.85cm]{G_{2^k}} \gategroup[6,steps=1,style={dashed,rounded corners,fill=blue!35, inner sep=0pt}, background,label style={label position=above}]{$\bm{C\left(g_{2^{k}}, 2\right)}$} & \gate[4,style={rounded corners}, disable auto height][.85cm]{A_{2^k}} & \rstick{$\ket{x_1}$}
    \\
    \lstick{$\ket{\alpha_1}$} & \qw                                                                                                                                                                                                                              & \qw                                                               & \ldots & \qw                                                                                                                                                                                                                                                            & \qw                                                                   & \rstick{$\ket{\alpha_1 \oplus f(x_1)}$}
    \\
    \lstick{$\ket{x_2}$}      & \qw                                                                                                                                                                                                                              & \qw                                                               & \ldots & \qw                                                                                                                                                                                                                                                            & \qw                                                                   & \rstick{$\ket{x_2}$}
    \\
    \lstick{$\ket{\alpha_2}$}  & \qw                                                                                                                                                                                                                              & \qw                                                               & \ldots & \qw                                                                                                                                                                                                                                                            & \qw                                                                   & \rstick{$\ket{\alpha_2 \oplus f(x_2)}$}
    \\
    [.5cm]
    \lstick{$\ket{x_3}$}      & \gate[2,style={rounded corners, fill=red!25}, disable auto height][.85cm]{G_0} \gategroup[wires=4, steps=5, style={dotted, draw=black, rounded corners,fill=black, inner sep=0pt, draw opacity=.6, fill opacity=.05}]{}                                                                                                                                                                & \gate[4,style={rounded corners}, disable auto height][.85cm]{A_0} & \ldots & \gate[2,style={rounded corners, fill=red!25}, disable auto height][.85cm]{G_{2^k}}                                                                                                                                                                                          & \gate[4,style={rounded corners}, disable auto height][.85cm]{A_{2^k}} & \rstick{$\ket{x_3}$}
    \\
    \lstick{$\ket{\alpha_3}$} & \qw                                                                                                                                                                                                                              & \qw                                                               & \ldots & \qw                                                                                                                                                                                                                                                            & \qw                                                                   & \rstick{$\ket{\alpha_3 \oplus f(x_3)}$}
    \\
    \lstick{$\ket{x_4}$}      & \qw                                                                                                                                                                                                                              & \qw                                                               & \ldots & \qw                                                                                                                                                                                                                                                            & \qw                                                                   & \rstick{$\ket{x_4}$}
    \\
    \lstick{$\ket{\alpha_4}$}  & \qw                                                                                                                                                                                                                              & \qw                                                               & \ldots & \qw                                                                                                                                                                                                                                                            & \qw                                                                   & \rstick{$\ket{\alpha_4 \oplus f(x_4)}$}
\end{quantikz}
	\caption{
	An illustration of how the \(r=4\)-query mass production protocol is constructed recursively.	
	Two copies of the \(2\)-query protocol are laid out in parallel, each surrounded by dotted lines.
	We replace the parallel calls to each \(G_\ell\) (pale red) with calls to the \(2\)-query mass production protocol for \(g_\ell\) (shaded blue and outlined with dashed lines).}
	\label{fig:recursion_cartoon}
\end{figure}
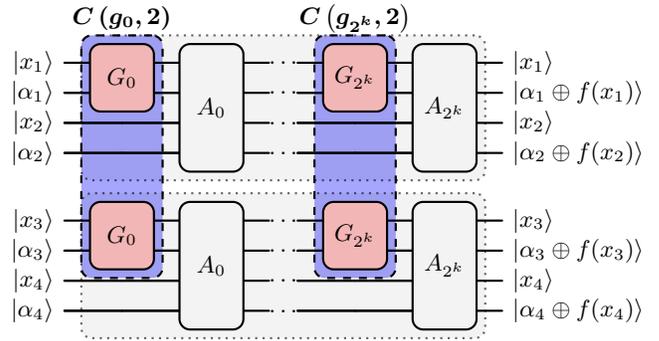

\begin{restatable}[Quantum Mass Production for quantum read-only memory]{theorem}{qromthm}
	Let \(f: \left\{ 0, 1 \right\}^n \rightarrow \left\{ 0, 1 \right\}^m\) be arbitrary, and let \(O_f\) be a unitary operator that queries this function in superposition, \(O_f: \ket{x}\ket{\alpha} \rightarrow \ket{x} \ket{\alpha \oplus f(x)}\).
	Let \(\lambda\) be a function of \(n\) that returns a power of \(2\) such that \(	\lambda = o\left(\frac{2^{n/2}}{\sqrt{n}}\right)\) and \(1 \leq \lambda \leq 2^{n/2} m^{-1/2}\), and let \(m\) be a constant function of $n$.
	For any \(r\) which is a power of two such that \(r = 2^{o \left( \frac{n - 2 \log \lambda}{\log(n)} \right)}\), there exists a quantum circuit \(C\) that implements \(O_f^{\otimes r}\) and satisfies the following properties:

	\begin{enumerate}
		\item \(C\) is composed entirely of one- and two-qubit Clifford gates and Toffoli gates.

		\item The number of one- and two-qubit Clifford gates in \(C\) is bounded by
		      \begin{equation}
			      \textsc{Clifford}\left( C \right) \leq \left( K + o(1) \right) 2^n m
				  \label{eq:thm_clifford_costs}
		      \end{equation}
		      for some universal constant \(K\).

		\item The number of Toffoli gates in \(C\) is
		      $
			      \textsc{Toffoli}\left( C \right) = \left( 1 + o(1) \right) 2^n \lambda^{-1}.
		      $
	\end{enumerate}
	\label{thm:mass_produced_qrom}
\end{restatable}

We note that the \(\lambda\) parameter in the statement of this theorem arises from the usage of ``SelectSwap'' QROM as a subroutine, which allows for a reduced circuit depth and number of Toffoli gates at the expense of requiring $\approx \lambda m$ ancilla qubits~\cite{Low2018-uu}.
While we focus on computational complexity here, it might sometimes be desirable to limit the choice of $\lambda$ to reduce the number of qubits used.
Conversely, choosing a sufficiently large $\lambda$ could make the additional space required for mass production negligible, leading to an overall spacetime cost that would be approximately proportional to the number of one- and two-qubit Clifford gates.
Finally, we observe that the number of copies we can mass produce ($r$) depends on $\lambda$ in a way that is consistent with our claim that we see no benefit from mass production when we ignore the cost of Clifford gates (and maximize $\lambda$ to reduce the number of Toffoli gates).

\section{Constant factor analysis}
\label{sec:practical_costs}

We use the Qualtran software package (See~\citen{Harrigan2024-rj}) to study the gate complexity of our mass production protocol in this section.
We calculate the gate complexity of a circuit \(C\) using the formula
\begin{equation}
	\textsc{Cost}(C) = \textsc{Clifford}(C) + \Xi \textsc{T}(C),
	\label{eq:cost_model_def}
\end{equation}
where \(\textsc{Clifford}(C)\) denotes the number of one and two-qubit Clifford gates, \(\textsc{T}(C)\) denotes the number of non-Clifford T gates, and \(\Xi\) is a constant (see \Cref{app:cost_models} for a discussion of this model).
To quantify the advantage provided by a circuit \(C_r\) that mass-produces \(r\) queries to a data loading oracle \(O\), we calculate the ``improvement factor''
\begin{equation}
	\mathcal{I} = \frac{\textsc{Cost}\left( O^{\otimes r} \right)}{\textsc{Cost}\left( C_r \right)} = \frac{r \textsc{Cost}\left( O \right)}{\textsc{Cost}\left( C_r \right)}.
\label{eq:improvement_factor_def}
\end{equation}
For all of the data we present in this section, we minimize the costs of both \(C_r\) and \(O\) by varying the tunable parameter \(\lambda\) (described in \Cref{app:qrom_background}) and the value of \(k\) chosen at each step in the recursive construction for mass production.
See \Cref{app:numerical_details} for additional details on our numerical experiments.

\begin{figure}[t]
	\centering
	\includegraphics[width=.48\textwidth]{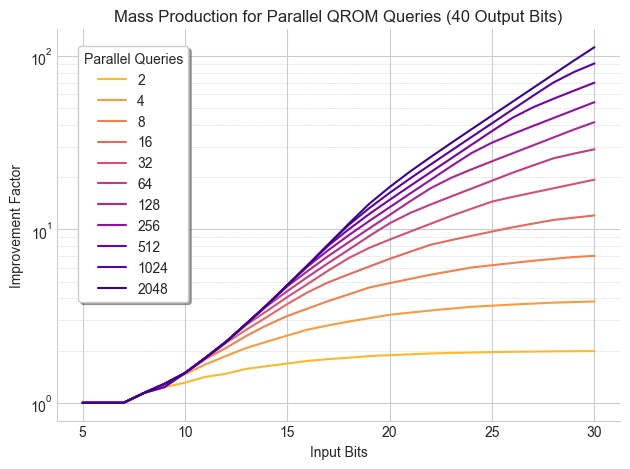}
	\caption{
		The ``improvement factor'' $\mathcal{I}$ when using mass production to implement \(r\) parallel queries to a data loading oracle for an arbitrary function that takes \(n\) input bits and outputs \(40\) bits. %
		When the number of parallel queries is small and the number of input bits is large, we observe that the improvement factor nearly attains the upper bound of \(r\), implying that the entire mass production protocol is only slightly more expensive than a single query without mass production.
	}
	\label{fig:mass_production_helps}
\end{figure}

We plot the improvement factor achievable for mass-producing queries to an arbitrary function of the form \(f: \left\{ 0, 1 \right\}^n \rightarrow \left\{ 0, 1 \right\}^{40}\) in \Cref{fig:mass_production_helps}.
We vary \(n\) and plot the improvement factor available when implementing \(O_f^{\otimes r}\) for several values of \(r\), taking \(\Xi = 1\).
At a fixed value of \(r\), the improvement factor converges towards its upper bound of \(r\) as the number of input bits (\(n\)) increases.
For the largest values of \(n\) and \(r\) we consider in this figure, we find an improvement factor of slightly more than \(10^2\), which implies that we can perform \(2048\) parallel queries to a function \(f: \left\{ 0, 1 \right\}^{30} \rightarrow \left\{ 0,1 \right\}^{40}\) with roughly the same cost as performing \(20\) using standard techniques.
Specifying such an \(f\) requires more than \(200\) billion parameters, suggesting that implementations on this scale will have to wait for very large fault-tolerant quantum computers.
However, we find that improvements by a factor of \(10\) or more are reasonable at smaller input sizes.
For example, we analyze the sparse simulation method of \citen{Berry2019-qo} in \Cref{app:lcu_overview} and find that model systems with sizes ranging from \(100\) to \(200\) spin-orbitals require between \(14\) and \(18\) input qubits.
The number of output qubits required for a desired level of accuracy can vary, but the original work considered output sizes of \(\approx 80\) for several systems \cite{Berry2019-qo}.

In \Cref{fig:T_gate_overhead}, we study how the cost of the non-Clifford T gate (relative to one- and two-qubit Clifford gates) impacts the effectiveness of mass production.
Here we fix the number of input bits (\(n\)) to \(20\) and the number of output bits (\(m\)) to \(40\) and plot the improvement factor available as a function of the T gate overhead (\(\Xi\), as defined in \Cref{eq:cost_model_def}) for several values of \(r\).
We see that the potential benefit from mass production is strongly suppressed as \(\Xi\) increases.
We discuss this behavior in \Cref{app:base_case_cost_analysis}, where we see that it emerges because the non-Clifford gate complexity of data loading scales as \(2^{n/2}\) rather than \(2^n\)~\cite{Low2018-uu}.
This suggests that recent work on reducing the cost of non-Clifford gates will be crucial to realizing the benefits of mass production for data loading.
For example, \citen{Litinski2019-ek} presents data that supports a value of \(\Xi\) roughly between \(5\) and \(20\) and the more recent \citen{Craig2024-kv} argues that it may be possible to achieve a value of \(\Xi\) close to \(1\).%

\begin{figure}[t]
	\centering
	\includegraphics[width=.48\textwidth]{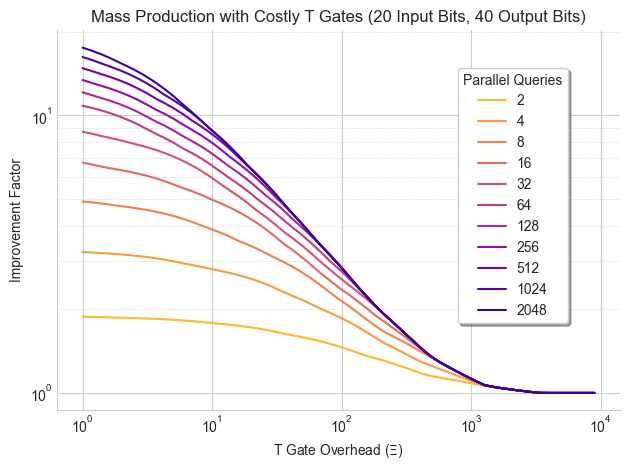}
	\caption{
		The improvement factor $\mathcal{I}$ available when using mass production to implement \(r\) parallel queries to a data loading oracle for an arbitrary \(f\) with \(20\) input bits and \(40\) output bits, plotted as a function of the T gate overhead (\(\Xi\) in \Cref{eq:cost_model_def}).
		When \(\Xi = 1\) and the cost of a T gate is taken to be the same as an arbitrary one- or two-qubit Clifford gate, mass production can offer reasonably large improvements.
		The potential benefit of using mass production is suppressed as \(\Xi\) increases.
	}
	\label{fig:T_gate_overhead}
\end{figure}

\section{Applications}

We highlight a few promising applications of quantum mass production in this section, beginning with its application to parallel phase estimation.
This approach involves sharing a GHZ state over $r$ workers that each have a copy of a state $\ket{\psi}$ such that \(U \ket{\psi} = e^{-iE} \ket{\psi}\) for some unitary \(U\).
Using phase kickback, each worker can apply the same controlled unitary to collectively accumulate the phase \(e^{-i r E}\).
Ordinarily, this would simply reduce the circuit depth.
However, as we explain in \Cref{app:lcu_overview} and \Cref{app:parallel_qpe}, classical data loading is the dominant cost in some of the state-of-the-art approaches for the quantum simulation of quantum chemistry~\cite{Lee2021-su,von-Burg2021-yq}.

We can therefore use quantum mass production to reduce the overall gate complexity of eigenvalue estimation compared with a standard serial phase estimation approach~\cite{Lee2021-su,von-Burg2021-yq}.
As examples, we examine the sparse simulation method of \citen{Berry2019-qo} and the tensor hypercontraction algorithm of \citen{Lee2021-su}.
By reanalyzing the mass production protocol for a choice of \(r\) such that \(r \propto 2^{n}\) (see \Cref{app:max_improvement_factor}), we are able to effectively reduce the cost of implementing the qubitized quantum walk operator by a factor that is polynomial in the system size, \(N_{orb}\).
Provided that the appropriate eigenstates are cheaply preparable, the combination of parallel phase estimation and mass production allows us to achieve the best known \textit{total} gate complexity for these methods. 
We expect that the same improvement should be available for the recently introduced techniques based on spectrum amplification~\cite{Low2025-ei}.

Amplitude amplification also can benefit from the use of mass production.
Given an algorithm \(A\) that produces some desired ``good'' state \(\ket{G}\) with a small amplitude \(\sqrt{p}\),
\begin{equation}
	A \ket{0^{\otimes n}} = \sqrt{p} \ket{G} \ket{1} + \sqrt{1 - p} \ket{B} \ket{0},
\end{equation}
amplitude amplification lets us output \(\ket{G}\) using \(\approx \frac{1}{\sqrt{p}}\) calls to \(A\).
In \Cref{app:parallel_amp_amp}, we show that quantum mass production can allow us to reduce the gate complexity of outputting \(\ket{G}\) by a factor of \(\approx \sqrt{r}\), provided that we can mass produce \(A^{\otimes r}\) at a cost that is comparable to \(A\), e.g., when the cost of \(A\) is dominated by data loading.
This is a promising direction because many quantum algorithms, including algorithms for linear algebra~\cite{Harrow2009-sm, Childs2017-is}, differential equations~\cite{Berry2014-fz, Berry2017-pd, Liu2021-ud}, ground state preparation~\cite{Lin2020-qh}, and machine learning~\cite{Lloyd2016-po, Liu2024-dv} involve heavy use of both amplitude amplification and data loading oracles that might be amenable to mass production.

Finally, we show how mass production can enable us to reduce the complexity of serial queries to the same data loading oracle.
The idea, which we elaborate on in \Cref{app:resource_state}, is simple.
We can use our scheme for mass-producing QROM queries to prepare multiple copies of the QROM resource state 
\begin{equation}
	\ket{\text{QROM}_{f}} = O_f \ket{+}^{\otimes n} \ket{0} = 2^{-n/2} \sum_{y=0}^{N - 1} \ket{y} \ket{f(y)}.
	\label{eq:qrom_resource_state_mt}
\end{equation}
We show how to consume one of these resource states to perform a data lookup that is effectively a scrambled version of the desired data lookup, \(\bar{O}_f^{(b)}: \ket{x}\ket{0} \rightarrow \ket{x}\ket{0 \oplus f(x \oplus b)}\), for some randomly determined bitstring \(b\).
We can correct this scrambling by performing another data lookup for a function \(g(x) = f(x) \oplus f(x \oplus b)\).
Furthermore, we can take advantage of the fact that \(g(x)\) outputs the same value for the inputs \(x\) and \(x \oplus b\) to implement this correction using half the resources required to implement \(O_f\) directly.
In many situations, this allows us to reduce the amortized cost of repeated queries to the same \(O_f\) by nearly a factor of two.

\section{Discussion}

In this paper, we have introduced a variant of quantum mass production that enables multiple parallel queries to a ``quantum read-only memory'' (QROM) to be performed at a cost that is comparable to the cost of a single query.
We used the Qualtran software package to analyze the cost of implementing these parallel queries~\cite{Harrigan2024-rj}, combining mass production ideas with state-of-the-art techniques for reducing the total number of non-Clifford gates.
In a naive cost model where we only account for the cost of non-Clifford gates, we find no benefit to mass-producing QROM queries.
However, in a more realistic cost model that accounts for both Clifford and non-Clifford gates, we find that quantum mass production can offer a practical benefit at reasonable problem sizes, reducing the cost of parallel queries by an order of magnitude or more.
This offers an alternative approach to benefiting from extra space that, unlike other advanced techniques for data loading~\cite{Low2018-uu, Berry2019-qo}, has the potential to reduce the overall gate complexity rather than merely the non-Clifford gate complexity.

We proposed several applications of quantum mass production that go beyond the straightforward idea of parallelizing independent executions of the same algorithm.
Focusing on the simulation of quantum chemistry as a concrete example, we discussed how mass production and parallel phase estimation might be combined to reduce the overall cost of eigenvalue estimation by parallelizing the \(\prep\) subroutine in a linear combination of unitaries framework.
More generally, we show how mass production could be used in combination with amplitude amplification to reduce the cost for a wide variety of quantum algorithms, provided that their overall cost is dominated by the cost of classical data loading.
These concrete examples we have discussed are only a small fraction of the possible applications of quantum mass production.
For example, the initial state preparation step for first-quantized chemical simulation naturally involves parallel applications of the same arbitrary unitaries~\cite{su2021fault,huggins2024efficient}, parallel calls to the same block encoding could be used to implement methods based on truncated Taylor series~\cite{Berry2015-ly}, and a host of quantum machine learning algorithms might benefit from cheaply encoding the same classical data multiple times in parallel~\cite{Lloyd2016-po, Gilboa2023-db}.

Additionally, we proposed a strategy that can reduce the amortized cost of serial queries to the same QROM by using mass production to precompute several copies of a QROM resource state.
We then consume these resource states to implement the desired query, up to a smaller correction.
This protocol, which we present in \Cref{app:resource_state}, can reduce the cost of repeated queries by almost a factor of two.
It would be interesting to understand if our protocol can be improved, or if there is a fundamental reason that the cost reduction is limited to a factor of two.
Outside of our general QROM resource state construction, there may be more specialized situations where mass production could be combined with precomputation to realize a speedup, perhaps using density matrix exponentiation as in \citen{Huggins2023-nx}.

\section{Acknowledgments}
We thank Craig Gidney for suggesting the resource state construction.  
We thank Stephen Jordan and Guang Hao Low for useful feedback on the manuscript.
NW acknowledges support from the NSERC funded Quantum Software Consortium and also U.S. Department of Energy, Office of Science,
National Quantum Information Science Research Centers, Co-design Center for Quantum Advantage (C2QA)
under contract number DE- SC0012704 (PNNL FWP 76274). NW's research is also supported by PNNL’s
Quantum Algorithms and Architecture for Domain Science (QuAADS) Laboratory Directed Research and
Development (LDRD) Initiative. The Pacific Northwest National Laboratory is operated by Battelle for
the U.S. Department of Energy under Contract DE-AC05-76RL01830.

\section{Code Availability}
The code used to generate the data for this paper is available upon request.

% \bibliography{paperpile,extracites}%

\appendix

\clearpage
\onecolumngrid
\setcounter{secnumdepth}{2}

\section{Querying a classical function in superposition with quantum read-only memory (QROM)}
\label{app:background_and_qrom}

Consider an arbitrary boolean function,
\begin{equation}
	f: \left\{ 0,1 \right\}^n \rightarrow \left\{ 0, 1 \right\}^m.
\end{equation}
We often want to construct an oracle that allows us coherent access to such a function,
\begin{equation}
	O_f: \ket{x}\ket{\alpha} \rightarrow \ket{x}\ket{\alpha \oplus f(x)}.
	\label{eq:O_f:def_qroam_app}
\end{equation}
In order to establish a framework for discussing asymptotic costs, we first begin with a short discussion of cost models for quantum algorithms in \Cref{app:cost_models}.
In \Cref{app:qrom_background}, we summarize prior work on instantiating such oracles~\cite{Babbush2018-tb,Childs2018-fq,Low2018-uu,Berry2019-qo}.
Throughout the paper, we use the phrase ``quantum read-only memory,'' or ``QROM'' to refer to the whole family of approaches to implementing such oracles, although specific variants often have other more technical names.
Then, in \Cref{app:modified_clean_qroam}, we introduce a variant of QROM that will be particularly useful in constructing our mass production protocol.

\subsection{A cost model for fault-tolerant algorithms}
\label{app:cost_models}

In order to discuss the asymptotic costs of various circuits, we need to specify a cost model.
One possible strategy is to count the number of arbitrary one- and two-qubit gates, or the number of gates in some particular universal gate set.
Motivated by the belief that non-Clifford gates will be dramatically more resource-intensive to implement using practical approaches to quantum error correction~\cite{Bravyi2005-vi}, many works on fault-tolerant quantum algorithms have focused exclusively on quantifying the number of non-Clifford, T or Toffoli gates~\cite{reiher2017elucidating,haner2016factoring,sanders2020compilation}.
However, continued progress in recent years suggests that this belief may be outdated~\cite{Bravyi2012-mg, reiher2017elucidating, O-Gorman2017-ul, Litinski2019-ek, Gidney2019-mv, Craig2024-kv}.
In particular, \citen{Craig2024-kv} argues that we may be able to ``cultivate'' magic states with sufficiently low error rates that we can implement the non-Clifford T gate in the surface code with only slightly more resources than a standard CNOT gate, although further work is needed to establish that this is true across a range of device parameters.

In order to model a range of plausible scenarios, we adopt the following cost model for a circuit \(C\) expressed in a Clifford + T gate set:
\begin{equation}
	\textsc{Cost} \left( C \right) = \textsc{Clifford}\left( C \right) + \Xi \textsc{T}(C),
	\label{eq:cost_model_summary}
\end{equation}
where \(\textsc{Clifford}\left( C \right)\) denotes the number of one- and two-qubit Clifford gates in \(C\), \(\textsc{T}\left( C \right)\) denotes the number of T gates, and \(\Xi\) is a constant.
Taking \(\Xi = 1\) recovers the model where we simply count the number of one- and two-qubit gates in a Clifford + T gate set, and taking \(\Xi \rightarrow \infty\) recovers the model where we discount the cost of Clifford gates entirely.
By varying \(\Xi\), we can explore regimes in between these extremes.
While initial estimates for \(\Xi\) were in the thousands~\cite{Bravyi2005-vi,Litinski2019-ek}, more modern approaches to magic state distillation have provided values in the tens or hundreds~\cite{Litinski2019-ek,Gidney2019-mv}, and \citen{Craig2024-kv} argues for a value of \(\Xi\) close to one.

Many of the constructions in this work, including the various implementations of QROM that we discuss in the next section, are naturally expressed using the non-Clifford Toffoli gate rather than T gates.
However, for simplicity, in our constant factor cost analyses we ultimately report the cost in terms of T gates by using a standard decomposition for Toffoli gates into four T gates plus several one- and two-qubit Clifford gates.

While our cost model is more nuanced than merely counting the total number of gates or the total number of non-Clifford gates, it still only provides a rough estimate for the true cost.
For example, some Clifford gates may be more difficult to implement than others, and the costs may vary based on the particular choice of quantum error correcting code or the underlying hardware platform~\cite{Fowler2018-he, Evered2023-bz, Xu2024-zd}.
Despite these subtleties, we expect that our cost model will provide useful estimates, particularly since we are focused on the relative costs of different approaches.

\subsection{Prior work on implementations of QROM}
\label{app:qrom_background}

There is a significant body of work on instantiating data loading oracles of the form given in \Cref{eq:O_f:def_qroam_app}.
The quantum read-only memory of \citen{Babbush2018-tb} provides a straightforward implementation whose gate complexity (in terms of the number of one- and two-qubit Clifford gates and T/Toffoli gates) is \(\Theta \left( 2^n m \right)\).
The SelectSwap QROM of \citen{Low2018-uu} demonstrated that a similar gate complexity could be obtained with a quadratically smaller number of non-Clifford gates, at the expense of requiring additional space.
\citen{Berry2019-qo}, which refers to this construction as advanced quantum read-only memory (QROAM), developed additional tools to efficiently uncompute QROAM queries in certain contexts.
Our constructions make use of the tools developed in these three works, but we note other improvements and variations have been proposed more recently~\cite{Zhu2024-pe}.

Our summary will mostly follow the notation and language of \citen{Berry2019-qo}.
Let \(N = 2^n\) and let \(\lambda\) be a power of two such that \(1 < \lambda < N\).
The assumption that \(N\) is a power of two can be relaxed to cover the setting where we only care about the action of \(f\) on the first \(N\) non-negative integers (with \(n = \ceil{\log N}\)).
Since this only introduces a constant factor difference in the cost, we only consider the simpler case where \(N\) is a power of \(2\) for simplicity.

Using the techniques of Berry \etal, we can implement \(O_f\) using \(N \lambda^{-1} + m(\lambda-1)\) Toffoli gates and \((\lambda-1)m + \log_2 \left( N \lambda^{-1}\right)\) clean ancilla (ancilla initialized in the \(\ket{0}\) state), together with \(\Theta\left(N m\right)\) one- and two-qubit Clifford gates (for a generic \(f\)).
More specifically, the number of one- and two-qubit Clifford gates is bounded by \(N m K\) for some small constant \(K\).
Furthermore, the spacetime volume (the product of the depth of the circuit and the number of qubits it requires) scales as \(\bigot{2^n m}\), regardless of \(\lambda\).

Technically their construction assumes that the output register is in the zero state, but the general case can be treated with only a small increase in cost.
Potentially more seriously, their approach leads either i) an entangled junk register that must be uncomputed, or ii) random (but known at execution-time)  phase errors that must be corrected.
The cost of these corrections is often neglected because, when the output of the QROAM is used and then immediately uncomputed, the corrections can be applied at no additional cost during the uncomputation step.
Since this will not necessarily be the case in our application of QROAM, we introduce a modified variant of the clean ancilla assisted QROAM in \Cref{fig:clean_ancilla_QROAM_without_uncomputation} that performs the desired operation without requiring any additional corrections and without any constraints on \(\alpha\).

An alternative construction makes use of dirty ancilla qubits (ancilla qubits borrowed and returned in an arbitrary state) rather than clean ancilla~\cite{Low2018-uu,Berry2019-qo}.
Implementing QROAM this way uses \(2 N \lambda^{-1} +4 m(\lambda-1)\) Toffoli gates, \(m(\lambda-1)\) dirty ancilla, and \(\log_2 \left( N \lambda^{-1}\right)\) clean ancilla, and \(\Theta\left(N m\right)\) one- and two-qubit Clifford gates.
There is no need to perform additional corrections when using this approach.

It is frequently critical to uncompute such a data lookup oracle.
Berry \etal{} showed how this uncomputation can be done particularly efficiently using a measurement-based strategy.
This approach uses \( N \lambda^{-1} + \lambda\) Toffoli gates and \(\lambda + \log_2 \left( N \lambda^{-1}\right)\) clean ancilla or \(2 N \lambda^{-1} + 4\lambda\) Toffoli gates, \(\lambda-1\) dirty ancilla, and \(\log_2 \left( N \lambda^{-1} \right) + 1\) clean ancilla~\cite{Berry2019-qo}.
The number of one- and two-qubit Clifford gates for the uncomputation step scales as \(\Theta \left(N\right)\) (provided that we don't count the \(m\) \(X\)-basis measurements as gates).
This measurement-based uncomputation strategy therefore avoids the scaling with \(N m\) that might be naively expected.

\subsection{Modified clean ancilla-assisted QROAM that uncomputes the junk register}
\label{app:modified_clean_qroam}

QROAM enables us to minimize the number of non-Clifford Toffoli gates required for coherently loading classical data by using additional space.
As we discussed in \Cref{app:qrom_background}, we can implement a data lookup oracle \(O_f\) for a function \(f: \left\{ 0, 1 \right\}^n \rightarrow \left\{ 0, 1 \right\}^m\) following the methods of \citen{Berry2019-qo} (which build on \citen{Low2018-uu}).
We can either do so using \(\lceil N/\lambda\rceil + m(\lambda-1)\) Toffoli gates and \((\lambda-1)m + \lceil \log N/\lambda \rceil\) clean ancilla (ancilla initialized in the \(\ket{0}\) state), or using \(2\lceil N/\lambda \rceil +4 m(\lambda-1)\) Toffoli gates, \((\lambda-1)m\) dirty ancilla (ancilla borrowed in an arbitrary state), and \(\lceil \log N/\lambda \rceil\) clean ancilla.
In both cases, we define \(N = 2^n\) to parallel the notation of \citen{Berry2019-qo}.
Provided we have the clean ancilla available, using them results in the lowest Toffoli cost.
Therefore, in this work we only consider the use of the clean ancilla variants of QROM, although we expect that our results could also be adapted to the dirty ancilla case.

However, the costs quoted for the clean ancilla version of QROM above do not exactly correspond to the costs for implementing the desired
\begin{equation}
	O_f: \ket{x} \ket{\alpha} \rightarrow \ket{x}\ket{\alpha \oplus f(x)}.
\end{equation}
Rather than implementing the above operation, the clean-ancilla version of QROAM actually implements a unitary operation \(\tilde{O}_f\) that acts as below,
\begin{equation}
	\tilde{O}_f: \ket{x} \ket{0}^{\otimes \lambda} \rightarrow \ket{x} \ket{f(x)} \ket{j(x)},
\end{equation}
where \(\ket{j(x)}\) is some ``junk'' computational basis state on \(m (\lambda - 1)\) qubits.\footnote{The dirty ancilla variants of QROAM do not have this issue since they explicitly restore the state of the borrowed qubits.}
In practice, QROAM is often employed in a context where the output values will be used and then immediately uncomputed.
If this is the case, then the uncomputation of the \(\ket{j(x)}\) register can be folded into the same step that uncomputes the \(\ket{f(x)}\) register at no additional cost, as described in \citen{Berry2019-qo}.

However, in our case, we do not immediately use the output and then uncompute it, so it is not clear that the same kind of optimization is available to us.
Naively, we might expect to have to pay an additional cost to uncompute the junk data before proceeding.
Specifically, we could perform the usual data loading step, then \(m\) CNOT gates to XOR \(f(x)\) in the register containing \(\alpha\), and then uncompute the registers containing \(f(x)\) and \(j(x)\).
However, the uncomputation itself requires a data loading, which would incur an additional cost of \(\lceil N \lambda^{-1} \rceil + \lambda\) Toffoli gates (using \(\lambda + \lceil \log N\lambda^{-1} \rceil\) clean ancilla) and \(\mathcal{O}\left( N \right)\) Clifford gates.
In the worst case, this would nearly double the cost of data loading and would make it more difficult to benefit from mass production at all.

We therefore propose a modification to the standard clean ancilla QROAM that guarantees that the additional ancilla will be returned to the zero state.
This modification requires zero additional non-Clifford gates, and only \(\Theta(m\lambda)\) additional one- and two-qubit Clifford gates.
We show this alternative strategy in \Cref{fig:clean_ancilla_QROAM_without_uncomputation}.

This QROAM is a modification of the construction from \citen{Berry2019-qo}, and we closely follow their notation for the remainder of this section.
In our circuit diagram, \(S\) denotes a series of controlled swap gates that permutes the \(\lambda\) output registers so that the register that initially contains \(\alpha\) ends up in the \(l\)th position.
\(T\) denotes the usual QROM data loading operation that outputs \(m \lambda\) bits so that the \(l\)th register contains the data for \(f(h \doubleplus l)\).
Implementing \(T\) requires \(\left \lceil 2^n / \lambda \right \rceil\) Toffoli gates and \(\mathcal{O}\left( 2^n m \right)\) Clifford gates.
Our construction takes advantage of the fact that \(T\) will act as the identity on any output registers in the \(\ket{+}^{\otimes m}\) state.
Initializing the additional clean qubits in the \(\ket{+}\) ensures that only the \(l\)th register is modified by the action of \(T\).

Berry \textit{et al.}'s construction for clean ancilla QROAM consists of exactly one controlled \(T\) gate and one controlled \(S\) gate.
Our circuit differs from theirs by the addition of the Hadamard gates as shown in the figure, the \(X\)-basis initialization of the clean ancilla qubits, and the controlled \(S^{\dagger}\) operation.
We could implement \(S^{\dagger}\) as a product of controlled swaps.
In our particular circuit, we can take advantage of the fact that we know the outputs for all of the registers but the first one will be zero to simplify the circuit and entirely remove the need for non-Clifford gates, as in Figure 10 of \citen{Berry2019-qo}.
The circuit implementation of \(S\) contains exactly \(m \left( \lambda - 1 \right)\) controlled swap gates (equivalently, \(2 m \left( \lambda - 1 \right)\) CNOT gates and \(m \left( \lambda - 1 \right)\) Toffoli gates).
The measurement-based implementation of \(S^{\dagger}\) contains \(m \left( \lambda - 1 \right)\) each of CNOT gates, Hadamard gates, single-qubit measurements, and classically-controlled CZ gates.\footnote{One could further simplify by combining the Hadamard gates of \Cref{fig:clean_ancilla_QROAM_without_uncomputation} with the implementation of \(S^{\dagger}\) from \citen{Berry2019-qo} but we neglect this optimization for now.}

Overall then, we can say that our modified advanced QROM circuit requires \(\lceil N/\lambda\rceil + m(\lambda-1)\) Toffoli gates and \((\lambda-1)m + \lceil \log N/\lambda \rceil\) clean ancilla, together with a number of Clifford gates bounded by \(K N m\) for some universal constant \(K\).
Furthermore, the number of Clifford gates required scales the same as the standard clean-ancilla advanced QROM, up to leading order.

\begin{figure}
	\centering
	\begin{quantikz}
		\lstick{$\ket{h}$}                & \qw         & \ctrl{2}    & \qw                  & \qw                 & \qw                  & \rstick{$\ket{h}$}
		\\
		\lstick{$\ket{l}$}                & \ctrl{1}    & \qw         & \qw                  & \ctrl{1}            & \qw                  & \rstick{$\ket{l}$}
		\\
		\lstick{$\ket{\alpha}$}           & \gate[4]{S} & \gate[4]{T} & \gate{H^{\otimes m}} & \gate[4]{S^\dagger} & \gate{H^{\otimes m}} & \rstick{$\ket{\alpha \oplus f(h \doubleplus l)}$}
		\\
		\lstick{$\ket{+}$}                & \qw         & \qw         & \gate{H^{\otimes m}} & \qw                 & \qw                  & \rstick{$\ket{0}$}
		\\
		{\ \ \vdots \ \ } \setwiretype{n} &             &             & {\ \ \vdots \ \ }    &                     &                      & {\ \ \vdots \ \ }
		\\
		\lstick{$\ket{+}$}                & \qw         & \qw         & \gate{H^{\otimes m}} & \qw                 & \qw                  & \rstick{$\ket{0}$}
	\end{quantikz}
	\caption{An alternative to the clean ancilla QROAM of \citen{Low2018-uu} and \citen{Berry2019-qo} that does not require uncomputation on the ancilla.
		Instead, it directly implements the desired data loading operation, XORing the data into the first output register and returning the ancilla to the \(\ket{0}\) state.
		Here \(T\) and \(S\) refer to a QROM data lookup and a series of controlled swap gates, as in \citen{Berry2019-qo}.
		\(S^{\dagger}\) can be implemented without using any non-Clifford gates using the measurement-based uncomputation approach shown in Figure 10 of \citen{Berry2019-qo}.}
	\label{fig:clean_ancilla_QROAM_without_uncomputation}
\end{figure}
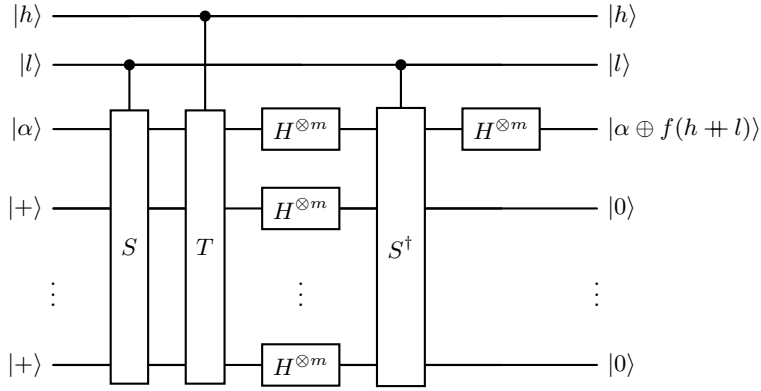

\section{Mass production for classical data loading}
\label{app:mass_production_protocol}

We detail our mass production protocol for classical data loading in this section.
Following Uhlig and Kretschmer, we construct our mass production protocol recursively~\cite{Ulig1974-xh,Uhlig1992-tg,Kretschmer2022-uy}.
We first present the two-query protocol for implementing \(O_f^{\otimes 2}\).
Then we explain how to construct the protocol for \(O_f^{\otimes 2r}\) given a protocol for \(O_f ^{\otimes r}\).
We comment on the asymptotic costs as we present the construction, culminating in \Cref{app:scaling_thm_proof}, where we prove \Cref{thm:mass_produced_qrom} by induction.
For an analysis of the constant factor costs, we refer the reader to the results presented in the main text and to our implementation of the protocol in the Qualtran software package~\cite{Harrigan2024-rj}.

Before we begin, it is helpful to introduce some convenient notation.
We use the symbol \(\doubleplus\) to denote the concatenation of two bitstrings, e.g., \(11 \doubleplus 00 = 1100\).
For integers \(a\) and \(b\) with \(a < b\), we use the notation \(\left[ a.
	.b \right]\) to denote the sequence \(a, a+1, a+2, \cdots b\).
We frequently conflate non-negative integers and the bitstrings that encode them in a standard unsigned big-endian format.

Our protocol implements \(r\) queries to an oracle \(O_f\) for an arbitrary function \(f: \left\{ 0, 1 \right\}^n \rightarrow \left\{ 0, 1 \right\}^m\), where \(r = 2^t, t \in \mathbb{Z}^+\).
We construct our protocol with several free parameters that we set later in order to optimize the costs.
We take \(\lambda\), a parameter that governs a tradeoff involved in the QROAM data lookups, to be a power of \(2\) such that \(\lambda \leq 2^{n/2} m^{-1/2}\).
For the purposes of statements about asymptotic scaling, we consider \(m\) a constant and \(\lambda\) to be a function of \(n\).
Similarly, viewing \(t\) as a function of \(n\), we require that \(t = o(\frac{n - 2 \log_2 \lambda}{\log_2 n})\) (although we will later relax this requirement for \Cref{thm:max_scaling_prop}).
We also introduce a parameter \(k \in \mathbb{Z}^+\), subject to the requirement that \(n > k t\) and \(k = \mathcal{O}\left( \log n \right)\).
More specifically, we will set \(k = \ceil{\log_2 n}\) when we prove \Cref{thm:mass_produced_qrom} in \Cref{app:scaling_thm_proof} and \(k=1\) when we prove \Cref{thm:max_scaling_prop},  although we keep it as a free parameter in the exposition since we will later optimize it to minimize the costs in our constant factor analysis.

Ultimately, we construct the circuit that implements \(O_f^{\otimes r}\) (which we denote by \(C_{f, n, m, \lambda, k, t}\)) by induction on \(t\).
Since we require that \(t = o(\frac{n - 2 \log_2 \lambda}{\log_2 n})\) and we are interested in understanding the asymptotic scaling for non-zero values of \(t\), we impose an additional technical condition on \(\lambda\).
Specifically, we demand that
\begin{equation}
	\lim_{n \rightarrow \infty} \frac{n - 2 \log_2 \lambda}{\log_2 n} \geq c
\end{equation}
for all \(c > 0\).
Rearranging this expression, we find that we need
\begin{equation}
	\lim_{n \rightarrow \infty} \lambda \leq 2^{-c} \frac{2^{n/2}}{\sqrt{n}}.
\end{equation}
Since \(2^{-c}\) is an arbitrary positive number, this is equivalent to the demand that
\begin{equation}
	\lambda = o\left(\frac{2^{n/2}}{\sqrt{n}}\right).
	\label{eq:lambda_cond}
\end{equation}
We note that this is a stricter demand (asymptotically) than the already-stated requirement that \(\lambda \leq 2^{n/2} m^{-1/2}\).

\subsection{The base case}
\label{app:base_case}

We begin by constructing a circuit \(C_{f, n, m, \lambda, k, 1}\), which we abbreviate as \(C\) throughout this section, that evaluates \(f\) on two inputs.
At a high level, we try to parallel \citen{Kretschmer2022-uy} and \citen{Uhlig1992-tg} in the construction and notation of the mass production theorems and \citen{Berry2019-qo} in the construction and notation of the QROAM data lookups.
The circuit \(C\) should act unitarily on \(2n + 2m\) qubits, applying \(O_f\) to a pair of input and output registers, i.e.,
\begin{equation}
	C: \ket{x}\ket{\alpha}\ket{y}\ket{\beta} \rightarrow \ket{x} \ket{\alpha \oplus f(x)} \ket{y} \ket{\beta \oplus f(y)} .
	\label{eq:inductive_step_desired}
\end{equation}

For simplicity, we make the assumption for now that \(x \leq y\).
If we can implement \(C\) in this case, then we can address the more general situation by first checking \(x \leq y\) with a comparator and swapping the input and output registers when \(x > y\).
After implementing \(C\), we can then swap the registers back and uncompute the comparator.

Following \citen{Uhlig1992-tg}, we define two families of functions, \(\left\{ f_\ell \right\}\) and \(\left\{ g_\ell \right\}\).
Let \(k \in \left\{ 1,\ldots,
	n \right\}\) be a constant that we will set later.
Let $\ell$ be a $k$-bit integer, \(\ell \in \{0,\ldots,2^k - 1\}\), and \(f_\ell: \left\{ 0,1 \right\}^{n-k} \rightarrow \left\{ 0, 1 \right\}^m\) denote the function \(f\) with the first \(k\) bits fixed to \(\ell\), i.e.,
\begin{equation}
	f_{\ell}\left( z \right) = f\left( \ell \doubleplus z \right).
\end{equation}
Then, for \(\ell \in \left[ 0\ldots 2^k \right]\), let
\begin{align}
	g_\ell = \begin{cases} f_0 & \ell=0,\\
    f_{\ell - 1} \oplus f_\ell& 1 \leq \ell \leq 2^k - 1,\\
    f_{2^k-1} & \ell=2^k 
    \end{cases}
	\label{eq:g_def_appendix}
\end{align}
Note that
\begin{equation}
	f_\ell(z) = \bigoplus_{j=0}^{\ell} g_j(z) = \bigoplus_{j=\ell+1}^{2^k}g_j(z).
	\label{eq:f_ell_g_ell_trick_app}
\end{equation}
We will show that the key insight that enables mass production is that we can take advantage of \Cref{eq:f_ell_g_ell_trick_app} to evaluate \(f\) on two different inputs while only evaluating each \(g_\ell\) at most once.
Because the \(g_{\ell}\) are ``simpler'' by a factor of \(2^k\), then the overall cost of implementing \(2^k + 1\) of them is comparable to the cost of implementing a single call to \(f\).

We explain the construction for \(C\) by describing its action on a computational basis state \(\ket{x}\ket{\alpha}\ket{y}\ket{\beta}\).
We treat the first \(k\) bits of \(x\) and \(y\) differently from the remaining \(n-k\), so we adopt the notation that
\begin{align}
	x & = x_L \doubleplus x_R
	\\
	y & = y_L \doubleplus y_R,
\end{align}
where \(x_L\) (\(y_L\)) denotes the first \(k\) bits of \(x\) (\(y\)) and \(x_R\) (\(y_R\)) denotes the remaining bits.
We implement \(C\) by alternating the application of two key primitives, \(G_\ell\) and \(A_{\ell}\), as illustrated in \Cref{fig:C_construction}.

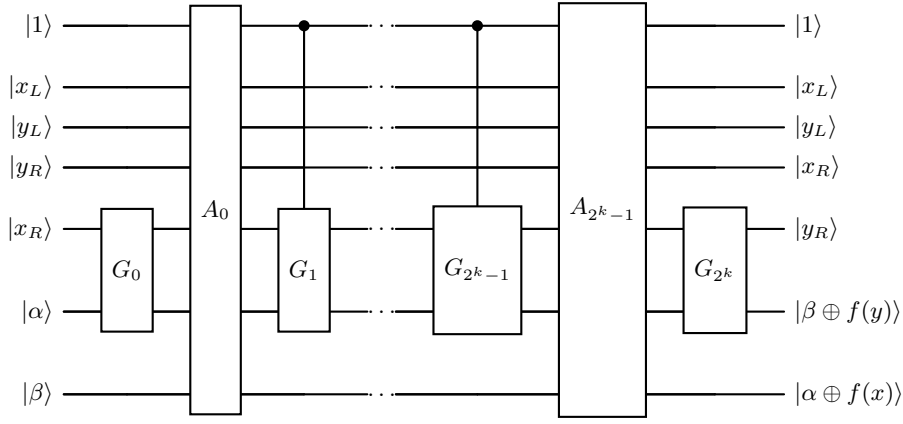
\begin{figure}
	\centering
	\begin{quantikz}[wire types={q,q,q,q,q,q,q}, scale=2]
		\lstick{$\ket{1}$}      & \qw           & \gate[7]{A_0} & \ctrl{4}      & \ldots & \ctrl{4}              & \gate[7]{A_{2^k - 1}} & \qw               & \rstick{$\ket{1}$}
		\\
		\lstick{$\ket{x_L}$}    & \qw           & \qw           & \qw           & \ldots & \qw                   & \qw                   & \qw               & \rstick{$\ket{x_L}$}
		\\
		\lstick{$\ket{y_L}$}    & \qw           & \qw           & \qw           & \ldots & \qw                   & \qw                   & \qw               & \rstick{$\ket{y_L}$}
		\\
		\lstick{$\ket{y_R}$}    & \qw           & \qw           & \qw           & \ldots & \qw                   & \qw                   & \qw               & \rstick{$\ket{x_R}$}
		\\
		\lstick{$\ket{x_R}$}    & \gate[2]{G_0} & \qw           & \gate[2]{G_1} & \ldots & \gate[2]{G_{2^k - 1}} & \qw                   & \gate[2]{G_{2^k}} & \rstick{$\ket{y_R}$}
		\\
		\lstick{$\ket{\alpha}$} & \qw           & \qw           & \qw           & \ldots & \qw                   & \qw                   & \qw               & \rstick{$\ket{\beta \oplus f(y)}$}
		\\
		\lstick{$\ket{\beta}$}  & \qw           & \qw           & \qw           & \ldots & \qw                   & \qw                   & \qw               & \rstick{$\ket{\alpha \oplus f(x)}$}
	\end{quantikz}
	\caption{A quantum circuit diagram for \(C\).
		For clarity, we omit some initial and final swaps required to make the registers correspond exactly to \Cref{eq:inductive_step_desired}.
		\(G_{\ell}\) is a data lookup that implements the function \(g_{\ell}\).
		The \(A_{\ell}\), which are defined in the text and illustrated in \Cref{fig:advance_ell}, handle the ``control flow,'' routing the inputs and outputs appropriately for each call to \(G_\ell\).
	}
	\label{fig:C_construction}
\end{figure}

The first type of primitive, \(G_\ell\), is a (usually controlled) call to a data lookup oracle for the corresponding function \(g_{\ell}\),
\begin{equation}
	G_\ell: \ket{c}\ket{z}\ket{\gamma} \rightarrow \ket{c}
	\begin{cases*}
		\ket{z}\ket{\gamma}                  & \text{ if } c=0
		\\
		\ket{z}\ket{\gamma \oplus g_\ell(z)} & \text{ if } c=1.
	\end{cases*}
\end{equation}
We implement this data lookup using the techniques described in \Cref{app:modified_clean_qroam}, which is a slight variation of the usual QROAM construction from \citen{Berry2019-qo}.
As shown in \Cref{fig:C_construction}, we apply the \(G_{\ell}\) sequentially for \(\ell \in \left[ 0,\ldots, 2^k \right]\).

For each \(G_{\ell}\), we will be in one of three cases,
\begin{enumerate}
	\item $\ell \leq x_L,$
	\item $x_L < \ell \leq y_L,$ or
	\item $y_L < \ell.$
\end{enumerate}
In case \(1\), we would like to use \(G_{\ell}\) to evaluate \(g_\ell (x_R)\) and XOR the output into the \(\alpha\) register.
In case \(2\), we should effectively not apply \(G_{\ell}\) at all.
In case \(3\), we want to evaluate \(g_{\ell}(y_R)\) and XOR the output into the \(\beta\) register.

In order to implement the desired operations, the control qubit for the data lookups should be in the \(\ket{1}\) state whenever \(G_\ell\) is called and we are in case \(1\) or case \(3\), and in the \(\ket{0}\) state when we are in case \(2\).
Furthermore, when we are in case \(1\), \(G_\ell\) should be called with the input register in the state \(\ket{x_R}\) and the output should be written to the register that is initially in the state \(\ket{\alpha}\).
When we are in case \(3\), the inputs and outputs should be swapped so that \(G_\ell\)'s input register is in the state \(\ket{y_R}\) and the output register is the one that initially contained the state \(\ket{\beta}\).

By definition we are in case \(1\) when \(\ell=0\), so we can simply apply \(G_0\) without a control qubit as in the first step of \Cref{fig:C_construction}.
After each call to \(G_\ell\) (except for the last), we can then apply \(A_\ell\) as defined in \Cref{fig:advance_ell} to flip the control qubit and swap the input and output registers as necessary.
Inspecting \Cref{fig:advance_ell}, we see that \(A_\ell\) will flip the control qubit for the \(G_\ell\) whenever \(\ell = x_L\), transitioning us from the settings necessary for case \(1\) to the settings necessary for case \(2\).
Likewise, when \(\ell = y_L\), the control qubit is flipped again and we swap the two input and two output registers, as required for case \(3\).
After applying \(G_{2^{k}}\), the only remaining task is to swap the input and output registers back, since they were swapped exactly once by the action of the \(A_{\ell}\)s.
Let us denote this final swapping operation by \(A_{2^k}\).

\begin{figure}
	\centering
	\begin{quantikz}[wire types={q,q,q,q,q,q,q,q}, classical gap=.05cm]
		\lstick{$\ket{0}$}        & \qw             & \targ{}         & \ctrl{1} & \ctrl{5} & \ctrl{7} & \targ{}         & \rstick{$\ket{0}$}
		\\
		\lstick{$\ket{\text{c}}$} & \targ{}         & \qw             & \targ{}  & \qw      & \qw      & \qw             & \qw
		\\
		\lstick{$\ket{x_L}$}      & \ctrl[open]{-1} & \qw             & \qw      & \qw      & \qw      & \qw             & \qw
		\\
		\lstick{$\ket{y_L}$}      & \qw             & \ctrl[open]{-3} & \qw      & \qw      & \qw      & \ctrl[open]{-3} & \qw
		\\
		\lstick{$\ket{y_R}$}      & \qw             & \qw             & \qw      & \targX{} & \qw      & \qw             & \qw
		\\
		\lstick{$\ket{x_R}$}      & \qw             & \qw             & \qw      & \targX{} & \qw      & \qw             & \qw
		\\
		\lstick{$\ket{\alpha}$}   & \qw             & \qw             & \qw      & \qw      & \targX{} & \qw             & \qw
		\\
		\lstick{$\ket{\beta}$}    & \qw             & \qw             & \qw      & \qw      & \targX{} & \qw             & \qw
	\end{quantikz}
	\caption{A quantum circuit diagram for a possible realization of \(A_\ell\).
		In this figure, the controlled operations where the control qubit is marked with an open circle indicate operations that are controlled based on the condition that the control register encodes the value \(\ell\).
		Our Qualtran implementation of this primitive uses machinery built into the library to compile the multi-qubit controls, which introduces \(\mathcal{O}\left( k \right)\) additional ancilla qubits not shown here.
	}
	\label{fig:advance_ell}
\end{figure}
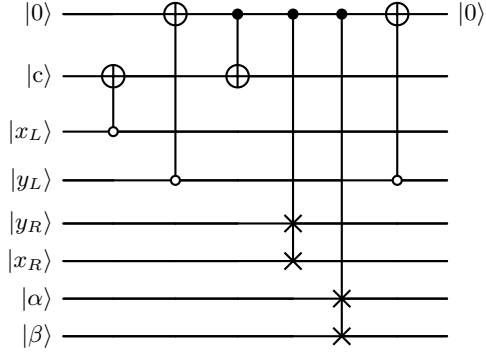

Letting \(S\) denote the additional swap operations necessary to convert \(\ket{x}\ket{\alpha}\ket{y}\ket{\beta}\) to \(\ket{x_L} \ket{y_L} \ket{y_R} \ket{x_R} \ket{\alpha} \ket{\beta}\) (which are really an artifact of our presentation anyway), we have
\begin{align}
	S^\dagger A_{2^k} G_{2^k} \cdots A_1 G_1 A_0 G_0 S \ket{1}\ket{x}\ket{\alpha}\ket{y}\ket{\beta} = \ket{1}\ket{x}\ket{\alpha'}\ket{y}\ket{\beta'},
\end{align}
where, taking advantage of \Cref{eq:f_ell_g_ell_trick_app}, we have
\begin{align}
	\alpha' & =
	\alpha \bigoplus_{j=0}^{x_L} g_j(x_R) = \alpha \oplus f_{x_L}\left( x_R \right) = \alpha \oplus f(x_L \doubleplus x_R) = \alpha \oplus f(x)
	\\
	\beta'  & =
	\beta \bigoplus_{j=y_L}^{2^k} g_j(y_R) = \beta \oplus f_{y_L}\left( y_R \right) = \beta \oplus f(y_L \doubleplus y_R) = \beta \oplus f(y)
	.
\end{align}
Therefore, the circuit presented in \Cref{fig:C_construction} implements \Cref{eq:inductive_step_desired} as desired.

\subsection{Cost analysis for the two copy protocol}
\label{app:base_case_cost_analysis}

The two-copy protocol for mass production forms the foundation for the \(r\)-copy protocol, so it is worth carefully analyzing its costs.
Here we address the asymptotic scaling.
There are contributions to the cost from four sources, which we order by their importance:
\begin{enumerate}
	\item The QROM data lookups \(G_{\ell}\).
	\item The control flow circuits \(A_{\ell}\).
	\item The reduction to the case where \(x \leq y\).
	\item A few additional swap gates.
\end{enumerate}

The dominant cost comes from the \(2^{k} + 1\) QROM data lookups.
As we discussed in \Cref{app:qrom_background}, while it is possible to reduce the number of non-Clifford gates required for a data lookup to scale sublinearly in \(2^n m\) (using techniques known as SelectSwap QROM or advanced QROM), the number of Clifford gates can not be reduced.
Therefore, in the cost model we are considering, the asymptotic scaling is driven by the number of Clifford gates.
Nevertheless, we find it illuminating to discuss the Clifford and non-Clifford costs separately here.

Using the modified QROM construction presented in \Cref{app:modified_clean_qroam}, we can implement an oracle \(O_f\) with \(n\) input qubits and \(m\) output qubits using
\begin{equation}
	\textsc{Clifford}(O_f) \leq \left(K + o(1)\right) 2^{n} m
	\label{eq:simple_qroam_clifford_cost}
\end{equation}
one- and two-qubit Clifford gates for some constant \(K\), and
\begin{equation}
	\textsc{Toffoli}(O_f) = 2^{n}\lambda^{-1} + \lambda m - m
	\label{eq:simple_qroam_toffoli_cost}
\end{equation}
Toffoli gates.
Therefore,
\begin{equation}
	\textsc{Cost}(O_f) \leq \left( K + o(1) \right) 2^{n} m + \Xi \left(2^{n}\lambda^{-1} + \lambda m \right).
	\label{eq:simple_qroam_cost}
\end{equation}
In our construction, we implement \(2^{k} + 1\) smaller data lookups, with \(n-k\) input bits and \(m\) output bits.
The overall cost of implementing all of the \(G_\ell\) is given by
\begin{equation}
	\textsc{Cost}(\left\{ G_\ell \right\}) \leq \left(1 + 2^{k}\right) \left(\left(K + o(1)\right) 2^{n-k} m  + \Xi \left(2^{n-k}\lambda^{-1} + \lambda m \right)\right).
	\label{eq:G_ell_asymptotic_cost_two_copies}
\end{equation}

We claim that the other three costs are asymptotically smaller (as we increase \(n\)), and so they can be neglected.
We address them point by point.
We implement \(2^{k}\) of the \(A_{\ell}\) circuits as described in \Cref{fig:advance_ell}.
Each of these involve some equality checks on \(k\) qubits and controlled swaps on \(n-k\) qubits and \(m\) qubits.
The number of Clifford and non-Clifford gates required by this component of the algorithm therefore both scale as \(\mathcal{O} \left(2^k \left(n + m\right) \right)\).
Using our assumption that \(k = \mathcal{O}\left( \log n \right)\), then this scaling is \(\mathcal{O}\left( \text{Poly}\left( n \right) \right)\), which is dominated by \(2^{n-k} m\).
The reduction to the case where \(x \leq y\) involves inequality checks on registers of size \(n\) and controlled swaps on registers of size \(n\) and \(m\).
The Clifford and non-Clifford costs for this aspect of the algorithm both scale as \(\mathcal{O} \left( n + m \right)\), which is also dominated by \(2^{n-k} m\).
The number of additional swap gates scales as \(\mathcal{O}(n + m)\) and is also negligible.
Therefore, we have 
\begin{equation}
	\textsc{Clifford}\left(C_{f, n, m, \lambda, k, 1}\right) \leq \left(1 + 2^k\right)\left(K + o(1)\right) 2^{n-k} m
	\label{eq:C_asymptotic_clifford_scaling}
\end{equation}
and
\begin{equation}
	\textsc{Toffoli}\left(C_{f, n, m, \lambda, k, 1}\right) = \left(1 + 2^k\right)\left(1 + o(1)\right)\left(2^{n-k}\lambda^{-1} + \lambda m - m\right).\label{eq:C_asymptotic_Toff_scaling}
\end{equation}

Asymptotically then, it is simple to compare the cost of the entire circuit for the two-copy protocol (which we abbreviate as \(C\)), with two separate QROAM queries (\(O_f^{\otimes 2}\)).
First, let us consider the limit where \(\lambda(n)\) is chosen such that the cost is dominated by the number of Clifford gates.
If we assume that the bounds are saturated in both \Cref{eq:simple_qroam_cost} and \Cref{eq:C_asymptotic_clifford_scaling}, we find that
\begin{equation}
	\frac{\textsc{Cost}\left( O_f^{\otimes 2} \right)}{\textsc{Cost}\left( C \right)} \sim \frac{2 \left( 2^n m \right)}{\left( 1 + 2^k \right)\left( 2^{n-k}m \right)} \sim\frac{2}{1 + 2^{-k}}.
\end{equation}
Therefore, we save a factor of \(2\) asymptotically.
Furthermore, we approach this level of savings quickly as \(k\) increases.

It is interesting to briefly consider what would happen if we focused on a cost model that ignores the cost of Clifford gates, as previous work on fault-tolerant quantum algorithms often has.
When we choose \(\lambda\) to minimize the cost of QROAM in this model,\footnote{We relax the requirement that \(\lambda\) is a power of two for the purposes of simplifying the analysis here.
} we find that implementing \(O_f\) for an \(n\) to \(m\) bit function \(f\) requires a number of Toffoli gates that scales as
\begin{equation}
	\textsc{Toffoli}\left( O_f \right) \sim 2 \left(2^{n/2} m^{1/2}\right). 
\end{equation}
We claim that the number of non-Clifford gates required to implement \(C\) is dominated by the cost of the \(G_{\ell}\) in this cost model as well.
Taking this claim as true, the \(2^{k} + 1\) calls to \(G_{\ell}\) therefore require a number of Toffoli gates that scales as \(2 \left(2^{k} + 1\right) 2^{\left(n-k\right)/2} m^{1/2}\).
We can therefore calculate the cost ratio to find that
\begin{equation}
	\frac{\textsc{Toffoli}\left( O_f^{\otimes 2} \right)}{\textsc{Toffoli}\left( C \right)} \sim
	\frac{4 \left( 2^{n/2} m^{1/2} \right)}{2\left( 1 + 2^k \right)\left( 2^{\left( n-k \right)/2}m^{1/2} \right)}
	\sim \frac{2^{k/2 + 1}}{2^{k} + 1},
\end{equation}
which is less than or equal to \(1\) for all positive \(k\) and asymptotically approaches \(2 n^{-1/2}\) when we set \(k = \ceil{\log_2 n}\).

In other words, it always requires fewer Toffoli gates to implement \(O^{\otimes 2}_f\) directly than to use mass production, and so there is no advantage to the two-copy mass production protocol in this cost model.
Since the advantage for larger values of \(r\) depends on the two-copy version, there would be no advantage in those cases either.
Fundamentally, this is because the strategy for mass production that we pursue relies on the smaller subproblems (the \(G_{\ell}\) in our case) becoming less costly at the same exponential rate that the number of such subproblems is increasing.
This is not the case when we quantify the complexity by the number of non-Clifford gates required.
It is an interesting open problem to determine whether or not it is possible to find an advantage in this cost model using a different construction for mass production.

\subsection{The inductive step and the proof of \Cref{thm:mass_produced_qrom}}
\label{app:scaling_thm_proof}

Now that we have discussed the two-copy protocol extensively, we can introduce the more general protocol.
Following Uhlig, we construct the protocol for mass-producing \(r=2^{t}\) copies recursively~\cite{Uhlig1992-tg}.
The key insight that enables this recursion is that the basic data loading task we are parallelizing is present in our construction as a subroutine.
Specifically, when we use mass production to implement \(O_f^{\otimes 2}\), we perform a series of calls to the data loading oracles \(G_{\ell}\).
Furthermore, as we saw, these calls to \(G_{\ell}\) are the dominant contribution to the overall cost.

Inspired by this observation, we can construct a protocol for mass-producing \(2 r\) queries that uses an \(r\)-query mass production protocol (for some \(r = 2^t\)) as a subroutine.
To proceed, we begin by imagining that we implement \(r\) copies of the two-query mass production protocol in parallel.
This would accomplish the desired task of implementing \(O_f^{\otimes 2 r}\), but it would not have the desired cost.
If we were to proceed this way, then we would find ourselves implementing \(r\) queries to each of the smaller data loading oracles \(G_{\ell}\) in parallel.
Instead of doing this, we simply implement the \(r\) parallel queries to each \(G_{\ell}\) using an \(r\)-query mass production protocol, as we illustrate in \Cref{fig:recursion_cartoon_big} for \(r=4\).

\begin{figure}
	\input{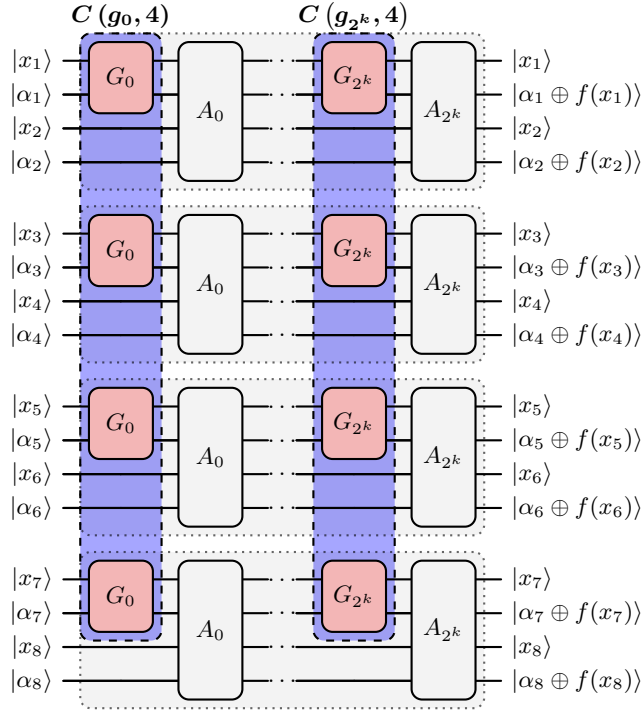}
	\caption{
		A diagram showing how the \(r=8\)-query mass production protocol is constructed recursively.}
	\label{fig:recursion_cartoon_big}
\end{figure}

Intuitively, if the \(r\)-query mass production scheme scales in the desired way (nearly independent of \(r\)), then the \(2r\)-query one will as well.
This is because, as we discussed, the cost of the two-query scheme is dominated by the cost of implementing the \(G_{\ell}\).
In order to prove \Cref{thm:mass_produced_qrom}, we will first prove the following lemma by induction:

\begin{restatable}[Quantum Mass Production for quantum read-only memory]{lemma}{qromlemma}
	Let \(f: \left\{ 0, 1 \right\}^n \rightarrow \left\{ 0, 1 \right\}^m\) be arbitrary, and let \(O_f\) be a unitary operator that queries this function in superposition, \(O_f: \ket{x}\ket{\alpha} \rightarrow \ket{x} \ket{\alpha \oplus f(x)}\).
	Let \(\lambda\) be a function of \(n\) that returns a power of \(2\) such that \(	\lambda = o\left(\frac{2^{n/2}}{\sqrt{n}}\right)\) and \(1 \leq \lambda \leq 2^{n/2} m^{-1/2}\), and let \(m\) be a constant function of $n$.
	Let \(k\) be a function of \(n\) that returns a positive integer such that \(k = \mathcal{O}\left( \log n \right)\).
	For any positive integer \(t\) such that \(kt < n\), there exists a quantum circuit \(C_{f, n, m, \lambda, k, t}\) with the following properties:

	\begin{enumerate}
		\item \(C_{f, n, m, \lambda, k, t}\) implements \(O_f^{\otimes r}\), where \(r=2^t\).
		\item \(C_{f, n, m, \lambda, k, t}\) is composed entirely of one- and two-qubit Clifford gates and Toffoli gates.
		
		\item The number of one- and two-qubit Clifford gates in \(C_{f, n, m, \lambda, k, t}\) is bounded by
		\begin{equation}
			\textsc{Clifford}\left( C_{f, n, m, \lambda, k, t} \right) \leq \left(1 + 2^k\right)^t \left( K + o(1) \right) \left(2^{n-tk} m + \bigo{n}\right)
			\label{eq:overall_clifford_cost_inductive_assumption}
		\end{equation}
		for some universal constant \(K\).
		\item The number of Toffoli gates in \(C_{f, n, m, \lambda, k, t}\) is
		\begin{equation}
			\textsc{Toffoli}\left( C_{f, n, m, \lambda, k, t} \right) = \left(1 + 2^{k}\right)^t \left( 1 + o(1) \right) \left(2^{n-tk} \lambda^{-1} + m \lambda + \bigo{n}\right),
			\label{eq:overall_T_cost_inductive_assumption}
		\end{equation}
	\end{enumerate}
	\label{lemma:mass_produced_qrom}
\end{restatable}

\begin{proof}
	
We begin by making the strong inductive assumption that the lemma holds for any \(t' \leq t\).
The base case in our inductive argument follows from the observations made in \Cref{eq:C_asymptotic_clifford_scaling} and \Cref{eq:C_asymptotic_Toff_scaling} of \Cref{app:base_case_cost_analysis}.
Taken together, they show that the hypothesis holds for $t=1$.
Using this assumption, we will show that our assumption implies that the lemma holds for \(t + 1\) under the condition that \(k\left( t+1 \right) < n\).
By induction, this will prove that the lemma is true for all \(t\) such that \(kt < n\).

Following the construction described above, we imagine implementing the two-copy mass production construction in parallel \(2^t\) times.
Note that this involves \(2^t\) parallel calls to each \(G_\ell\) and that our assumption that \(k (t + 1) < n\) implies that \(kt < n-k\).
Therefore our inductive assumption tells us that there exists a circuit \(C_{g_\ell, n-k, m, \lambda, k, t}\) that implements \(G_{\ell}^{\otimes 2^t}\) for each \(g_\ell\) and that each of these circuits requires a number of Clifford gates that scales as
\begin{equation}
	\textsc{Clifford}\left( C_{g_{\ell}, n-k, m, \lambda, k, t} \right) \leq \left(1 + 2^k\right)^t \left( K + o(1) \right) \left(2^{n-k-tk} m + \bigo{n}\right).
	\label{eq:sub_step_apply_inductive_clifford_cost}
\end{equation}
and a number of Toffoli gates that scales as
\begin{equation}
	\textsc{Toffoli}\left( C_{g_{\ell}, n-k, m, \lambda, k, t} \right) = \left(1 + 2^{k}\right)^t \left( 1 + o(1) \right) \left(2^{n-k-tk} \lambda^{-1} + m \lambda + \bigo{n}\right).
	\label{eq:sub_step_apply_inductive_T_cost}
\end{equation}

We define the circuit \(C_{f, n, m, \lambda, k, t+1}\) by taking the circuit that implements the two-copy mass production in \(2^t\) times in parallel and replacing the parallel calls to each \(G_\ell\) with the corresponding \(C_{g_\ell, n-k, m, \lambda, k, t}\).
Now we need to show that the number of gates scales appropriately.
We begin by considering the number of Clifford gates, finding that
\begin{align}
	\textsc{Clifford}\left( C_{f, n, m, \lambda, k, t + 1} \right) \leq
	\underbrace{\left(1 + 2^k\right)^{t+1} \left( K + o(1) \right) \left(2^{n-\left(t+1\right)k} m + \bigo{n} \right)}_{\text{Mass producing the } 2^k + 1 \text{ calls to the }
	G_{\ell}^{\otimes 2^{t}}.
	}
	\nonumber
	\\
	+
	\underbrace{2^t \mathcal{O} \left(2^k \left(n + m \right)\right)}_{\text{The \(2^t\) copies of the control flow circuitry.}}
\end{align}
The first term comes from accounting for the cost of implementing a call to \(C_{g_{\ell}, n-k, m, \lambda, k ,t}\) for each of the \(2^{k} + 1\) different data lookups \(G_\ell\), i.e., we multiply \Cref{eq:sub_step_apply_inductive_clifford_cost} by \(2^k + 1\).
The second term follows from multiplying the complexity of the control flow circuitry for a single two-query mass production protocol by \(2^t\) (see \Cref{app:base_case_cost_analysis} for more details).
Because \(2^t 2^{k} < \left( 1 + 2^k \right)^{t+1}\) and \(n + m = \bigo{n}\), the second term can be absorbed into the \(\left( 1 + 2^k \right)^{t+1} K \bigo{n}\) component of the first term.
As a result, we have
\begin{equation}
	\textsc{Clifford}\left( C_{f, n, m, \lambda, k, t + 1} \right) \leq
	\left(1 + 2^k\right)^{t+1} \left( K + o(1) \right) \left(2^{n-\left(t+1\right)k} m + \bigo{n}\right).
\end{equation}
This is the desired inequality for the number of Clifford gates in the \(t+1\) case.

Now we turn our attention to the number of Toffoli gates.
Performing the same calculation in this case, we have
\begin{align}
	\textsc{Toffoli}\left( C_{f, n, m, \lambda, k, t + 1} \right) =
	\underbrace{\left(1 + 2^k\right)^{t+1} \left( 1 + o(1) \right) \left(2^{n-\left(t+1\right)k} \lambda^{-1} + \lambda m + \bigo{n}\right)}_{\text{Mass producing the } 2^k + 1 \text{ calls to the }
	G_{\ell}^{\otimes 2^{t}}} \nonumber
	\\
	+ \underbrace{2^t \mathcal{O} \left(2^k \left(n + m \right)\right)}_{\text{The \(2^t\) copies of the control flow circuitry}},
\end{align}
where the first term arises from mass-producing the calls to the \(G_\ell\) and the second term arises from having \(2^{t}\) copies of the control flow circuitry used by the two-copy protocol.
Following the same steps as the Clifford case, we can absorb the second term, yielding
\begin{equation}
	\textsc{Toffoli}\left( C_{f, n, m, \lambda, k, t + 1} \right) =
	\left(1 + 2^k\right)^{t+1} \left( 1 + o(1) \right) \left(2^{n-\left(t+1\right)k} \lambda^{-1} + \lambda m + \bigo{n}\right).
\end{equation}
as desired, completing the proof of the lemma.
\end{proof}

Having proved \Cref{lemma:mass_produced_qrom}, we are ready to prove \Cref{thm:mass_produced_qrom}, which we recall below.
Our proof is straightforward, essentially using the assumptions in the statement of \Cref{thm:mass_produced_qrom} and a particular choice for \(k\) to simplify the Clifford and Toffoli counts of \Cref{lemma:mass_produced_qrom} further.
\qromthm*

\begin{proof}
We set \(k = \ceil{\log_2 n}\) and recall the assumption that \(r = 2^{o(\frac{n - 2 \log_2 \lambda}{\log_2 n})}\), or, equivalently, \(t = o(\frac{n - 2 \log_2 \lambda}{\log_2 n})\).
Rearranging \Cref{eq:overall_clifford_cost_inductive_assumption} by pulling out and cancelling a factor of \(2^{kt}\), we find that
\begin{equation}
	\textsc{Clifford}\left( C_{f, n, m, \lambda, k, t} \right) \leq \left(1 + 2^{-k}\right)^t \left( K + o(1) \right) \left(2^{n} m + 2^{kt} \bigo{n}\right).
\end{equation}
Taking advantage of the fact that \( 1 + x < e^x\) for \(x > 0\), we have \(\left( 1 + x \right)^t < e^{tx}\) for any \(x, t > 0\), so \(\left(1 + 2^{-k}\right)^t < e^{\frac{t}{2^k}} \leq e^{\frac{t}{n}}\).
The assumption that \(t = o(\frac{n - 2 \log_2 \lambda}{\log_2 n})\) implies that \(\frac{t}{n} = o(\log_2^{-1} n)\), so \(e^{\frac{t}{n}} = 1 + o(1)\), and therefore \(\left( 1 + 2^{-k} \right)^t = 1 + o(1)\).
Asymptotically, \(2^{kt} \bigo{n}\) is bounded by \(2^{kt + c \log_2 n}\) for some constant \(c\).
Since \(kt = o(n - 2 \log_2 \lambda) = o(n)\), we also have \(kt + c \log_2 n = o(n)\) and therefore \(2^{kt} \bigo{n}\) is dominated by \(2^{n} m\).
Applying both of these simplifications yields
\begin{equation}
	\textsc{Clifford}\left( C_{f, n, m, \lambda, k, t} \right) \leq \left( K + o(1) \right) 2^{n} m.
	\label{eq:final_clifford_scaling}
\end{equation}

Beginning a similar analysis of \Cref{eq:overall_T_cost_inductive_assumption}, we have
\begin{equation}
	\textsc{Toffoli}\left( C_{f, n, m, \lambda, k, t} \right) = \left(1 + 2^{-k}\right)^t \left( 1 + o(1) \right) \left(2^{n} \lambda^{-1} + 2^{kt} \lambda m + 2^{kt}\bigo{n}\right).
\end{equation}
Following the same reasoning as above, since \(\left(1 + 2^{-k}\right)^t = 1 + o(1)\), we have
\begin{equation}
	\textsc{Toffoli}\left( C_{f, n, m, \lambda, k, t} \right) = \left( 1 + o(1) \right) \left(2^{n} \lambda^{-1} + 2^{kt} \lambda m + 2^{kt}\bigo{n}\right).
\end{equation}
In order to establish that the \(2^{kt} \bigo{n}\) term is superfluous, we note that \(kt = o(n)\) implies that \(2^{kt} \bigo{n}\) is dominated by \(2^{n/2} m\).
From our assumption that \(\lambda \leq 2^{n/2}m^{-1/2}\), we have that \(2^{n} \lambda^{-1} \geq 2^{n/2}\), and therefore the \(2^{kt} \bigo{n}\) term can be safely dropped.
Finally, using our assumption on \(t\), we have \(2^{kt} \lambda m = 2^{o(n - \log_2 \lambda)}\), implying that the contribution from this term is dominated by \(2^{n} \lambda^{-1}\).
Therefore we can further simplify to yield
\begin{equation}
	\textsc{Toffoli}\left( C_{f, n, m, \lambda, k, t} \right) = \left( 1 + o(1) \right) 2^{n} \lambda^{-1},
\end{equation}
as desired, completing the proof of the theorem.
\end{proof}

\subsection{Maximizing the improvement factor}
\label{app:max_improvement_factor}

Taking \(\lambda\) to be a constant and assuming that Clifford costs saturate the inequalities in \Cref{eq:simple_qroam_clifford_cost} and \Cref{eq:final_clifford_scaling}, \Cref{thm:mass_produced_qrom} implies that we can implement \(O_f^{\otimes r}\) for a cost that is asymptotically equal to the cost of implementing \(O_f\) for any \(r\) such that \(r = 2^{o\left( n / \log n \right)}\).
Recall the definition of the improvement factor \(\mathcal{I}\),
\begin{equation}
	\mathcal{I} = \frac{\textsc{Cost}\left( O_f^{\otimes r} \right)}{\textsc{Cost}\left( C_{f, t} \right)} = \frac{r \textsc{Cost}\left( O_f \right)}{\textsc{Cost}\left( C_{f, t} \right)},
	\label{eq:improvement_factor_app}
\end{equation}
where we drop the irrelevant subscripts and use \(C_{f, t}\) to denote the circuit that mass-produces \(O_f^{\otimes r}\) for some \(r = 2^t\).
When \(\textsc{Cost}\left( O_f \right) \sim \textsc{Cost}\left( C_{f, t} \right)\), we simply have \(\mathcal{I} \sim r\).
As a consequence of the restriction that \(r = 2^{o (n / \log n)}\),
\begin{equation}
	\mathcal{I} = o \left( \text{Poly}\left( 2^n \right) \right).
\end{equation}
In other words, the improvement factor we obtain by applying \Cref{thm:mass_produced_qrom} is subpolynomial in the total size of the data set represented by \(f\).

Here we consider optimizing the mass production protocol to obtain a larger improvement factor.
This is a different aim than the one that motivated \Cref{thm:mass_produced_qrom} and prior works on mass production theorems~\cite{Ulig1974-xh,Uhlig1992-tg,Kretschmer2022-uy}.
These formulations of mass production focused on implementing a large number of parallel operations for a cost that is asymptotically proportional to the cost of implementing the same operation once.
Here we relax the requirement that the mass production circuit has a cost that is asymptotically proportional to the cost of a single copy and instead aim to maximize the improvement factor.
This optimization is captured in the following proposition:
\begin{restatable}[Mass production with a larger number of copies]{proposition}{improvementfactorprop}
	Let \(n\) be a positive integer such that \(n > a\), where \(a\) is a positive integer and a constant function of \(n\).
	Let \(f: \left\{ 0, 1 \right\}^n \rightarrow \left\{ 0, 1 \right\}^m\) be arbitrary, with \(m\) a constant function of \(n\).
	Let \(O_f\) be a unitary operator that queries this function in superposition, \(O_f: \ket{x}\ket{\alpha} \rightarrow \ket{x} \ket{\alpha \oplus f(x)}\).
	There exists a circuit \(C\) composed of one- and two-qubit Clifford gates and Toffoli gates implementing \(O_f^{\otimes r}\) for \(r = 2^{n-a}\) such that 
	\begin{equation}
		\textsc{Cost}\left( C \right) = \widetilde{\Theta} \left( 3^n\right),
	\end{equation}
	where the tilde notations indicates that we neglect logarithmic factors.
	\label{thm:max_scaling_prop}
\end{restatable}

\begin{proof}

Our proof proceeds by applying \Cref{lemma:mass_produced_qrom} with \(k=1\).
This choice maximizes the number of times we can recursively apply mass production while still offering a reasonable reduction in cost.
For simplicity, we will also choose \(\lambda\) such that the Clifford costs dominate by taking \(\lambda = 2^{n/4}\).
For any positive integer \(t < n\), Lemma~\ref{lemma:mass_produced_qrom} establishes the existence of a circuit \(C\) with the following properties:
\begin{enumerate}
	\item \(C\) implements \(O_f^{\otimes r}\) where \(r=2^t\).
	\item \(C\) is composed of one- and two-qubit Clifford gates and Toffoli gates.
	\item The cost of \(C\) is given by the expression
	\begin{equation}
		\textsc{Cost}\left( C\right) \leq 3^t \left( K + o(1) \right) \left(2^{n-t} m + \bigo{n}\right).
	\end{equation}
\end{enumerate}

Now we take \(t = n - a\), yielding
\begin{equation}
		\textsc{Cost}\left( C\right) \leq 3^{n - a} \left( K + o(1) \right) \left(2^{a} m + \bigo{n}\right) = 3^n \mathcal{O}\left( n \right).
\end{equation}
We make the conservative assumption that \(\textsc{Cost}\left( C \right)\) saturates the inequality (note that we could always add more gates to \(C\) to make this true without loss of generality).
We therefore achieve the desired scaling, completing the proof.
\end{proof}

When we take \(r = 2^{n - a}\) for some suitable constant \(a\) and apply \Cref{thm:max_scaling_prop}, we find that the improvement factor is
\begin{equation}
	\mathcal{I} = \frac{r \textsc{Cost}\left( O_f \right)}{\textsc{Cost}\left( C \right)} = \widetilde{\Theta}\left( \frac{4^n}{3^n} \right) = \widetilde{\Theta}\left( 2^{n\left(2 - \log_2 3\right)}\right).
	\label{eq:big_theta_without_a}
\end{equation}
Note that since \(a\) is a constant we can essentially omit it from the expression in \Cref{eq:big_theta_without_a}.
Using this approach, we are able to obtain an improvement factor that scales polynomially in \(2^n\).
A simple calculation reveals that the improvement factor has a negative derivative for all $k>1$, under the assumption that $t=(n-a)/k$.
This raises the question: is possible to obtain a polynomial improvement factor with an exponent larger than \(2 - \log_2 \left( 3 \right) \approx .415\) by some other construction?

\subsection{Mass production for QROM implies mass production for state preparation and unitary synthesis}
\label{app:qrom_for_state_prep_and_synthesis}

Nearly-optimal methods for approximately preparing arbitrary states and approximately implementing arbitrary unitaries have been known for some time~\cite{Grover2002-st, Low2018-uu}.
In order to provide the best known asymptotic scaling, we can rely on the construction for state preparation recently introduced in \citen{Gosset2024-xi}.
\citen{Gosset2024-xi}'s results are focused on achieving the optimal asymptotic scaling with respect to the number of non-Clifford T gates.
To accomplish this goal, they design a state preparation protocol that uses \(\mathcal{O}\left( 1 \right)\) QROM calls, using the SelectSwap techniques of \citen{Low2018-uu} to implement these calls with the smallest possible T gate cost.
While their results are framed in terms of the overall T-gate complexity, it is straightforward to extract the number of QROM calls and the complexities of the other steps from their exposition.

\citen{Gosset2024-xi}'s main results on state preparation assume that one is preparing a state with real amplitudes, and they separately show how to efficiently implement an arbitrary diagonal unitary to address the general case.
According to Lemma 3.5 and the proof of Theorem 1.1, the state preparation requires:
\begin{itemize}
	\item \(\mathcal{O}\left( 1 \right)\) reflections about the zero state on \(n + \log \log \left( 1 / \epsilon \right) + \mathcal{O}(1)\) qubits, which can be implemented using \(n + \log \log \left( 1 / \epsilon \right) + \mathcal{O}(1)\) one- and two-qubit Clifford gates and T gates;
	\item \(\mathcal{O}\left( 1 \right)\) reflections about a smaller state;
	\item \(\mathcal{O}\left( n \right)\) Hadamard gates;
	\item \(\mathcal{O}\left( 1 \right)\) QROM reads with \(n + \log \log \left( 1 / \epsilon \right)\) input bits and one output bit;
	\item A layer of single-qubit gates that can be synthesized using \(\log \left( 1 / \epsilon \right)\) T gates altogether.
\end{itemize}
Implementing the diagonal unitary to the necessary precision requires (according to Theorem 1.2):
\begin{itemize}
	\item \(\mathcal{O}\left( \log \left( 1 / \epsilon \right) \right)\) one- and two-qubit gates from a Clifford + T gateset.
	\item \(\mathcal{O}\left( 1 \right)\) QROM reads with \(n\) input bits and \(\mathcal{O} \left( \log \left( 1 / \epsilon \right)\right)\) output bits.
\end{itemize}

Now let us consider the cost of implementing \(r\) copies of this state preparation procedure in parallel for some \(r = 2^{o(n / \log n)}\) that is a power of two.
We will use \Cref{thm:mass_produced_qrom} to mass-produce the QROM calls, setting \(\lambda\) to be a constant for simplicity.
The number of Clifford gates dominates the gate complexity of the QROM calls, and \(2^n m\) is \(2^n \log \left( 1/ \epsilon \right)\) for both types of call, so the overall cost from the QROM calls is \(\mathcal{O}\left( 2^n \log \left( 1 / \epsilon \right) \right)\) gates\footnote{Note that even though \(m\) is assumed to be a constant in the statement of \Cref{thm:mass_produced_qrom}, we can apply the results here without a loss of rigor by performing \(\log\left( 1/\epsilon \right)\) QROM calls with a constant number of output bits.}.
The non-QROM steps of each state preparation routine have a gate complexity that is bounded by \(n + \log \left( 1 / \epsilon \right)\), so the overall cost from these non-mass-producible components of the \(r\) parallel implementations is \(\mathcal{O}\left( 2^{o (n / \log n)} \left( n + \log \left( 1/ \epsilon \right) \right) \right)\), which is dominated by the complexity of the QROM calls.
Therefore, the overall cost to implement \(r\) arbitrary state preparation routines in parallel is 
\begin{equation}
	\textsc{Cost} \left( \ket{\psi}^{\otimes r} \right) = \mathcal{O}\left( 2^n \log \left( 1 / \epsilon \right) \right).
	\label{eq:parallel_state_prep_cost_gosset}
\end{equation}

There are a variety of ways to show a similar result for the parallelization of arbitrary unitary synthesis using the ability to mass-produce QROM reads.
For example, \citen{Low2018-uu} explains a technique for reducing unitary synthesis to (\(\approx 2^n\) calls to) arbitrary state preparation.
However, another approach we can take is to follow Kretschmer, who explains in the proof of Theorem 2 in \citen{Kretschmer2022-uy} how an arbitrary unitary on \(n\) qubits can be implemented as a product of \(2^n - 1\) controlled single-qubit rotations, each with \(n - 1\) controls.
Using the same diagonal unitary implementation as above (from Theorem 1.2 of \citen{Gosset2024-xi}), we can implement each single-qubit rotation up to an error \(\epsilon\) in the spectral norm using \(\mathcal{O}\left( \log \left(1/\epsilon\right) \right)\) one- and two-qubit Clifford gates and T gates, together with \(\mathcal{O}\left( 1 \right)\) QROM reads from \(n\) input bits to \(\mathcal{O}\left( \log \left( 1 / \epsilon \right) \right)\) output bits.
Taking \(r = 2^{o (n / \log n)}\) as above and letting \(\lambda\) be a constant, we can mass-produce the QROM reads and implement an arbitrary \(U^{\otimes r}\) to within a precision \(\epsilon'\) by setting the error for each single-qubit rotation to \(\epsilon =  \epsilon' / 2^n\) for a cost
\begin{equation}
	\textsc{Cost}\left( U^{\otimes r} \right) = \mathcal{O}\left( r 2^n \log \left( 2^n /  \epsilon' \right) + 4^n \log \left( 2^n / \epsilon' \right)\right) = \mathcal{O} \left( 4^n \left( n + \log \left( 1 / \epsilon' \right) \right)\right).
\end{equation}

\section{The gate complexity of Kretschmer's mass production}

The main body of our work examined the use of the classical mass production theorem of Uhlig~\cite{Uhlig1992-tg} and in \Cref{app:qrom_for_state_prep_and_synthesis}, we discussed how this result alone can be used to enable mass production of quantum states and unitaries through QROM constructions.
However, there are cases where a direct application of the mass production theorem may be preferable due to memory considerations or other concerns.
Furthermore, it is instructive to compare the costs achievable through a QROM-based approach with the costs achievable when using prior work directly.
As we shall see, our mass production results for state preparation actually have a slightly lower cost than those of \citen{Kretschmer2022-uy} due to the fact that we were able to build on the state preparation techniques of \citen{Gosset2024-xi}.
Here we provide an extension of the analysis in \citen{Kretschmer2022-uy} goes beyond counting the number of CNOT gates and specifically accounts for the number of gates in a discrete Clifford + T gate set.
This is relevant especially for settings where CNOT operations do not dominate the overall cost such as in both magic state cultivation and traditional magic state distillation~\cite{gidney2024magic}.

The main results of \citen{Kretschmer2022-uy} are captured in the following theorems:
\begin{theorem}
	Let \(\ket{\psi}\) be an \(n\)-qubit quantum state, and let \(r = 2^{o(n / \log n)}\).
	Then there exists a quantum circuit with at most \(\left( 1 + o\left(1\right) \right)2^n\) CNOT gates that prepares \(\ket{\psi}^{\otimes r}\).
\end{theorem}
\begin{theorem}
	Let \(U\) be an \(n\)-qubit unitary transformation, and let \(r = 2^{o(n / \log n)}\).
	Then there exists a quantum circuit with at most \(\left( 5/2 + o\left(1\right) \right)4^n\) CNOT gates that implements \(U^{\otimes r}\).
\end{theorem}

By applying standard results for synthesizing single-qubit unitaries, we can obtain the following corollaries:
\begin{corollary}
	Let $\ket{\psi}$ be an $n$-qubit quantum state.
	There exists a quantum circuit using at most $24(1+o(1))2^n\left(n+\log_2(6/\epsilon)\right)$ Hadamard, T, and CNOT gates that prepares the state $\ket{\bar{\psi}}^{\otimes r}$, where \(\ket{\bar{\psi}}\) is a state such that \(\norm{\ket{\psi} - \ket{\bar{\psi}}} < \epsilon\) in the Euclidean norm.
	Similarly, for any \(n\)-qubit unitary \(U\), there exists a quantum algorithm that uses $60(1+o(1))4^n\left(2n+\log_2(15/\epsilon)\right)$ Hadamard, T, and CNOT gates to implement \(\bar{U}^{\otimes r}\), where \(\norm{U - \bar{U}} < \epsilon\) in the operator norm.
\end{corollary}
Note that the number of gates required for parallel state preparation using this approach is asymptotically larger than the result we obtain using QROM-based methods (see \Cref{eq:parallel_state_prep_cost_gosset}).

\begin{proof}
	The proof immediately follows from the following observation:
	Any quantum circuit on \(n\)-qubits composed of single-qubit gates and \(L\) CNOT gates can be written as a product of at most \(n + 2L\) arbitrary single-qubit gates and \(L\) CNOT gates.
	To see this, consider the following canonical form for a circuit composed of \(L\) CNOT gates:
	\begin{equation}
		\left(\prod_{j=1}^L (U_{x_j} \otimes U_{y_j}){\rm CNOT}_{x_j,y_j}\right)\left(\bigotimes_{i=1}^{n} U_i\right),
	\end{equation}
	where we use the notation $U_i$ to denote a single-qubit unitary acting on the \(i\)th qubit.
	Any circuit composed of single-qubit unitaries and CNOT gates can be put into this form by merging single-qubit unitaries together to the right until they are obstructed by a CNOT gate or the end of the circuit is reached.

	We can make use of this canonical form to count the number of gates required in a discrete Clifford + T gate set.
	Recall that at most $4\log_2(1/\delta) +O(1)$ Hadamard and T gates are needed to approximate a single-qubit rotation about the Z axis to within a factor of \(\delta\) in the operator norm (which is also an induced Euclidean-norm) \cite{ross2014optimal}.
	A single qubit rotation consists of at most $3$ single qubit rotations using an Euler angle decomposition, which implies that $12\log_2(3/\delta)+O(1)$ Hadamard and T gates ar sufficient to provide a \(\delta\)-accurate approximation to an arbitrary single-qubit rotation.
	The canonical circuit consists of $n + 2L$ single qubit rotations and \(2L\) CNOT gates and therefore we can bound the number of gates required for an overall precision of \(\epsilon\) by setting \(\delta = \epsilon/(n + 2L)\), yielding
	\begin{equation}
		(n + 2L) (12\log_2(3(n +2 L)/\epsilon) + \mathcal{O}(1)).
		\label{eq:number_of_gates}
	\end{equation}

	To calculate the total number of gates in a discrete gate set for Kretschmer's mass production, we set the \(n\) in \Cref{eq:number_of_gates} to \(rn\), and \(L\) to \(\left(1 + o(1)\right)2^n\), finding that the overall number of gates required in a Clifford + T gate set is bounded by
	\begin{equation}
		24(1+o(1))2^n\left(n+\log_2(6/\epsilon)\right).
	\end{equation}
	For the case of unitary synthesis $L=(5/2+o(1))4^n$~\cite{Kretschmer2022-uy}, and so the number of gates needed is bounded by
	\begin{equation}
		60(1+o(1))4^n\left(2n+\log_2(15/\epsilon)\right),
	\end{equation}
	completing the proof.
\end{proof}

\section{Additional supporting data and computational details}
\label{app:numerical_details}

In the main text we presented some numerical data comparing the cost of our mass production protocol with the cost of a naive implementation of parallel data loading.
In this section, we provide some additional details on our numerical implementations as well as some supplementary data in support of our conclusions.

In our asymptotic analysis we argued that the cost of the data loading subroutines \(G_\ell\) dominates the cost of mass production.
Asymptotically this is true, but it would be useful to understand how true this is for implementations at practical sizes.
To this end, we take our optimized circuits for mass production and plot the fraction of the costs not due to data loading in  \Cref{fig:not_data_loading}.
Specifically, we calculate the following quantity as a function of the number of input bits for several values of \(\Xi\) and \(r\),
\begin{equation}
	1 - \frac{\textsc{Cost}\left( \left\{ G_\ell \right\} \right)}{\textsc{Cost}\left( C \right)},
\end{equation}
where \(\textsc{Cost}\left( C \right)\) denotes the cost of the entire mass production protocol and \(\textsc{Cost}\left( \left\{ G_\ell \right\} \right)\) denotes the cost of all of the data loading steps involved.

\begin{figure}
	\centering
	\includegraphics[width=.7\textwidth]{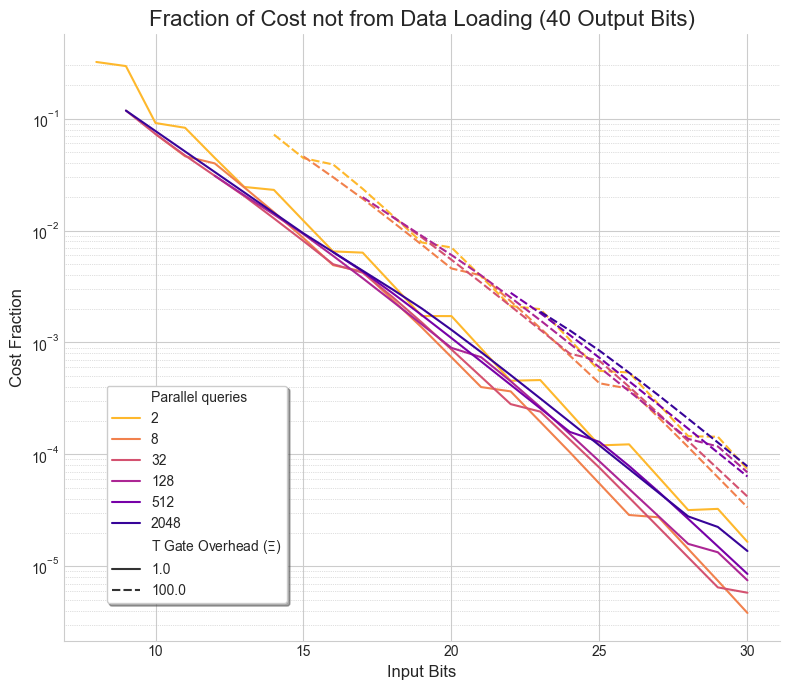}
	\caption{
	The fraction of costs not due to data loading in our optimized mass production circuits.
	We fix the number of output bits to \(40\) and vary the number of input bits (\(n\)), the T gate overhead (\(\Xi\)), and the number of parallel copies (\(r\)).
	}
	\label{fig:not_data_loading}
\end{figure}

In the main text, we presented some data that shows how the improvement factor available when using mass production depends on the number of input bits.
Specifically, in \Cref{fig:mass_production_helps}, we focused on the case where the number of output bits (\(m\)) is \(40\) and the T gate overhead (\(\Xi\)) is \(1\).
We plotted the improvement factor as a function of the number of input bits (\(n\)) for different choices of the total number of parallel queries (\(r\)).
In \Cref{fig:big_data}, we show similar data for several different values of \(m\) and \(\Xi\).
Each subfigure of this \(4 \times 3\) collection contains a single plot in the same format as \Cref{fig:mass_production_helps}, with the number of output bits and the T gate overhead specified in the subfigure title.
As in our other calculations, we numerically optimized the choice of \(k\) at each step in the recursive construction for mass production as well as the choice of \(\lambda\) in the data lookup subroutines.

\begin{figure}
	\centering
	\includegraphics[width=.95\textwidth]{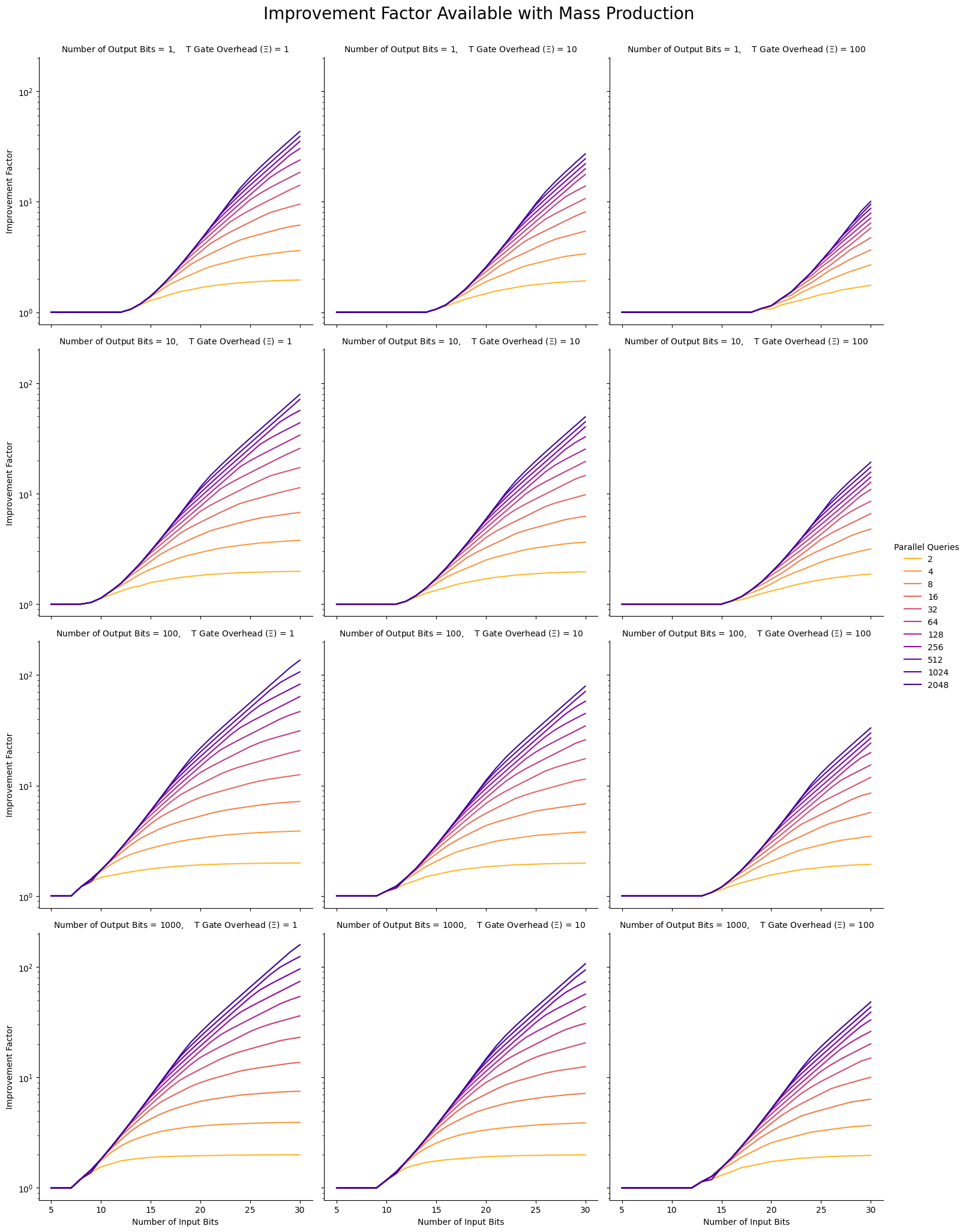}
	\caption{
		The ``improvement factor'' $\mathcal{I}$ when using mass production to implement \(r\) parallel queries to a data loading oracle for an arbitrary function that takes \(n\) input bits and outputs \(m\) bits.
		We plot the improvement factor for various values of the number of output bits (\(m\), different rows), the T gate overhead (\(\Xi\), different columns).
	}
	\label{fig:big_data}
\end{figure}

Our cost analysis was carried out using the Qualtran software package~\cite{Harrigan2024-rj}, which contains implementations of many of the basic primitives we use, including variants of ``SelectSwap'' QROM used for the \(G_\ell\) (with minor modification, as detailed in \Cref{app:modified_clean_qroam}) and the multicontrolled gates used in the construction of the \(A_\ell\).
For completeness, we mention some details about these implementations.
QROM and its variants include a data lookup step that repeatedly flips bits of the output register(s) conditioned on a single control qubit.
These multitarget single-control gates can be expressed as a product of CNOT gates, one for each \(1\) in the corresponding output.
When estimating the cost of this step without specific data, Qualtran upper bounds the number of CNOT gates by assuming that there is one for each output qubit.

In our construction of \(A_\ell\), we repeatedly use multicontrolled gates to help us efficiently flag which of the three cases we are in.
By default, Qualtran compiles these multicontrolled gates using a series of "AND" operations, which introduce additional ancilla.
However, since these multicontrolled gates have \(k\) controls and we generally take \(k\) to be either \(1\) or \(\approx \log_2 n\), this additional ancilla cost is negligible compared to the space required by the input and output registers (\(r \left( m + n \right)\)) and the ancilla space used by the SelectSwap QROM.
Furthermore, a slightly more sophisticated implementation could replace our naive use of multicontrolled gates with a strategy more like unary iteration, which would reduce the number of one- and two-qubit gates slightly at the expense of increasing the complexity of the presentation.

For additional details, we refer the reader to our Qualtran implementation of the mass production protocol.

\section{Applications}
\label{app:applications}

We discuss several potential practical applications of quantum mass production in this appendix.
We highlight several situations where it is natural to implement some mass-producible operation multiple times in parallel.
We focus the discussion around the mass production of QROM queries, although in some cases the applications that we consider could be framed in terms of mass production for state preparation or unitary synthesis.
While we present some numerical evidence about the potential for benefitting from mass production in \Cref{app:lcu_overview}, we leave the end-to-end analysis of particular applications to a future work.

QROM queries can be used to realize a wide range of quantum oracles.
We begin in \Cref{app:lcu_overview} by focusing on a particular well-studied example, the \(\prep\) oracle in a simple formulation of the linear combination of unitaries (LCU) framework.
By analyzing previously-published data using the Qualtran software package~\cite{Lee2021-su,Harrigan2024-rj}, we provide an example of a natural situation where the cost of QROM queries dominates the overall cost of implementing an algorithm.
Building on this, we explain in \Cref{app:parallel_qpe} and \Cref{app:parallel_mps} how the combination of parallel phase estimation with quantum mass production can provide a substantial reduction to the overall cost of eigenvalue estimation.
Along the way, we discuss the application of quantum mass production to the tensor hypercontraction simulation algorithm, providing a concrete example of a situation where the ability to mass-produce QROM queries appears more powerful than the ability to mass-produce state preparation or unitary synthesis.

We close by discussing a more general example in \Cref{app:parallel_amp_amp}, the combination of mass production with amplitude amplification.
Specifically, we argue that quantum mass production will enable us to reduce the overall cost of algorithms that require amplifying some small success probability using amplitude amplification.
We leave the analysis of particular applications to future work, but we note that there are many quantum algorithms that 1) rely heavily on amplitude amplification and 2) involve oracles whose implementation may require large QROM queries~\cite{Harrow2009-sm, Childs2017-is,Berry2014-fz, Berry2017-pd, Liu2021-ud, Lin2020-qh, Lloyd2016-po, Liu2024-dv}.

\subsection{Constructing block encodings with a linear combination of unitaries}
\label{app:lcu_overview}

In this section, we briefly introduce the concepts of block encoding and linear combination of unitaries.
We refer the reader to \citen{Gilyen2019-jt} for a more comprehensive introduction.
Many variations on these ideas have been explored over the last few years, and we do not attempt to provide a comprehensive review.
Instead, we provide enough background to discuss some particular examples that illustrate how these primitives can benefit from quantum mass production.

The block encoding framework is a general approach to representing and manipulating matrices using quantum algorithms.
We say that a unitary \(U\) is a block-encoding of a square matrix \(A\) if
\begin{equation}
	\left(\bra{0}^{k} \otimes \mathbb{I}\right) U \left(\ket{0}^{k} \otimes \mathbb{I}\right) = A / \lambda,
\end{equation}
for some \(\lambda > 0\)\footnote{Block encodings can also be generalized to non-square matrices.}.
Equivalently, we can say that the sub-normalized \(A/\lambda\) is the top left block of \(U\),
\begin{equation}
	U =
	\begin{pmatrix}
		A/\lambda & \cdot
		\\
		\cdot     & \cdot
	\end{pmatrix}
	.
\end{equation}
Access to a block encoding allows us to perform a variety of tasks, many of which can be cast as implementing polynomial functions of the original matrix~\cite{Low2019-oa, Gilyen2019-jt, Martyn2021-mf}.

Block encodings can be obtained in various ways, but one of the most commonly used approaches is to represent the matrix \(A\) as a linear combination of (efficiently implementable) unitaries.
We begin with a decomposition of \(A\) as a linear combination of unitaries (LCU),
\begin{equation}
	A / \lambda = \sum_{k=0}^{N-1} \alpha_k U_k,
\end{equation}
where the \(\alpha_k\) are positive numbers such that \(\sum_{k=0}^{N-1} \alpha_k = 1\).
We define a unitary operator \(\prep\) by its action on the zero state,
\begin{equation}
	\prep \ket{0} \ket{0} = \sum_{k=0}^{N-1} \sqrt{\alpha_k} \ket{k} \ket{\text{junk}_k},
	\label{eq:prep_def}
\end{equation}
where \(\ket{k}\) is a computational basis state that encodes \(k\) (usually in binary), and \(\ket{\text{junk}_k}\) can be an arbitrary normalized state.
Then we define an operator \(\sel\),
\begin{equation}
	\sel = \sum_{k} \ketbra{k} \otimes \mathbb{I} \otimes U_k.
\end{equation}
It is straightforward to see that we can combine these two operators to block encode \(A\),
\begin{equation}
	\left(\bra{0} \bra{0} \otimes \mathbb{I} \right)  \prep^{\dagger} \cdot \sel \cdot \prep  \left( \ket{0} \ket{0}
	\otimes \mathbb{I}\right) = A/\lambda.
\end{equation}

One standard way to implement \(\prep\) is to use a technique known as coherent alias sampling~\cite{Babbush2018-tb}.
This technique allows one to efficiently approximate the unitary \(\prep,\) preparing a state of the form
\begin{equation}
	\ket{\phi} = \sum_{k=0}^{N-1} \sqrt{\tilde{\alpha}_k} \ket{k} \ket{\text{junk}_k},
	\label{eq:coherent_alias_sampling_state}
\end{equation}
where the \(\tilde{\alpha}_{k}\) are \(\mu\) bit approximations to the \(\alpha_k\) with \(\abs{\alpha_k - \tilde{\alpha}_k} \leq \frac{1}{2^{\mu}
	N} \).
The procedure begins by preparing a uniform superposition, \(N^{-1/2}\sum_{k=0}^{N-1} \ket{k}\), and then uses a single QROM query to coherently load bitstrings \(\ket{\text{alt}_k}\) and \(\ket{\text{keep}_k}\), resulting in the state
\begin{equation}
	N^{-1/2} \sum_{k=0}^{N-1} \ket{k} \ket{\text{alt}_k} \ket{\text{keep}_k}.
	\label{eq:coherent_alias_sampling_Phi}
\end{equation}
These registers are of size \(\ceil{\log_2(N)}\), \(\ceil{\log_2(N)}\), and \(\mu\) respectively.
We then adjoin a final register of size \(\mu\) initialized in a uniform superposition, \(2^{-\mu / 2}\sum_{\sigma=0}^{2^{\mu} - 1} \ket{\sigma}\).
To complete the procedure, we swap the first and second registers (containing \(\ket{k}\) and \(\ket{\text{alt}_k}\)) if \(\text{keep}_k \leq \sigma\).

As \citen{Babbush2018-tb} explains, it is always possible to choose the bitstrings \(\text{alt}_k\) and \(\text{keep}_k\) such that the resulting state takes the form of \(\ket{\phi}\) in \Cref{eq:coherent_alias_sampling_state}.
Furthermore, this classical preprocessing can be done sequentially in time that scales roughly linearly with \(N\).
As we discussed in \Cref{app:qrom_background}, the cost of the QROM read is driven by the number of one- and two- qubit Clifford gates, which scales as \(\mathcal{O}\left( N \left( \log_2 N + \mu \right) \right)\).
All other steps scale logarithmically in \(\mu\) and \(N\), so the cost of coherent alias sampling is dominated by the cost of the QROM read.

The overall cost of implementing a block encoding depends on the cost of both \(\prep\) and \(\sel\).
When we implement \(\prep\) using coherent alias sampling and the cost of \(\prep\) dominates the overall cost, then there is significant potential to benefit from applying quantum mass production.
To provide a concrete example, we use Qualtran to analyze the fraction of the cost that comes from classical data loading in the sparse simulation algorithm proposed in \citen{Berry2019-qo}.
The implementation of the \(\prep\) subroutine in this algorithm is slightly more sophisticated than the coherent alias sampling that we described above, but at a high level it works similarly and relies heavily on classical data loading.
The quantum chemical Hamiltonian is a two-body operator, so in a naive representation one might expect that a system of size \(N_{orb}\) (the number of spin-orbitals) would require loading \(\mathcal{O}\left( N_{orb}^4 \right)\) parameters.
In the sparse simulation approach of Ref.~\citen{Berry2019-qo}, the cost of the classical data loading is reduced if the Hamiltonian is sparse in a particular basis once very small terms have been dropped.
Asymptotically, there are arguments that the number of non-zero terms should scale as \(O(N_{orb}^2)\) for a sufficiently large system in a properly chosen basis~\cite{McClean2014-ll}.
Regardless of the sparsity, the other components of the algorithm scale nearly linearly in \(N_{orb}\).

In \Cref{fig:cost_fraction}, we consider the task of implementing a block encoding of the quantum chemical Hamiltonian in second quantization using this sparse simulation approach.
Specifically, we consider two examples reported in \citen{Lee2021-su}, a chain of Hydrogen atoms represented in an STO-6G basis where the system size was varied by changing the number of atoms, and a square configuration of four Hydrogen atoms where the system size was varied by changing the size of the single-particle basis set.
We refer to the first example as the ``Thermodynamic'' scaling limit and the second as the ``Continuum'' scaling limit.
In both cases, we use the data from Figure 14 of \citen{Lee2021-su} to fit an equation of the form
\begin{equation}
	c = a N_{orb}^b,
\end{equation}
where \(c\) is the number of non-zero coefficients, \(N_{orb}\) is the number of spin-orbitals, \(a\) is a free parameter, and \(b\) is taken from the caption of the figure (\(1.78\) and \(3.83\) for the two limits, respectively).
Then, using the cost model described in \Cref{app:cost_models}, we use the Qualtran software package to calculate the fraction of the cost that comes from the classical data loading subroutine as a function of the system size for various values of \(\Xi\) (the overhead of a T gate compared with a typical Clifford gate).
We plot this data in the left panel.
In the right panel, we show the number of input bits required for the classical data loading subroutine as a function of system size.

\begin{figure}
	\centering
	\includegraphics[width=.485\textwidth]{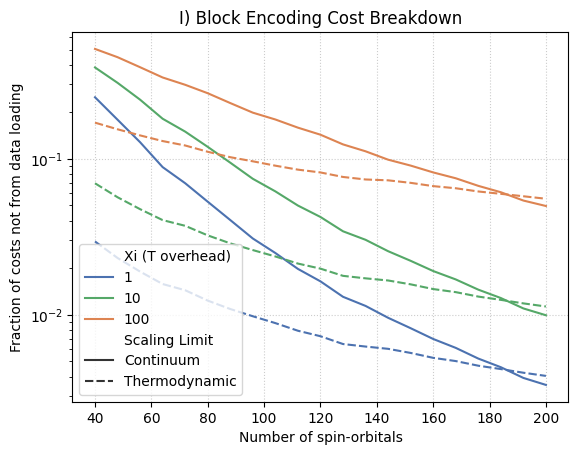}
	\;
	\includegraphics[width=.485\textwidth]{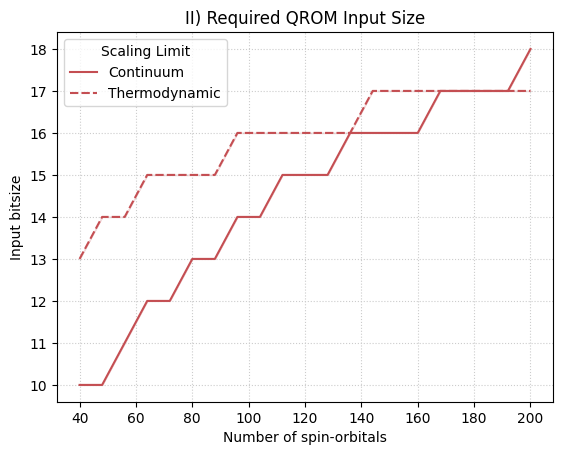}
	\caption{Left Panel (I): The fraction of the cost of block-encoding the quantum chemical Hamiltonian that is not due to data loading when using the sparse simulation algorithm of \citen{Berry2019-qo}, for several values of \(\Xi\) (the T gate overhead in the cost model defined in \Cref{app:cost_models}).
		We vary the system size by adding Hydrogen atoms to a chain (Thermodynamic) or by adding additional diffuse basis functions to a simulation of four Hydrogen atoms in a square geometry (Continuum).
		In both cases, the number of non-zero matrix elements is obtained from Figure 14 of \citen{Lee2021-su}.
		Right Panel (II): The number of input bits required by the classical data loading subroutine used to block-encode the Hamiltonians.
		We find that classical data loading dominates the cost, particularly at larger system sizes and smaller values of \(\Xi\).
		Furthermore, the number of input bits grows relatively large.
		Taken together, this suggests that it is possible to significantly benefit from mass production techniques.
	}
	\label{fig:cost_fraction}
\end{figure}

In both scaling limits, \Cref{fig:cost_fraction} reveals that the cost of implementing the block encoding using this algorithm is dominated by the QROM queries at large system sizes and small values of \(\Xi\).
Furthermore, the number of input bits is large enough that we may be able to obtain a significant cost reduction from mass production, although reducing the cost by an order of magnitude or more appears challenging at these system sizes (see \Cref{fig:mass_production_helps}).

While the linear combination of unitaries approach to constructing a block encoding is frequently used in state-of-the-art quantum algorithms, the actual implementation is often more complicated than the straightforward approach described above.
In particular, a significant body of work is dedicated to more advanced techniques for efficiently block-encoding the quantum chemical Hamiltonian~\cite{Babbush2018-tb,Berry2019-qo,von-Burg2021-yq, Lee2021-su}.
By varying the form of the LCU decomposition as well as the strategies for implementing \(\prep\) and \(\sel\), the costs of the block encoding can be dramatically reduced (both in terms of the asymptotic scaling and the constant factors)~\cite{Lee2021-su}.
Despite this variation, the cost of classical data loading dominates many of the most advanced algorithms for this application.
In the most advanced algorithms, this classical data is used for tasks other than arbitrary state preparation and unitary synthesis~\cite{von-Burg2021-yq, Lee2021-su}.
This observation highlights the benefit of our approach to quantum mass production, which focuses on parallelizing classical data loading rather than the more specific tasks of state preparation or unitary synthesis.

\subsection{Parallel Phase Estimation}
\label{app:parallel_qpe}

In this section, we discuss how quantum mass production can be combined with parallel phase estimation in order to reduce the gate complexity required for eigenvalue estimation.
In its standard form, parallel phase estimation assumes that we have access to a unitary \(U\) and a state such that $\ket{\psi_0}$ such that
\begin{equation}
	U\ket{\psi_0} = e^{i\phi_0} \ket{\psi_0}.
\end{equation}
The task is to estimate the phase \(\phi_0\).
Parallel phase estimation differs from standard phase estimation in that it uses \(r\) copies of the initial state to reduce the required depth.
Parallel phase estimation does not generally reduce the query complexity or gate complexity of the phase estimation task, although it can enable multiple quantum processors to collaborate on solving the phase estimation task with very little communication overhead.

We illustrate a circuit that implements the parallel version of an iterative phase estimation protocol in \Cref{fig:parPE}.
This approach begins by first preparing the GHZ state tensored with $r$ copies of the state $\ket{\psi_0},$
\begin{equation}
	\frac{1}{\sqrt{2}} \left(\ket{0}^{\otimes r} + \ket{1}^{\otimes r} \right) \ket{\psi_0}^{\otimes r}.
\end{equation}
We then take each qubit in the GHZ state and perform a controlled unitary $U$ onto the corresponding $\ket{\psi_0}$ state.
This yields the state
\begin{equation}
	\prod_{j=1}^r \left(\ketbra{0}{0}_j\otimes I_j + \ketbra{1}{1}_j\otimes U_j \right) \frac{1}{\sqrt{2}} \left(\ket{0}^{\otimes r} + \ket{1}^{\otimes r} \right) \ket{\psi_0}^{\otimes r} = \frac{1}{\sqrt{2}} \left(\ket{0}^{\otimes r} + e^{ir\phi_0}\ket{1}^{\otimes r}  \right) \ket{\psi_0}^{\otimes r},
\end{equation}
where we use the notation that $[\cdot]_j$ is an operator acting on the $j^{\rm th}$ ancilla qubit or state $\ket{\psi_0}$ as appropriate.
By using \(r\) parallel applications of the controlled form of \(U\), we apply a phase equal to \(r \phi_0\) to the GHZ state.
Uncomputing the GHZ state effectively leaves us with a single control qubit in a state proportional to \(\ket{0} + e^{ir\phi_0}\ket{1}\).
An equivalent experiment performed with the standard iterative phase estimation circuit would require total evolution time \(rt\) instead of \(t\).

\begin{figure}
\scriptsize
	\centering
	\begin{quantikz}[column sep = 1mm, row sep = 1mm]
		\lstick{$\ket{0}$}    & \gate{H} & \ctrl{4} & \ctrl{1}        & \ctrl{4} & \gate{H} & \qw
		\\
		\lstick{$\ket{\psi}$} & \qw      & \qw      & \gate{e^{-iHt}} & \qw      & \qw      & \qw
		\\
		\lstick{$\ket{0}$}    & \qw      & \targ{}  & \ctrl{1}        & \targ{}  & \qw      & \qw
		\\
		\lstick{$\ket{\psi}$} & \qw      & \qw      & \gate{e^{-iHt}} & \qw      & \qw      & \qw
		\\
		\lstick{$\ket{0}$}    & \qw      & \targ{}  & \ctrl{1}        & \targ{}  & \qw      & \qw
		\\
		\lstick{$\ket{\psi}$} & \qw      & \qw      & \gate{e^{-iHt}} & \qw      & \qw      & \qw
	\end{quantikz}
	\caption{Parallel QPE circuit for $r=3$ workers.  \label{fig:parPE}}
\end{figure}
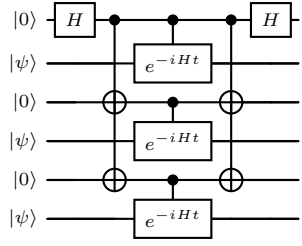

There are many ways that we could perform such a simulation but we will focus on qubitization based approaches here.
We assume that we have access to the Hamiltonian as a linear combination of \(N\) different unitaries, as reviewed in \Cref{app:lcu_overview}.
Using the oracles \(\prep\) and \(\sel\), we can construct a ``qubitized quantum walk operator'' \(W\) of the form
\begin{equation}
	W= - \sel\left(I - 2~\prep\ketbra{0}{0}\prep^\dagger\otimes I\right),
\end{equation}
where \(\ket{0}\) denotes the combined zero state of both registers in \Cref{eq:prep_def}.
\(W\) has the property that if $\ket{\psi_j}$ is an eigenvector of $H$ with eigenvalue $E_j,$ then there exist eigenstates $\ket{\phi_j^{\pm}}\in {\rm span}((\prep\ket{0})\ket{\psi_j}, W(\prep \ket{0})\ket{\psi_j} )$ with eigenvalues
\begin{equation}
	W\ket{\phi_j^{\pm}} = e^{\pm i \arccos(E_j/\lambda)} \ket{\phi_j^{\pm}},
\end{equation}
where \(\lambda\) is the normalization factor of the block encoding.
We can therefore directly perform phase estimation on \(W\) in order to determine properties about the spectrum of \(H\)\footnote{ Note that the controlled variant of the walk operator can be found by controlling the \(\sel\) operator, which can also be trivially constructed if atleast one of the unitaries has a known $+1$ eigenvector through controlled swap operations and a single query.
}.

The simplest application of quantum mass production in this context is in the application of the \(\prep\) operation.
This could be done either through the use of coherent alias sampling and a parallel QROM read, as we discussed in \Cref{app:lcu_overview}, or by directly using the techniques of \citen{Kretschmer2022-uy} for arbitrary state preparation.
Using the coherent alias sampling approach described above and setting \(\mu\) to a constant, then we have that \(\textsc{Cost}\left( \prep \right) = \mathcal{O}\left( N \log N \right)\).
Assuming that we are working in a regime where we can fruitfully mass-produce \(r\) copies of the QROM read for coherent alias sampling by applying \Cref{thm:mass_produced_qrom}, we can implement \(\prep^{\otimes r}\) for a comparable cost.
So the cost of implementing \(W^{\otimes r}\) with mass production is
\begin{equation}
	\mathcal{O}\left( N \log \left(N\right) + r \textsc{Cost}\left( \sel \right) \right),
	\label{eq:mass_production_costs_coherent_alias_sampling_W}
\end{equation}
whereas the cost without mass production would be
\begin{equation}
	\mathcal{O}\left( r N \log \left(N\right) + r \textsc{Cost}\left( \sel \right) \right).
\end{equation}
This implies that there is an asymptotic advantage provided that
\begin{equation}
	\textsc{Cost} \left( \sel \right) = o \left( N \log \left(N\right) \right).
\end{equation}

As we discussed in \Cref{app:lcu_overview}, the sparse simulation method of \citen{Berry2019-qo} satisfies this requirement.
Specifically, for this approach, \(\textsc{Cost}\left( \sel \right)\) scales as \(\mathcal{O}\left( N_{orb} \right)\), and the size of the QROM queries used to implement \(\prep\)  scales as \(N \propto N_{orb}^{b}\), where \(b\) is naively \(4\), but is perhaps \(2\) in some asymptotic limits.
In two particular examples (a Hydrogen chain and a square arrangement of Hydrogen atoms), \citen{Berry2019-qo} numerically estimated \(b \approx 1.78\) and \(b \approx 3.83\) respectively.
This allows us substantial room to benefit by taking \(r \gg 1\).

Naively, one might hope to balance the costs in \Cref{eq:mass_production_costs_coherent_alias_sampling_W} by taking \(r \approx N_{orb}^{b-1}\).
However, \Cref{thm:mass_produced_qrom} restricts us to values of \(r\) such that \(r = 2^{o \left( \frac{n - \log \lambda}{\log n} \right)}\).
For any choice of \(a > 0\), taking \(r = N^a\) implies \(r = 2^{n \log a}\), which does not satisfy the requirement on \(r\).
Nevertheless, we show in \Cref{thm:max_scaling_prop} that we can obtain a polynomial advantage when taking \(r = N/2\), reducing the cost of implementing all \(r\) data lookups from \(\Theta\left( N^2 \right)\) to \(\Theta \left(N^{\log_2 3} \right)\).
In other words, the amortized cost per data lookup goes as \(N^{\log_2 \left( 3/2 \right)}\).
For \(N \propto N_{orb}^b\), this means that \Cref{thm:max_scaling_prop} implies that we can implement \(W^{\otimes r}\) for \(r = N_{orb}^b /2\) for a cost that scales as 
\begin{equation}
	\mathcal{O}\left( N_{orb}^{b \log_2 3} \log \left( N_{orb} \right)\right) + N_{orb}^{b + 1}.
\end{equation}

Using the sparse simulation method to construct the walk operator, mass production with an appropriate value of \(r\), we can follow a parallel phase estimation approach to apply the transformation 
\begin{equation}
	\ket{\phi_j^{\pm}}^{\otimes r} \mapsto (W^t)^{\otimes r} \ket{\phi_j^{\pm}}^{\otimes r} = e^{\pm i r t \arccos(E_j/\lambda)} \ket{\phi_j^{\pm}}^{\otimes r}.
\end{equation}
Requiring that the error from iterative phase estimation is $O(\epsilon)$ with high probability, then the above circuit needs to be repeated a logarithmic number of times using a total evolution at each iteration that is on the order of $r t\in \bigot{\lambda_{sparse} / \epsilon}$~\cite{svore2013faster}.
Up to logarithmic factors, the overall cost is therefore
\begin{equation}
	\textsc{Cost} \left( \textrm{Parallel QPE with sparse simulation and mass production} \right) = \bigot{\frac{\lambda_{sparse}}{\epsilon} \left( N_{orb} + N_{orb}^{b  \log_2 \left(3/2 \right)} \right)},
	\label{eq:overall_cost_mass_produced_sparse}
\end{equation}
which is polynomially smaller than the cost without mass production,
\begin{equation}
	\textsc{Cost} \left( \textrm{Standard QPE with sparse simulation} \right) = \bigot{\frac{\lambda_{sparse}}{\epsilon}  N_{orb}^b}.
\end{equation}

We now briefly explain how mass production may be useful in accelerating this THC-based approach to quantum chemistry simulation, although we refer the reader to \citen{Lee2021-su} for a complete description of the algorithm.
Starting with a standard second-quantized representation of the electronic structure Hamiltonian,
\begin{equation}
	H=T+V=\sum_{pq} h_{pq}a^\dagger_p a_q+\sum_{pqrs} h_{pqrs} a^\dagger_p a^\dagger_q a_r a_s,
\end{equation}
tensor hypercontraction gives a method of rewriting the tensor describing the Coulomb operator in an approximate form,
\begin{equation}
	V \approx \sum_{pqrs} \sum_{\mu,\nu=1}^R \chi_p^{(\mu)} \chi_{q}^{(\mu)}\chi_r^{(\nu)} \chi_{s}^{(\nu)}\zeta_{\mu\nu} a^\dagger_p a^\dagger_q a_r a_s.
\end{equation}
Empirically, it has been argued that
\begin{equation}
	R = \bigot{N_{orb} \log(1/\epsilon_{\rm THC})}\label{eq:Reqn}
\end{equation}
is sufficient, where \(\epsilon_{\rm THC}\) is some error one allows in the THC decomposition. 

We claim that the actual cost (including the Clifford gates) of this algorithm is dominated by a QROM read with input size \(\approx R^2\) (and a small number of output bits that scales logarithmically with the desired precision).
Furthermore, the next most dominant cost is a QROM read with input size \(\approx R\) and output size \(\approx N_{orb}\), and all remaining components of the algorithm scale as \(\bigot{N_{orb}}\).
Notably, this second QROM read is not used for arbitrary state preparation (or a similar primitive), but instead it is used to coherently load parameters for a series of controlled rotations.
This application offers an example that highlights the power of using mass production for data loading rather than the more specific tasks of state preparation or unitary synthesis.

Mass production can help us reduce the cost of parallel calls to the walk operator in this case as well.
Applying \Cref{thm:max_scaling_prop} and performing a similar calculation to the one above reveals that we could reduce the amortized cost of implementing a call to the walk operator in the THC algorithm (in the context of parallel phase estimation) from \(\bigot{ R^2 + R N_{orb} + N_{orb}}\) to \(\bigot{R^{\log_2(9/4)} + R^{\log_2{3/2}} N_{orb} + N_{orb}}\).
Assuming the empirical scaling for \(R\) in~\eqref{eq:Reqn} holds, the overall cost therefore scales as 
\begin{equation}
	\textsc{Cost} \left( \textrm{Parallel QPE with THC and mass production} \right) = \bigot{\frac{\lambda_{THC}}{\epsilon} N_{orb}^{\log_2  3 }},
	\label{eq:parallel_qpe_thc_cost}
\end{equation}
where \(\lambda_{THC}\) is the rescaling factor of the THC block encoding.
Note that the QROM read with the smaller input size and larger output size is the dominant cost after applying \Cref{thm:max_scaling_prop}.
Contrast this with the cost available in the standard THC approach,
\begin{equation}
	\textsc{Cost} \left( \textrm{Standard QPE with THC} \right) = \bigot{\frac{\lambda_{THC}}{\epsilon} N_{orb}^{2}}.
\end{equation}
The scaling in \Cref{eq:parallel_qpe_thc_cost} represents a polynomial improvement in the gate complexity for this method when we account, as we do here, for the cost of Clifford gates as well as non-Clifford gates.

The parallel phase estimation considered above takes \(r\) copies of the appropriate eigenstate as an input to the simulation algorithm.
This is not a free resource and must be considered in the overall complexity.
The simplest way to address this is to use eigenstate filtering to prepare each of the individual eigenstates $\ket{\psi_j}$.
Let us assume that we are promised that we know the eigenenergy within error $\Delta$ and assume that the eigenvalue gap is at least $2\Delta$.
In this case, the cost of performing this preparation within error $\epsilon$ is given by~\cite{Lin2020-qh}
\begin{equation}
	\bigot{\frac{\lambda_{THC}
			N_{orb}^2 \log(1/\epsilon)}{\Delta}}
\end{equation}
which is sub-dominant if we assume that the gap is larger than $\epsilon$.
This is almost certainly the case if we are given an eigenstate of the Hamiltonian and merely use filtering to prepare a specific eigenstate of the walk operator.
In the more general situation, we may or may not have a gap between \(\Delta\) and \(\epsilon\).
The proposed combination of mass production and parallel phase estimation is most likely to be practically useful for examples where \(\Delta \gg \epsilon\).
We could, of course, use mass production to help implement multiple filtering steps in parallel as a prelude to using parallel phase estimation to estimate an eigenvalue.

\subsection{Matrix Product State Preparation}
\label{app:parallel_mps}

Parallel state preparation is a major obstacle in the application of parallel phase estimation.
In this section, we consider the use of mass production to accelerate the parallel preparation of several matrix product states.
This could serve as a complement or alternative to the phase estimation procedure discussed above, or it may be independently useful.
For example, the state-of-the-art approach for initial state preparation for first-quantized quantum algorithms naturally involves parallel calls to the same matrix product state preparation circuits as a subroutine~\cite{huggins2024efficient}.

Matrix product states take the form
\begin{equation}
	\ket{\psi} = \sum_{s\in \{0,1\}^n, \{\alpha\}}
	A^{s_0}_{\alpha_0} A^{s_1}_{\alpha_0 \alpha_1}\cdots A^{s_{n-1}}_{\alpha_{n-2}}\ket{s_0\ldots s_{n-1}}
\end{equation}
for a set of two- and three-index tensors $\{A^{s_0}_{\alpha_0} A^{s_1}_{\alpha_0 \alpha_1}\cdots A^{s_{n-1}}_{\alpha_{n-2}}\}$.
Matrix product states and are particularly efficient at representing states of one-dimensional quantum systems~\cite{Perez-Garcia2007-wt}, as well as amplitude-encoded functions of few variables~\cite{Oseledets2011-xm,huggins2024efficient}.
The individual dimensions of each of the \(\alpha\) indices are known as the ``bond dimension,'' which we denote as $\chi_j$ for the index $\alpha_j$.
For simplicity, we refer to the bond dimension of a matrix product state by a single number \(\chi\) such that \(\chi_j \leq \chi\) for all \(j\).
Matrix product states can exactly represent any state using an exponentially large bond dimension, and they are frequently used as ansatz for the ground state problem even when they are expected to scale exponentially.

It is also possible to approximately prepare a matrix product state using a quantum circuit whose cost scales polynomially with the bond dimension(s).
Specifically, consider the task of preparing a state $\ket{\phi}$ such that it closely approximates the target matrix product \(\ket{\psi}\), i.e.,
\begin{equation}
	\|\ket{\psi} - \ket{\phi} \| \le \epsilon
\end{equation}
for a chosen $\epsilon>0$.
A standard protocol for this task, discussed in Refs.~\citenum{huggins2024efficient,fomichev2024initial,Berry2024-qe} involves preparing an arbitrary state in a space of dimension \(\mathcal{O}\left( \chi \right)\) and then implementing a series of \(n-1\) unitaries of dimension \(\mathcal{O}\left( \chi \right)\).
This requires a number of gates that scales as \(\bigot{n \left( \chi^2 + \chi \log \left( 1/ \epsilon \right)\right)}\).
Note that optimizations to this process can reduce the number of non-Clifford Toffoli gates needed down to $\widetilde{O}(n \chi^{3/2} \log(1/\epsilon))$~\cite{huggins2024efficient,fomichev2024initial,Berry2024-qe}.
However, as argued previously this reduction may not be as useful in settings such as magic state cultivation where the costs of Clifford and non-Clifford gates are approximately equivalent.

Because this protocol for matrix product state preparation reduces the task to arbitrary state preparation and unitary synthesis, we can clearly accelerate it using mass production.
The states and unitaries that we require act on spaces of dimension at most \(\chi\), so we can benefit from mass-producing the most expensive such steps \(r\) times in parallel, provided that 
\begin{equation}
	r\in 2^{o(\log(\chi)/\log\log(\chi))}.
\end{equation}
The procedures outlined in \citen{fomichev2024initial}, \citen{huggins2024efficient}, and \citen{Berry2024-qe} already use QROM for the state preparation and unitary synthesis steps, so we could either apply our theorem to parallelize those QROM calls, or apply the results of \citen{Kretschmer2022-uy} directly.
In either case, we expect that a total of $\tilde{O}(\chi^2\log(1/\epsilon))$ gates will be required.

Recent work has found that using matrix product states with very high bond dimensions as initial states for quantum phase estimation is a promising route towards practical quantum advantage for quantum chemical simulation~\cite{fomichev2024initial, Berry2024-qe}.
\citen{Berry2024-qe}'s findings suggest that the cost of matrix product state preparation is moderately smaller than the cost of phase estimation for some of the smallest molecular systems that may be beyond the reach of classical methods.
For larger systems, the cost of matrix product state preparation is expected to grow exponentially, so it will be useful to be able to reduce its cost using mass production techniques.
This could be useful in the context of parallel phase estimation, or when it is necessary to perform several separate phase estimation experiments (on separate initial states) in order to have a high probability of obtaining the ground state energy.

\subsection{Parallel execution with amplitude amplification of success}
\label{app:parallel_amp_amp}

In this section, we consider the combination of our mass production protocol with amplitude amplification.
We consider an \(n\)-qubit unitary \(A\) such that
\begin{equation}
	A \ket{0^{\otimes n}} = \sqrt{p} \ket{G} \ket{1} + \sqrt{1 - p} \ket{B} \ket{0},
\end{equation}
where \(\ket{G}\) is some arbitrary ``good'' state and \(\ket{B}\) is some arbitrary ``bad'' state.
The usual task of amplitude amplification is to use queries to \(A\) and \(A^\dagger\) to prepare \(\ket{G}\) with some constant success probability greater than, say, \(1/2\).
This can be done using \(\approx \frac{1}{\sqrt{p}}\) queries to \(A\) and \(A^{\dagger}\).
We focus on the simplest case, where \(p\) is known ahead of time, although the more general case where \(p\) is unknown can be handled with a constant factor increase in cost~\cite{brassard2000quantum}.
Before we explain how to speed up amplitude amplification using quantum mass production, we note that query lower bounds for related problems, such as Grover's search algorithm~\cite{Zalka1999-no,brassard2000quantum}, will not apply in the models that we consider.
This is because implementing quantum mass production necessarily makes use of ``white-box'' access to the function being mass-produced.

We begin by making the assumption that we can fruitfully mass-produce up to \(r_{max}\) queries to \(A\).
Concretely, we assume that, there is a mass production protocol \(C(A, r)\) such that \(C(A, r) = A^{\otimes r}\) and
\begin{equation}
	\textsc{Cost}\left( C(A, r) \right) \approx \textsc{Cost}\left( A \right)
\end{equation}
for any \(r \leq r_{max}\), where \(r\) and \(r_{max}\) are both powers of two.
Similarly, we assume that we can mass-produce queries to \(A^{\dagger}\) under the same conditions.
Furthermore, we make the assumption that \(p r_{max} \ll 1\) in order to focus on the most interesting case.

Consider the action of \(A^{\otimes r}\) on the appropriate zero state,
\begin{equation}
	A^{\otimes r} \ket{0^{\otimes r n}} = \bigotimes_{i=1}^{r} \left( \sqrt{p} \ket{G} \ket{1} + \sqrt{1 - p} \ket{B} \ket{0} \right).
\end{equation}
Let \(\ket{\mathcal{G}}\) denote the ``good'' state obtained by expanding this tensor product, collecting all of the terms where at least one of the copies is in the state \(\ket{G}\ket{1}\), and normalizing.
Likewise, let \(\mathcal{B}\) denote the ``bad'' state \(\ket{\mathcal{B}} = \otimes_{i=1}^{r} \ket{B}\ket{0}\).
Finally, let \(\mathcal{A}_r\) denote the unitary that applies \(A^{\otimes r}\) and then uses a multi-qubit OR gate to flip an ancilla qubit if any of the \(A\)s have successfully prepared a \(\ket{G}\) state.
We assume that the cost of this OR operation is negligible compared to the cost of implementing \(A\).
Then we can write
\begin{equation}
	\mathcal{A}_r \ket{0^{\otimes n r + 1}} = \sqrt{p_{r}} \ket{\mathcal{G}}\ket{1} + \sqrt{1 - p_{r}} \ket{\mathcal{B}}\ket{0},
	\label{eq:script_A_def}
\end{equation}
where \(p_{r}\) is implicitly defined by the expression \(\sqrt{1 - p_r} = \left( \sqrt{1 - p} \right)^r\).
Solving for \(p_r,\) we have
\begin{equation}
	p_r = \sum_{i=1}^r \left( -1 \right)^{i + 1} \binom{r}{i} p^i = r p + \mathcal{O}(r^2 p^2).\label{eq:pr}
\end{equation}

In order to combine mass production with amplitude amplification, we simply apply amplitude amplification to \(\mathcal{A}_r\) to amplify the probability of obtaining \(\mathcal{G}\).
Once we have obtained \(\mathcal{G}\), we can either measure the individual ancilla qubits to determine which register is in the state \(\ket{G}\), or sort the registers according to their ancilla values (if we wish to avoid non-unitary operations).
Producing at least one copy of \(\ket{G}\) with high probability using this approach requires \(\approx \frac{1}{\sqrt{r p}}\) queries to \(\mathcal{A}_r\).
By assumption, the cost for implementing \(\mathcal{A}_r\) is comparable to the cost of implementing \(A\) for any \(r \leq r_{max}\).
Therefore, mass production allows us to reduce the costs for preparing \(\ket{G}\) by a factor of \(\sqrt{r_{max}}\).

\begin{corollary}
	Let \(A\) be a quantum algorithm that makes no measurements, implemented using \(q\) queries to an oracle \(O_f: \ket{x}\ket{\alpha} \rightarrow \ket{x} \ket{\alpha \oplus f(x)}\) for some known \(f:\left\{0,1\right\}^n \rightarrow \left\{0, 1\right\}^m\), together with at most \(c\ge 1\) additional %
    Clifford and Toffoli gates.
	Furthermore, let \(A \ket{0} = \sqrt{p} \ket{G} \ket{1} + \sqrt{1 - p} \ket{B} \ket{0}\) for some \(p < 1\), and assume that \(A\) uses at most \(n'\) qubits.
	For any \(r\) that satisfies \(r = 2^{o\left( n / \log n \right) }\) and \( r p \in o(1) \), there exists an algorithm that prepares a copy of the state \(\ket{G}\) with probability \(1 - \delta\) using \(\mathcal{O}\left(\frac{\log\left( 1/\delta \right)}{\sqrt{r p}} \left(2^n m q  + r c \right)\right)\) Clifford and Toffoli gates.
This algorithm uses at most \(r n' + 1\) qubits.
\end{corollary}

\begin{proof}
	We can implement the transformation in \Cref{eq:script_A_def} by combining \(r\) parallel queries to \(A\) with an OR operation that flags the success of any of the individual calls to \(A\).
	This requires exactly \(r n' + 1\) qubits and the OR operation requires a number of gates that scales linearly in \(r\).
	Rather than implementing each call to \(O_f^{\otimes r}\) directly, we can apply \Cref{thm:mass_produced_qrom} to do so with a number of gates that scales as \(\mathcal{O}\left( 2^n m \right)\) (taking \(\lambda\) to be a constant).
	Since $q$ queries need to be made to $O_f$, the cost of $\mathcal{A}_r$ is $q$ times the cost of querying $O_f^{\otimes r}$ followed by at most $rc$ additional Clifford and Toffoli gates where $c\ge 1$ is the additional Clifford and Toffoli gates needed by $A$ in addition to query operations. 
 The overall number of gates required to implement \(\mathcal{A}_r\) is therefore 
 \begin{equation}
     {\rm COST}(\mathcal{A}_r)= \mathcal{O}\left( 2^n m q + rc\right)
 \end{equation}

	We can use standard fixed-point amplitude amplification techniques to prepare \(\ket{\mathcal{G}}\) using \(\mathcal{O}\left( \frac{\log\left( 1/\delta \right)}{\sqrt{p_r}} \right)\) invocations of \(\mathcal{A}_r\) and its inverse~\cite{Yoder2014-ek} together with \(\mathcal{O}\left( \frac{\log\left( 1/\delta \right)}{\sqrt{p_r}} \right)\) additional gates.
	By assumption, \(rp\in o(1)\) so we have from~\eqref{eq:pr} \(p_r \sim r p\), implying that \(\mathcal{O}\left( \frac{\log\left( 1/\delta \right)}{\sqrt{r p}} \right)\) calls to \(\mathcal{A}_r\) are sufficient.
	Multiplying the number of calls to \(\mathcal{A}_r\) by the number of gates required to implement it yields the claimed complexity.
\end{proof}

\section{Mass production for serial operations}
\label{app:resource_state}

Quantum mass production is most useful when it is natural to parallelize calls to the same data loading oracle.
However, we show in this section that it is also possible to obtain some benefit when serial queries are required.
The basic idea is to use mass production to cheaply prepare several copies of the following resource state:
\begin{equation}
	\ket{\text{QROM}_{f}} = O_f \ket{+}^{\otimes n} \ket{0} = 2^{-n/2} \sum_{y=0}^{N - 1} \ket{y} \ket{f(y)}.
	\label{eq:qrom_resource_state}
\end{equation}
As we explain below, we can consume this resource state to implement the operation 
\begin{equation}
	\bar{O}_f^{(b)}: \ket{x}\ket{0} \rightarrow \ket{x}\ket{0 \oplus f(x \oplus b)},
\end{equation}
where \(b\) is a bitstring that is determined randomly when the resource state is consumed and the bar above \(O_f^{(b)}\) indicates that the output register must be in the zero state (we discuss the implications of this requirement after introducing the protocol).
We can then recover the operation
\begin{equation}
	\bar{O}_f: \ket{x}\ket{0} \rightarrow \ket{x}\ket{f(x)}
\end{equation}
by performing an additional QROM query,
\begin{equation}
	O_g: \ket{x}\ket{\alpha} \rightarrow \ket{x}\ket{\alpha \oplus g(x)},
\end{equation}
where we define
\begin{equation}
	g(x) = f(x) \oplus f(x \oplus b).
\end{equation}
Because \(g(x) = g(x \oplus b)\) for all \(x\), we can specify \(g(x)\) using half as much information as \(f(x)\). We show below how this allows us to implement \(O_g\) more cheaply than \(O_f\) or \(\bar{O}_f\).

Provided that the cost of storing several copies of the resource state is negligible, these observations allow us to nearly halve the cost of repeated serial queries to \(\bar{O}_f\).
In order to implement \(c\) serial queries to \(\bar{O}_f\), we first use mass production to prepare \(\ket{\text{QROM}_{f}}^{\otimes c}\) by replacing \(c\) parallel calls to \(O_f\) with a single call to a circuit \(C\) that mass-produces \(O_f^{\otimes c}\).
The cost to consume these resource states is dominated by implementing the \(O_g\), which, as we show below, is approximately half the cost of implementing \(O_f\).
Assuming that we are in the regime where \(\textsc{Cost}\left( C \right) \approx \textsc{Cost}\left( O_f \right)\), the overall cost is
\begin{equation}
	\textsc{Cost} \underbrace{ \left( \bar{O}_f \cdots \bar{O}_f \cdots \bar{O}_f \right)}_{c \text{ serial queries}} \approx \textsc{Cost}\left( O_f \right) + c \cdot \textsc{Cost}\left( O_g \right),
\end{equation}
which approaches \(\frac{c}{2} \cdot \textsc{Cost} \left( O_f \right)\) as \(c\) increases.

The process for consuming the resource states is simple.
Let \(\ket{\psi} = \sum_{x=0}^{N-1} c_x \ket{x}\) be an arbitrary state of an input register.
To obtain \(\bar{O}_f \ket{\psi} \ket{0}\), we begin with the input state \(\ket{\psi}\) and one copy of the resource state \(\ket{\text{QROM}_f}\) and perform the following steps:
\begin{align}
	\MoveEqLeft \ket{\psi} \ket{\text{QROM}_f} &
	\\
	                                           & = \ket{\psi} \left( 2^{-n/2} \sum_{y \in \{0,1\}^n} \ket{y}\ket{f(y)} \right)
	\\
	                                           & = 2^{-n/2} \sum_{x, y \in \{0,1\}^n} c_x \ket{x} \ket{y} \ket{f(y)}
	\\
	                                           & \xrightarrow{\text{Apply CNOTs: } \ket{x}\ket{y} \to \ket{x}\ket{x \oplus y}} 2^{-n/2} \sum_{x, y \in \{0,1\}^n} c_x \ket{x} \ket{x \oplus y} \ket{f(y)}
	\\
	                                           & \xrightarrow{\substack{\text{Measure 2nd register,}
	\\
			\text{outcome } b \in \{0,1\}^n}} \sum_{x \in \{0,1\}^n} c_x \ket{x} \ket{b} \ket{f(x \oplus b)}
	\\
	                                           & \xrightarrow{\text{Discard 2nd register (containing } \ket{b}\text{)}} \sum_{x \in \{0,1\}^n} c_x \ket{x} \ket{f(x \oplus b)}.
\end{align}
Note that each value of \(b\) is equally likely and that the post-measurement state in the above equations is normalized.
After discarding the unentangled register, the final state is equivalent to \(\bar{O}_f^{(b)} \ket{\psi} \ket{0}\), where \(\bar{O}_f^{(b)}\) is an operator that computes \(f(x \oplus b)\) into the output register.
To obtain the desired \(f(x)\), we apply a correction using a QROM read \(O_g\) for a function \(g\) that can be specified using only \(2^{n-1}\) values.
The specific construction of \(O_g\) depends on the measurement outcome \(b\), as detailed below.
Let \(x'\) denote an \((n-1)\)-bit string representing the last \(n-1\) bits of an \(n\)-bit string \(x\), such that \(x = 0 \doubleplus x'\) or \(x = 1 \doubleplus x'\).
There are four cases to consider.

\begin{enumerate}[label={\textbf{Case \arabic*:}}, wide, labelwidth=!, labelindent=2pt]
	\item \textbf{\(b = 0^n\):}
	      No correction is needed, as \(f(x \oplus 0^n) = f(x)\).
	      The process is complete.

	\item \textbf{\(b = 1 \doubleplus 0^{n-1}\) (i.e., \(b = 100\dots0\)):}
	      The state after resource consumption is effectively \[ \sum_{x' \in \{0,1\}^{n-1}} \left( c_{0 \doubleplus x'} \ket{0 \doubleplus x'}\ket{f(1 \doubleplus x')} + c_{1 \doubleplus x'} \ket{1 \doubleplus x'}\ket{f(0\doubleplus x')} \right).
	      \]

	      We define a correction function \(g(z) = f(0 \doubleplus z) \oplus f(1 \doubleplus z)\) for \(z \in \{0,1\}^{n-1}\).
	      Performing a QROM read \(O_g\) (where the input bits are the last \(n-1\) bits of the full input register) yields:
	      \begin{align}
		      \MoveEqLeft \sum_{x' \in \{0,1\}^{n-1}} \left( c_{0 \doubleplus x'} \ket{0 \doubleplus x'}\ket{f(1 \doubleplus x') \oplus g(x')} + c_{1 \doubleplus x'} \ket{1 \doubleplus x'} \ket{f(0 \doubleplus x') \oplus g(x')} \right)
		      \\
		       & = \sum_{x' \in \{0,1\}^{n-1}} \left( c_{0 \doubleplus x'} \ket{0 \doubleplus x'}\ket{f(0 \doubleplus x')} + c_{1 \doubleplus x'} \ket{1 \doubleplus x'}\ket{f(1 \doubleplus x')} \right)
		      \\
		       & = \sum_{x \in \{0,1\}^n} c_{x}\ket{x}\ket{f(x)}.
	      \end{align}
	\item \textbf{\(b = 1 \doubleplus b'\), where \(b' \in \{0,1\}^{n-1}\) and \(b' \neq 0^{n-1}\):}
	      This generalizes Case 2.
	      The state after resource consumption is:
	      \begin{equation}
		      \ket{\phi_{\text{init}}} = \sum_{x' \in \{0,1\}^{n-1}} \left( c_{0 \doubleplus x'} \ket{0 \doubleplus x'} \ket{f(1 \doubleplus (x'\oplus b'))} + c_{1 \doubleplus x'} \ket{1 \doubleplus x'}\ket{f(0\doubleplus (x'\oplus b'))} \right).
	      \end{equation}
	      The correction function is \(g(z) = f(0\doubleplus z) \oplus f(1 \doubleplus (z \oplus b'))\) for \(z \in \{0,1\}^{n-1}\).
	      The correction proceeds in steps:
	      \begin{align*}
		      \ket{\phi_{\text{init}}} & \xrightarrow{\text{Step 1}} \sum_{x' \in \{0,1\}^{n-1}} \left(
																							\begin{aligned}
																									& c_{0 \doubleplus x'} \ket{0 \doubleplus x'} \ket{f(1 \doubleplus (x'\oplus b'))}
																								\\
																								+ & c_{1 \doubleplus x'} \ket{1 \doubleplus (x'\oplus b')}\ket{f(0\doubleplus (x'\oplus b'))}
																							\end{aligned}
		      \right)
		      \\
		                               & \xrightarrow{\text{Step 2}} \sum_{x' \in \{0,1\}^{n-1}} \left(
																							\begin{aligned}
																									& c_{0 \doubleplus x'} \ket{0 \doubleplus x'} \ket{f(1 \doubleplus (x'\oplus b')) \oplus g(x')}
																								\\
																								+ & c_{1 \doubleplus x'} \ket{1 \doubleplus (x'\oplus b')}\ket{f(0\doubleplus (x'\oplus b')) \oplus g(x'\oplus b')}
																							\end{aligned}
		      \right)
		      \\
		                               & \xrightarrow{\text{Step 3}} \sum_{x' \in \{0,1\}^{n-1}} \left( c_{0 \doubleplus x'} \ket{0 \doubleplus x'} \ket{f(0 \doubleplus x')} + c_{1 \doubleplus x'} \ket{1 \doubleplus x'} \ket{f(1\doubleplus x')} \right)
		      \\
		                               & = \sum_{x \in \{0,1\}^n} c_x \ket{x}\ket{f(x)}.
	      \end{align*}
	      Step 1 conditionally applies CNOTs to XOR \(b'\) into the last \(n-1\) bits of the input register if the first bit is \(1\). Step 2 applies the QROM \(O_g\), controlled by the last \(n-1\) bits of the (potentially modified) input register. Step 3 substitutes the definition of \(g\) and uncomputes the conditional XOR from Step 1, restoring the input register to its original state.
	\item \textbf{\(b = 0 \doubleplus b'\), where \(b' \in \{0,1\}^{n-1}\) and \(b' \neq 0^{n-1}\):}
	      If \(b' = 0^{n-1}\), this reduces to Case 1.
	      Otherwise, we find any index \(i \in \{2, \dots, n\}\) such that the \(i\)-th bit of \(b\) is 1.
	      We can then relabel the qubits of the input register so that the \(i\)th qubit becomes the first qubit.
	      This permutation is also applied to the classical definition of \(f\) (to define \(f_{\text{perm}}\)) and to \(b\) (to get \(b_{\text{perm}}\)).
	      This transformation ensures that the first bit of \(b_{\text{perm}}\) is 1.
	      This relabeling is a classical pre-processing step that determines which \(2^{n-1}\)-entry QROM \(O_g\) to synthesize for \(f_{\text{perm}}\); it does not require any additional quantum operations during the correction phase itself.
	      The problem then reduces to Case 3, using \(f_{\text{perm}}\) and \(b_{\text{perm}}\).
\end{enumerate}
The cost of consuming the resource state and applying the correction is dominated by the cost of implementing \(O_g\), a QROM read with exactly half as many inputs as \(O_f\) (and therefore half the cost).

This strategy allows us to benefit when we are repeatedly implementing \(\bar{O}_f\), which queries \(f\) in superposition, subject to the requirement that the output register is in the \(\ket{0}\) state.
Operationally, we could define the action of this protocol on a more general state by resetting the output register to the zero state before proceeding.
Letting
\begin{equation}
	\ket{\Psi} = \sum_{i,j} c_{ij} \ket{x_i} \ket{\alpha_j},
\end{equation}
our protocol would implement the quantum channel
\begin{equation}
	\ketbra{\Psi} \rightarrow O_f \left( \tr_{\text{output}} \left[ \ketbra{\Psi} \right] \otimes \ketbra{0} \right) O_f^\dagger.
\end{equation}
This would not be sufficient for applications that require us to implement the map \(\ket{x}\ket{\alpha} \rightarrow \ket{x}\ket{\alpha \oplus f(x)}\) for arbitrary \(\alpha\), nor would it be sufficient for applications that require us to act unitarily.
\footnote{In some applications of QROM, the data loading oracle is only specified by its action on input states of the form \(\ket{x}\ket{0}\).
	In many of these cases though, there is also a (sometimes implicit) requirement that the data loading oracle act unitarily on an arbitrary input.
}

However, it is frequently possible to use \(\bar{O}_f\) to implement the more general
\begin{equation}
	O_f: \ket{x}\ket{\alpha} \rightarrow \ket{x}\ket{\alpha \oplus f(x)}
\end{equation}
with little to no additional overhead.
This is because we can always perform the lookup to an ancilla output register initialized in the \(\ket{0}\) state before XORing the result into the real output register, implementing the operation
\begin{equation}
	\tilde{O}_f: \ket{x}\ket{\alpha}\ket{0} \rightarrow \ket{x}\ket{\alpha \oplus f(x)} \ket{f(x)}.
\end{equation}
As we briefly discuss in \Cref{app:qrom_background}, we can use the measurement-based uncomputation techniques of \citen{Berry2019-qo} to uncompute the ancilla register with a cost that scales as \(\bigo{2^n}\) rather than \(\bigo{2^n m}\).
When \(m\) is large, this additional overhead is negligible.
In other cases, when we intend to temporarily use the value output by the QROM and then uncompute it anyway, it may be possible to fold the cost of uncomputing the junk register into the subsequent uncomputation step with no additional cost at all.
See, e.g., the use of QROM to implement a block encoding using a linear combination of unitaries approach, as we review in \Cref{app:lcu_overview}.

In the future, it would be interesting to investigate other applications for our resource state construction, and to understand whether or not it is possible to generalize it to reduce the cost of serial queries by more than a factor of \(\approx 2\).


\begin{thebibliography}{60}%
\makeatletter
\providecommand \@ifxundefined [1]{%
 \@ifx{#1\undefined}
}%
\providecommand \@ifnum [1]{%
 \ifnum #1\expandafter \@firstoftwo
 \else \expandafter \@secondoftwo
 \fi
}%
\providecommand \@ifx [1]{%
 \ifx #1\expandafter \@firstoftwo
 \else \expandafter \@secondoftwo
 \fi
}%
\providecommand \natexlab [1]{#1}%
\providecommand \enquote  [1]{``#1''}%
\providecommand \bibnamefont  [1]{#1}%
\providecommand \bibfnamefont [1]{#1}%
\providecommand \citenamefont [1]{#1}%
\providecommand \href@noop [0]{\@secondoftwo}%
\providecommand \href [0]{\begingroup \@sanitize@url \@href}%
\providecommand \@href[1]{\@@startlink{#1}\@@href}%
\providecommand \@@href[1]{\endgroup#1\@@endlink}%
\providecommand \@sanitize@url [0]{\catcode `\\12\catcode `\$12\catcode `\&12\catcode `\#12\catcode `\^12\catcode `\_12\catcode `\%12\relax}%
\providecommand \@@startlink[1]{}%
\providecommand \@@endlink[0]{}%
\providecommand \url  [0]{\begingroup\@sanitize@url \@url }%
\providecommand \@url [1]{\endgroup\@href {#1}{\urlprefix }}%
\providecommand \urlprefix  [0]{URL }%
\providecommand \Eprint [0]{\href }%
\providecommand \doibase [0]{https://doi.org/}%
\providecommand \selectlanguage [0]{\@gobble}%
\providecommand \bibinfo  [0]{\@secondoftwo}%
\providecommand \bibfield  [0]{\@secondoftwo}%
\providecommand \translation [1]{[#1]}%
\providecommand \BibitemOpen [0]{}%
\providecommand \bibitemStop [0]{}%
\providecommand \bibitemNoStop [0]{.\EOS\space}%
\providecommand \EOS [0]{\spacefactor3000\relax}%
\providecommand \BibitemShut  [1]{\csname bibitem#1\endcsname}%
\let\auto@bib@innerbib\@empty
%</preamble>
\bibitem [{\citenamefont {Ulig}(1974)}]{Ulig1974-xh}%
  \BibitemOpen
  \bibfield  {author} {\bibinfo {author} {\bibfnamefont {D.}~\bibnamefont {Ulig}},\ }\bibfield  {title} {\bibinfo {title} {On the synthesis of self-correcting schemes from functional elements with a small number of reliable elements},\ }\href {https://doi.org/10.1007/BF01152835} {\bibfield  {journal} {\bibinfo  {journal} {Mathematical notes of the Academy of Sciences of the USSR}\ }\textbf {\bibinfo {volume} {15}},\ \bibinfo {pages} {558} (\bibinfo {year} {1974})}\BibitemShut {NoStop}%
\bibitem [{\citenamefont {Uhlig}(1992)}]{Uhlig1992-tg}%
  \BibitemOpen
  \bibfield  {author} {\bibinfo {author} {\bibfnamefont {D.}~\bibnamefont {Uhlig}},\ }\bibfield  {title} {\bibinfo {title} {Networks computing boolean functions for multiple input values},\ }in\ \href {https://doi.org/10.1017/CBO9780511526633.013} {\emph {\bibinfo {booktitle} {Boolean Function Complexity}}}\ (\bibinfo  {publisher} {Cambridge University Press},\ \bibinfo {year} {1992})\ pp.\ \bibinfo {pages} {165--173}\BibitemShut {NoStop}%
\bibitem [{\citenamefont {Kretschmer}(2022)}]{Kretschmer2022-uy}%
  \BibitemOpen
  \bibfield  {author} {\bibinfo {author} {\bibfnamefont {W.}~\bibnamefont {Kretschmer}},\ }\bibfield  {title} {\bibinfo {title} {Quantum mass production theorems},\ }\Eprint {https://arxiv.org/abs/2212.14399} {arXiv:2212.14399 [quant-ph]}  (\bibinfo {year} {2022})\BibitemShut {NoStop}%
\bibitem [{\citenamefont {von Burg}\ \emph {et~al.}(2021)\citenamefont {von Burg}, \citenamefont {Low}, \citenamefont {Häner}, \citenamefont {Steiger}, \citenamefont {Reiher}, \citenamefont {Roetteler},\ and\ \citenamefont {Troyer}}]{von-Burg2021-yq}%
  \BibitemOpen
  \bibfield  {author} {\bibinfo {author} {\bibfnamefont {V.}~\bibnamefont {von Burg}}, \bibinfo {author} {\bibfnamefont {G.~H.}\ \bibnamefont {Low}}, \bibinfo {author} {\bibfnamefont {T.}~\bibnamefont {Häner}}, \bibinfo {author} {\bibfnamefont {D.~S.}\ \bibnamefont {Steiger}}, \bibinfo {author} {\bibfnamefont {M.}~\bibnamefont {Reiher}}, \bibinfo {author} {\bibfnamefont {M.}~\bibnamefont {Roetteler}},\ and\ \bibinfo {author} {\bibfnamefont {M.}~\bibnamefont {Troyer}},\ }\bibfield  {title} {{\selectlanguage {en}\bibinfo {title} {Quantum computing enhanced computational catalysis}},\ }\href {https://doi.org/10.1103/physrevresearch.3.033055} {\bibfield  {journal} {\bibinfo  {journal} {Phys. Rev. Res.}\ }\textbf {\bibinfo {volume} {3}},\ \bibinfo {pages} {033055} (\bibinfo {year} {2021})}\BibitemShut {NoStop}%
\bibitem [{\citenamefont {Lee}\ \emph {et~al.}(2021)\citenamefont {Lee}, \citenamefont {Berry}, \citenamefont {Gidney}, \citenamefont {Huggins}, \citenamefont {McClean}, \citenamefont {Wiebe},\ and\ \citenamefont {Babbush}}]{Lee2021-su}%
  \BibitemOpen
  \bibfield  {author} {\bibinfo {author} {\bibfnamefont {J.}~\bibnamefont {Lee}}, \bibinfo {author} {\bibfnamefont {D.~W.}\ \bibnamefont {Berry}}, \bibinfo {author} {\bibfnamefont {C.}~\bibnamefont {Gidney}}, \bibinfo {author} {\bibfnamefont {W.~J.}\ \bibnamefont {Huggins}}, \bibinfo {author} {\bibfnamefont {J.~R.}\ \bibnamefont {McClean}}, \bibinfo {author} {\bibfnamefont {N.}~\bibnamefont {Wiebe}},\ and\ \bibinfo {author} {\bibfnamefont {R.}~\bibnamefont {Babbush}},\ }\bibfield  {title} {{\selectlanguage {en}\bibinfo {title} {Even more efficient quantum computations of chemistry through tensor hypercontraction}},\ }\href {https://doi.org/10.1103/prxquantum.2.030305} {\bibfield  {journal} {\bibinfo  {journal} {PRX quantum}\ }\textbf {\bibinfo {volume} {2}},\ \bibinfo {pages} {030305} (\bibinfo {year} {2021})}\BibitemShut {NoStop}%
\bibitem [{\citenamefont {Caesura}\ \emph {et~al.}(2025)\citenamefont {Caesura}, \citenamefont {Cortes}, \citenamefont {Pol}, \citenamefont {Sim}, \citenamefont {Steudtner}, \citenamefont {Anselmetti}, \citenamefont {Degroote}, \citenamefont {Moll}, \citenamefont {Santagati}, \citenamefont {Streif},\ and\ \citenamefont {Tautermann}}]{Caesura2025-wl}%
  \BibitemOpen
  \bibfield  {author} {\bibinfo {author} {\bibfnamefont {A.}~\bibnamefont {Caesura}}, \bibinfo {author} {\bibfnamefont {C.~L.}\ \bibnamefont {Cortes}}, \bibinfo {author} {\bibfnamefont {W.}~\bibnamefont {Pol}}, \bibinfo {author} {\bibfnamefont {S.}~\bibnamefont {Sim}}, \bibinfo {author} {\bibfnamefont {M.}~\bibnamefont {Steudtner}}, \bibinfo {author} {\bibfnamefont {G.-L.~R.}\ \bibnamefont {Anselmetti}}, \bibinfo {author} {\bibfnamefont {M.}~\bibnamefont {Degroote}}, \bibinfo {author} {\bibfnamefont {N.}~\bibnamefont {Moll}}, \bibinfo {author} {\bibfnamefont {R.}~\bibnamefont {Santagati}}, \bibinfo {author} {\bibfnamefont {M.}~\bibnamefont {Streif}},\ and\ \bibinfo {author} {\bibfnamefont {C.~S.}\ \bibnamefont {Tautermann}},\ }\bibfield  {title} {\bibinfo {title} {Faster quantum chemistry simulations on a quantum computer with improved tensor factorization and active volume compilation},\ }\Eprint {https://arxiv.org/abs/2501.06165} {arXiv:2501.06165 [quant-ph]}  (\bibinfo {year} {2025})\BibitemShut {NoStop}%
\bibitem [{\citenamefont {Low}\ \emph {et~al.}(2025)\citenamefont {Low}, \citenamefont {King}, \citenamefont {Berry}, \citenamefont {Han}, \citenamefont {DePrince}, \citenamefont {White}, \citenamefont {Babbush}, \citenamefont {Somma},\ and\ \citenamefont {Rubin}}]{Low2025-ei}%
  \BibitemOpen
  \bibfield  {author} {\bibinfo {author} {\bibfnamefont {G.~H.}\ \bibnamefont {Low}}, \bibinfo {author} {\bibfnamefont {R.}~\bibnamefont {King}}, \bibinfo {author} {\bibfnamefont {D.~W.}\ \bibnamefont {Berry}}, \bibinfo {author} {\bibfnamefont {Q.}~\bibnamefont {Han}}, \bibinfo {author} {\bibfnamefont {A.~E.}\ \bibnamefont {DePrince}, \bibfnamefont {III}}, \bibinfo {author} {\bibfnamefont {A.}~\bibnamefont {White}}, \bibinfo {author} {\bibfnamefont {R.}~\bibnamefont {Babbush}}, \bibinfo {author} {\bibfnamefont {R.~D.}\ \bibnamefont {Somma}},\ and\ \bibinfo {author} {\bibfnamefont {N.~C.}\ \bibnamefont {Rubin}},\ }\bibfield  {title} {\bibinfo {title} {Fast quantum simulation of electronic structure by spectrum amplification},\ }\Eprint {https://arxiv.org/abs/2502.15882} {arXiv:2502.15882 [quant-ph]}  (\bibinfo {year} {2025})\BibitemShut {NoStop}%
\bibitem [{\citenamefont {Babbush}\ \emph {et~al.}(2021)\citenamefont {Babbush}, \citenamefont {McClean}, \citenamefont {Newman}, \citenamefont {Gidney}, \citenamefont {Boixo},\ and\ \citenamefont {Neven}}]{Babbush2021-aq}%
  \BibitemOpen
  \bibfield  {author} {\bibinfo {author} {\bibfnamefont {R.}~\bibnamefont {Babbush}}, \bibinfo {author} {\bibfnamefont {J.~R.}\ \bibnamefont {McClean}}, \bibinfo {author} {\bibfnamefont {M.}~\bibnamefont {Newman}}, \bibinfo {author} {\bibfnamefont {C.}~\bibnamefont {Gidney}}, \bibinfo {author} {\bibfnamefont {S.}~\bibnamefont {Boixo}},\ and\ \bibinfo {author} {\bibfnamefont {H.}~\bibnamefont {Neven}},\ }\bibfield  {title} {{\selectlanguage {en}\bibinfo {title} {Focus beyond quadratic speedups for error-corrected quantum advantage}},\ }\href {https://doi.org/10.1103/prxquantum.2.010103} {\bibfield  {journal} {\bibinfo  {journal} {PRX quantum}\ }\textbf {\bibinfo {volume} {2}},\ \bibinfo {pages} {010103} (\bibinfo {year} {2021})}\BibitemShut {NoStop}%
\bibitem [{\citenamefont {Beverland}\ \emph {et~al.}(2022)\citenamefont {Beverland}, \citenamefont {Murali}, \citenamefont {Troyer}, \citenamefont {Svore}, \citenamefont {Hoefler}, \citenamefont {Kliuchnikov}, \citenamefont {Low}, \citenamefont {Soeken}, \citenamefont {Sundaram},\ and\ \citenamefont {Vaschillo}}]{Beverland2022-dh}%
  \BibitemOpen
  \bibfield  {author} {\bibinfo {author} {\bibfnamefont {M.~E.}\ \bibnamefont {Beverland}}, \bibinfo {author} {\bibfnamefont {P.}~\bibnamefont {Murali}}, \bibinfo {author} {\bibfnamefont {M.}~\bibnamefont {Troyer}}, \bibinfo {author} {\bibfnamefont {K.~M.}\ \bibnamefont {Svore}}, \bibinfo {author} {\bibfnamefont {T.}~\bibnamefont {Hoefler}}, \bibinfo {author} {\bibfnamefont {V.}~\bibnamefont {Kliuchnikov}}, \bibinfo {author} {\bibfnamefont {G.~H.}\ \bibnamefont {Low}}, \bibinfo {author} {\bibfnamefont {M.}~\bibnamefont {Soeken}}, \bibinfo {author} {\bibfnamefont {A.}~\bibnamefont {Sundaram}},\ and\ \bibinfo {author} {\bibfnamefont {A.}~\bibnamefont {Vaschillo}},\ }\bibfield  {title} {\bibinfo {title} {Assessing requirements to scale to practical quantum advantage},\ }\Eprint {https://arxiv.org/abs/2211.07629} {arXiv:2211.07629 [quant-ph]}  (\bibinfo {year} {2022})\BibitemShut {NoStop}%
\bibitem [{\citenamefont {Babbush}\ \emph {et~al.}(2018)\citenamefont {Babbush}, \citenamefont {Gidney}, \citenamefont {Berry}, \citenamefont {Wiebe}, \citenamefont {McClean}, \citenamefont {Paler}, \citenamefont {Fowler},\ and\ \citenamefont {Neven}}]{Babbush2018-tb}%
  \BibitemOpen
  \bibfield  {author} {\bibinfo {author} {\bibfnamefont {R.}~\bibnamefont {Babbush}}, \bibinfo {author} {\bibfnamefont {C.}~\bibnamefont {Gidney}}, \bibinfo {author} {\bibfnamefont {D.~W.}\ \bibnamefont {Berry}}, \bibinfo {author} {\bibfnamefont {N.}~\bibnamefont {Wiebe}}, \bibinfo {author} {\bibfnamefont {J.}~\bibnamefont {McClean}}, \bibinfo {author} {\bibfnamefont {A.}~\bibnamefont {Paler}}, \bibinfo {author} {\bibfnamefont {A.}~\bibnamefont {Fowler}},\ and\ \bibinfo {author} {\bibfnamefont {H.}~\bibnamefont {Neven}},\ }\bibfield  {title} {\bibinfo {title} {Encoding electronic spectra in quantum circuits with linear {T} complexity},\ }\href {https://doi.org/10.1103/PhysRevX.8.041015} {\bibfield  {journal} {\bibinfo  {journal} {Phys. Rev. X}\ }\textbf {\bibinfo {volume} {8}},\ \bibinfo {pages} {041015} (\bibinfo {year} {2018})}\BibitemShut {NoStop}%
\bibitem [{\citenamefont {Low}\ \emph {et~al.}(2018)\citenamefont {Low}, \citenamefont {Kliuchnikov},\ and\ \citenamefont {Schaeffer}}]{Low2018-uu}%
  \BibitemOpen
  \bibfield  {author} {\bibinfo {author} {\bibfnamefont {G.~H.}\ \bibnamefont {Low}}, \bibinfo {author} {\bibfnamefont {V.}~\bibnamefont {Kliuchnikov}},\ and\ \bibinfo {author} {\bibfnamefont {L.}~\bibnamefont {Schaeffer}},\ }\bibfield  {title} {\bibinfo {title} {Trading {T} gates for dirty qubits in state preparation and unitary synthesis},\ }\Eprint {https://arxiv.org/abs/1812.00954} {arXiv:1812.00954 [quant-ph]}  (\bibinfo {year} {2018})\BibitemShut {NoStop}%
\bibitem [{\citenamefont {Berry}\ \emph {et~al.}(2019)\citenamefont {Berry}, \citenamefont {Gidney}, \citenamefont {Motta}, \citenamefont {McClean},\ and\ \citenamefont {Babbush}}]{Berry2019-qo}%
  \BibitemOpen
  \bibfield  {author} {\bibinfo {author} {\bibfnamefont {D.~W.}\ \bibnamefont {Berry}}, \bibinfo {author} {\bibfnamefont {C.}~\bibnamefont {Gidney}}, \bibinfo {author} {\bibfnamefont {M.}~\bibnamefont {Motta}}, \bibinfo {author} {\bibfnamefont {J.~R.}\ \bibnamefont {McClean}},\ and\ \bibinfo {author} {\bibfnamefont {R.}~\bibnamefont {Babbush}},\ }\bibfield  {title} {{\selectlanguage {en}\bibinfo {title} {Qubitization of arbitrary basis quantum chemistry leveraging sparsity and low rank factorization}},\ }\href {https://doi.org/10.22331/q-2019-12-02-208} {\bibfield  {journal} {\bibinfo  {journal} {Quantum}\ }\textbf {\bibinfo {volume} {3}},\ \bibinfo {pages} {208} (\bibinfo {year} {2019})}\BibitemShut {NoStop}%
\bibitem [{\citenamefont {Zhu}\ \emph {et~al.}(2024)\citenamefont {Zhu}, \citenamefont {Sundaram},\ and\ \citenamefont {Low}}]{Zhu2024-pe}%
  \BibitemOpen
  \bibfield  {author} {\bibinfo {author} {\bibfnamefont {S.}~\bibnamefont {Zhu}}, \bibinfo {author} {\bibfnamefont {A.}~\bibnamefont {Sundaram}},\ and\ \bibinfo {author} {\bibfnamefont {G.~H.}\ \bibnamefont {Low}},\ }\bibfield  {title} {\bibinfo {title} {Unified architecture for a quantum lookup table},\ }\Eprint {https://arxiv.org/abs/2406.18030} {arXiv:2406.18030 [quant-ph]}  (\bibinfo {year} {2024})\BibitemShut {NoStop}%
\bibitem [{\citenamefont {Giovannetti}\ \emph {et~al.}(2008)\citenamefont {Giovannetti}, \citenamefont {Lloyd},\ and\ \citenamefont {Maccone}}]{Giovannetti2008-un}%
  \BibitemOpen
  \bibfield  {author} {\bibinfo {author} {\bibfnamefont {V.}~\bibnamefont {Giovannetti}}, \bibinfo {author} {\bibfnamefont {S.}~\bibnamefont {Lloyd}},\ and\ \bibinfo {author} {\bibfnamefont {L.}~\bibnamefont {Maccone}},\ }\bibfield  {title} {\bibinfo {title} {Architectures for a quantum random access memory},\ }\Eprint {https://arxiv.org/abs/0807.4994} {arXiv:0807.4994 [quant-ph]}  (\bibinfo {year} {2008})\BibitemShut {NoStop}%
\bibitem [{\citenamefont {Giovannetti}\ \emph {et~al.}(2007)\citenamefont {Giovannetti}, \citenamefont {Lloyd},\ and\ \citenamefont {Maccone}}]{Giovannetti2007-cg}%
  \BibitemOpen
  \bibfield  {author} {\bibinfo {author} {\bibfnamefont {V.}~\bibnamefont {Giovannetti}}, \bibinfo {author} {\bibfnamefont {S.}~\bibnamefont {Lloyd}},\ and\ \bibinfo {author} {\bibfnamefont {L.}~\bibnamefont {Maccone}},\ }\bibfield  {title} {\bibinfo {title} {Quantum random access memory},\ }\Eprint {https://arxiv.org/abs/0708.1879} {arXiv:0708.1879 [quant-ph]}  (\bibinfo {year} {2007})\BibitemShut {NoStop}%
\bibitem [{\citenamefont {Arunachalam}\ \emph {et~al.}(2015)\citenamefont {Arunachalam}, \citenamefont {Gheorghiu}, \citenamefont {Jochym-O'Connor}, \citenamefont {Mosca},\ and\ \citenamefont {Srinivasan}}]{Arunachalam2015-vr}%
  \BibitemOpen
  \bibfield  {author} {\bibinfo {author} {\bibfnamefont {S.}~\bibnamefont {Arunachalam}}, \bibinfo {author} {\bibfnamefont {V.}~\bibnamefont {Gheorghiu}}, \bibinfo {author} {\bibfnamefont {T.}~\bibnamefont {Jochym-O'Connor}}, \bibinfo {author} {\bibfnamefont {M.}~\bibnamefont {Mosca}},\ and\ \bibinfo {author} {\bibfnamefont {P.~V.}\ \bibnamefont {Srinivasan}},\ }\bibfield  {title} {\bibinfo {title} {On the robustness of bucket brigade quantum {RAM}},\ }\href {https://doi.org/10.1088/1367-2630/17/12/123010} {\bibfield  {journal} {\bibinfo  {journal} {New J. Phys.}\ }\textbf {\bibinfo {volume} {17}},\ \bibinfo {pages} {123010} (\bibinfo {year} {2015})}\BibitemShut {NoStop}%
\bibitem [{\citenamefont {Jaques}\ and\ \citenamefont {Rattew}(2023)}]{Jaques2023-fo}%
  \BibitemOpen
  \bibfield  {author} {\bibinfo {author} {\bibfnamefont {S.}~\bibnamefont {Jaques}}\ and\ \bibinfo {author} {\bibfnamefont {A.~G.}\ \bibnamefont {Rattew}},\ }\bibfield  {title} {\bibinfo {title} {{QRAM}: A survey and critique},\ }\Eprint {https://arxiv.org/abs/2305.10310} {arXiv:2305.10310 [quant-ph]}  (\bibinfo {year} {2023})\BibitemShut {NoStop}%
\bibitem [{\citenamefont {Bravyi}\ and\ \citenamefont {Kitaev}(2005)}]{Bravyi2005-vi}%
  \BibitemOpen
  \bibfield  {author} {\bibinfo {author} {\bibfnamefont {S.}~\bibnamefont {Bravyi}}\ and\ \bibinfo {author} {\bibfnamefont {A.}~\bibnamefont {Kitaev}},\ }\bibfield  {title} {\bibinfo {title} {Universal quantum computation with ideal clifford gates and noisy ancillas},\ }\href {https://doi.org/10.1103/physreva.71.022316} {\bibfield  {journal} {\bibinfo  {journal} {Phys. Rev. A}\ }\textbf {\bibinfo {volume} {71}},\ \bibinfo {pages} {022316} (\bibinfo {year} {2005})}\BibitemShut {NoStop}%
\bibitem [{\citenamefont {Bravyi}\ and\ \citenamefont {Haah}(2012)}]{Bravyi2012-mg}%
  \BibitemOpen
  \bibfield  {author} {\bibinfo {author} {\bibfnamefont {S.}~\bibnamefont {Bravyi}}\ and\ \bibinfo {author} {\bibfnamefont {J.}~\bibnamefont {Haah}},\ }\bibfield  {title} {{\selectlanguage {en}\bibinfo {title} {Magic-state distillation with low overhead}},\ }\href {https://doi.org/10.1103/PhysRevA.86.052329} {\bibfield  {journal} {\bibinfo  {journal} {Phys. Rev. A}\ }\textbf {\bibinfo {volume} {86}},\ \bibinfo {pages} {052329} (\bibinfo {year} {2012})}\BibitemShut {NoStop}%
\bibitem [{\citenamefont {Fowler}\ and\ \citenamefont {Gidney}(2018)}]{Fowler2018-he}%
  \BibitemOpen
  \bibfield  {author} {\bibinfo {author} {\bibfnamefont {A.~G.}\ \bibnamefont {Fowler}}\ and\ \bibinfo {author} {\bibfnamefont {C.}~\bibnamefont {Gidney}},\ }\bibfield  {title} {\bibinfo {title} {Low overhead quantum computation using lattice surgery},\ }\Eprint {https://arxiv.org/abs/1808.06709} {arXiv:1808.06709 [quant-ph]}  (\bibinfo {year} {2018})\BibitemShut {NoStop}%
\bibitem [{\citenamefont {Litinski}(2019)}]{Litinski2019-ek}%
  \BibitemOpen
  \bibfield  {author} {\bibinfo {author} {\bibfnamefont {D.}~\bibnamefont {Litinski}},\ }\bibfield  {title} {{\selectlanguage {en}\bibinfo {title} {Magic state distillation: Not as costly as you think}},\ }\href {https://doi.org/10.22331/q-2019-12-02-205} {\bibfield  {journal} {\bibinfo  {journal} {Quantum}\ }\textbf {\bibinfo {volume} {3}},\ \bibinfo {pages} {205} (\bibinfo {year} {2019})}\BibitemShut {NoStop}%
\bibitem [{\citenamefont {Gidney}\ \emph {et~al.}(2024)\citenamefont {Gidney}, \citenamefont {Shutty},\ and\ \citenamefont {Jones}}]{gidney2024magic}%
  \BibitemOpen
  \bibfield  {author} {\bibinfo {author} {\bibfnamefont {C.}~\bibnamefont {Gidney}}, \bibinfo {author} {\bibfnamefont {N.}~\bibnamefont {Shutty}},\ and\ \bibinfo {author} {\bibfnamefont {C.}~\bibnamefont {Jones}},\ }\bibfield  {title} {\bibinfo {title} {Magic state cultivation: growing t states as cheap as cnot gates},\ }\href@noop {} {\bibfield  {journal} {\bibinfo  {journal} {arXiv preprint arXiv:2409.17595}\ } (\bibinfo {year} {2024})}\BibitemShut {NoStop}%
\bibitem [{\citenamefont {Harrigan}\ \emph {et~al.}(2024)\citenamefont {Harrigan}, \citenamefont {Khattar}, \citenamefont {Yuan}, \citenamefont {Peduri}, \citenamefont {Yosri}, \citenamefont {Malone}, \citenamefont {Babbush},\ and\ \citenamefont {Rubin}}]{Harrigan2024-rj}%
  \BibitemOpen
  \bibfield  {author} {\bibinfo {author} {\bibfnamefont {M.~P.}\ \bibnamefont {Harrigan}}, \bibinfo {author} {\bibfnamefont {T.}~\bibnamefont {Khattar}}, \bibinfo {author} {\bibfnamefont {C.}~\bibnamefont {Yuan}}, \bibinfo {author} {\bibfnamefont {A.}~\bibnamefont {Peduri}}, \bibinfo {author} {\bibfnamefont {N.}~\bibnamefont {Yosri}}, \bibinfo {author} {\bibfnamefont {F.~D.}\ \bibnamefont {Malone}}, \bibinfo {author} {\bibfnamefont {R.}~\bibnamefont {Babbush}},\ and\ \bibinfo {author} {\bibfnamefont {N.~C.}\ \bibnamefont {Rubin}},\ }\bibfield  {title} {\bibinfo {title} {Expressing and analyzing quantum algorithms with qualtran},\ }\Eprint {https://arxiv.org/abs/2409.04643} {arXiv:2409.04643 [quant-ph]}  (\bibinfo {year} {2024})\BibitemShut {NoStop}%
\bibitem [{\citenamefont {Craig}\ \emph {et~al.}(2024)\citenamefont {Craig}, \citenamefont {Noah},\ and\ \citenamefont {Cody}}]{Craig2024-kv}%
  \BibitemOpen
  \bibfield  {author} {\bibinfo {author} {\bibfnamefont {G.}~\bibnamefont {Craig}}, \bibinfo {author} {\bibfnamefont {S.}~\bibnamefont {Noah}},\ and\ \bibinfo {author} {\bibfnamefont {J.}~\bibnamefont {Cody}},\ }\bibfield  {title} {\bibinfo {title} {Magic state cultivation: growing {T} states as cheap as {CNOT} gates},\ }\Eprint {https://arxiv.org/abs/2409.17595} {arXiv:2409.17595 [quant-ph]}  (\bibinfo {year} {2024})\BibitemShut {NoStop}%
\bibitem [{\citenamefont {Harrow}\ \emph {et~al.}(2009)\citenamefont {Harrow}, \citenamefont {Hassidim},\ and\ \citenamefont {Lloyd}}]{Harrow2009-sm}%
  \BibitemOpen
  \bibfield  {author} {\bibinfo {author} {\bibfnamefont {A.~W.}\ \bibnamefont {Harrow}}, \bibinfo {author} {\bibfnamefont {A.}~\bibnamefont {Hassidim}},\ and\ \bibinfo {author} {\bibfnamefont {S.}~\bibnamefont {Lloyd}},\ }\bibfield  {title} {{\selectlanguage {en}\bibinfo {title} {Quantum algorithm for linear systems of equations}},\ }\href {https://doi.org/10.1103/PhysRevLett.103.150502} {\bibfield  {journal} {\bibinfo  {journal} {Phys. Rev. Lett.}\ }\textbf {\bibinfo {volume} {103}},\ \bibinfo {pages} {150502} (\bibinfo {year} {2009})}\BibitemShut {NoStop}%
\bibitem [{\citenamefont {Childs}\ \emph {et~al.}(2017)\citenamefont {Childs}, \citenamefont {Kothari},\ and\ \citenamefont {Somma}}]{Childs2017-is}%
  \BibitemOpen
  \bibfield  {author} {\bibinfo {author} {\bibfnamefont {A.~M.}\ \bibnamefont {Childs}}, \bibinfo {author} {\bibfnamefont {R.}~\bibnamefont {Kothari}},\ and\ \bibinfo {author} {\bibfnamefont {R.~D.}\ \bibnamefont {Somma}},\ }\bibfield  {title} {\bibinfo {title} {Quantum algorithm for systems of linear equations with exponentially improved dependence on precision},\ }\href {https://doi.org/10.1137/16M1087072} {\bibfield  {journal} {\bibinfo  {journal} {SIAM J. Comput.}\ }\textbf {\bibinfo {volume} {46}},\ \bibinfo {pages} {1920} (\bibinfo {year} {2017})}\BibitemShut {NoStop}%
\bibitem [{\citenamefont {Berry}(2014)}]{Berry2014-fz}%
  \BibitemOpen
  \bibfield  {author} {\bibinfo {author} {\bibfnamefont {D.~W.}\ \bibnamefont {Berry}},\ }\bibfield  {title} {{\selectlanguage {en}\bibinfo {title} {High-order quantum algorithm for solving linear differential equations}},\ }\href {https://doi.org/10.1088/1751-8113/47/10/105301} {\bibfield  {journal} {\bibinfo  {journal} {J. Phys. A Math. Theor.}\ }\textbf {\bibinfo {volume} {47}},\ \bibinfo {pages} {105301} (\bibinfo {year} {2014})}\BibitemShut {NoStop}%
\bibitem [{\citenamefont {Berry}\ \emph {et~al.}(2017)\citenamefont {Berry}, \citenamefont {Childs}, \citenamefont {Ostrander},\ and\ \citenamefont {Wang}}]{Berry2017-pd}%
  \BibitemOpen
  \bibfield  {author} {\bibinfo {author} {\bibfnamefont {D.~W.}\ \bibnamefont {Berry}}, \bibinfo {author} {\bibfnamefont {A.~M.}\ \bibnamefont {Childs}}, \bibinfo {author} {\bibfnamefont {A.}~\bibnamefont {Ostrander}},\ and\ \bibinfo {author} {\bibfnamefont {G.}~\bibnamefont {Wang}},\ }\bibfield  {title} {{\selectlanguage {en}\bibinfo {title} {Quantum algorithm for linear differential equations with exponentially improved dependence on precision}},\ }\href {https://doi.org/10.1007/s00220-017-3002-y} {\bibfield  {journal} {\bibinfo  {journal} {Commun. Math. Phys.}\ }\textbf {\bibinfo {volume} {356}},\ \bibinfo {pages} {1057} (\bibinfo {year} {2017})}\BibitemShut {NoStop}%
\bibitem [{\citenamefont {Liu}\ \emph {et~al.}(2021)\citenamefont {Liu}, \citenamefont {Kolden}, \citenamefont {Krovi}, \citenamefont {Loureiro}, \citenamefont {Trivisa},\ and\ \citenamefont {Childs}}]{Liu2021-ud}%
  \BibitemOpen
  \bibfield  {author} {\bibinfo {author} {\bibfnamefont {J.-P.}\ \bibnamefont {Liu}}, \bibinfo {author} {\bibfnamefont {H.~O.}\ \bibnamefont {Kolden}}, \bibinfo {author} {\bibfnamefont {H.~K.}\ \bibnamefont {Krovi}}, \bibinfo {author} {\bibfnamefont {N.~F.}\ \bibnamefont {Loureiro}}, \bibinfo {author} {\bibfnamefont {K.}~\bibnamefont {Trivisa}},\ and\ \bibinfo {author} {\bibfnamefont {A.~M.}\ \bibnamefont {Childs}},\ }\bibfield  {title} {{\selectlanguage {en}\bibinfo {title} {Efficient quantum algorithm for dissipative nonlinear differential equations}},\ }\href {https://doi.org/10.1073/pnas.2026805118} {\bibfield  {journal} {\bibinfo  {journal} {Proc. Natl. Acad. Sci. U. S. A.}\ }\textbf {\bibinfo {volume} {118}},\ \bibinfo {pages} {e2026805118} (\bibinfo {year} {2021})}\BibitemShut {NoStop}%
\bibitem [{\citenamefont {Lin}\ and\ \citenamefont {Tong}(2020)}]{Lin2020-qh}%
  \BibitemOpen
  \bibfield  {author} {\bibinfo {author} {\bibfnamefont {L.}~\bibnamefont {Lin}}\ and\ \bibinfo {author} {\bibfnamefont {Y.}~\bibnamefont {Tong}},\ }\bibfield  {title} {{\selectlanguage {en}\bibinfo {title} {Near-optimal ground state preparation}},\ }\href {https://doi.org/10.22331/q-2020-12-14-372} {\bibfield  {journal} {\bibinfo  {journal} {Quantum}\ }\textbf {\bibinfo {volume} {4}},\ \bibinfo {pages} {372} (\bibinfo {year} {2020})}\BibitemShut {NoStop}%
\bibitem [{\citenamefont {Lloyd}\ \emph {et~al.}(2016)\citenamefont {Lloyd}, \citenamefont {Garnerone},\ and\ \citenamefont {Zanardi}}]{Lloyd2016-po}%
  \BibitemOpen
  \bibfield  {author} {\bibinfo {author} {\bibfnamefont {S.}~\bibnamefont {Lloyd}}, \bibinfo {author} {\bibfnamefont {S.}~\bibnamefont {Garnerone}},\ and\ \bibinfo {author} {\bibfnamefont {P.}~\bibnamefont {Zanardi}},\ }\bibfield  {title} {{\selectlanguage {en}\bibinfo {title} {Quantum algorithms for topological and geometric analysis of data}},\ }\href {https://doi.org/10.1038/ncomms10138} {\bibfield  {journal} {\bibinfo  {journal} {Nat. Commun.}\ }\textbf {\bibinfo {volume} {7}},\ \bibinfo {pages} {10138} (\bibinfo {year} {2016})}\BibitemShut {NoStop}%
\bibitem [{\citenamefont {Liu}\ \emph {et~al.}(2024)\citenamefont {Liu}, \citenamefont {Liu}, \citenamefont {Liu}, \citenamefont {Ye}, \citenamefont {Wang}, \citenamefont {Alexeev}, \citenamefont {Eisert},\ and\ \citenamefont {Jiang}}]{Liu2024-dv}%
  \BibitemOpen
  \bibfield  {author} {\bibinfo {author} {\bibfnamefont {J.}~\bibnamefont {Liu}}, \bibinfo {author} {\bibfnamefont {M.}~\bibnamefont {Liu}}, \bibinfo {author} {\bibfnamefont {J.-P.}\ \bibnamefont {Liu}}, \bibinfo {author} {\bibfnamefont {Z.}~\bibnamefont {Ye}}, \bibinfo {author} {\bibfnamefont {Y.}~\bibnamefont {Wang}}, \bibinfo {author} {\bibfnamefont {Y.}~\bibnamefont {Alexeev}}, \bibinfo {author} {\bibfnamefont {J.}~\bibnamefont {Eisert}},\ and\ \bibinfo {author} {\bibfnamefont {L.}~\bibnamefont {Jiang}},\ }\bibfield  {title} {{\selectlanguage {en}\bibinfo {title} {Towards provably efficient quantum algorithms for large-scale machine-learning models}},\ }\href {https://doi.org/10.1038/s41467-023-43957-x} {\bibfield  {journal} {\bibinfo  {journal} {Nat. Commun.}\ }\textbf {\bibinfo {volume} {15}},\ \bibinfo {pages} {434} (\bibinfo {year} {2024})}\BibitemShut {NoStop}%
\bibitem [{\citenamefont {Su}\ \emph {et~al.}(2021)\citenamefont {Su}, \citenamefont {Berry}, \citenamefont {Wiebe}, \citenamefont {Rubin},\ and\ \citenamefont {Babbush}}]{su2021fault}%
  \BibitemOpen
  \bibfield  {author} {\bibinfo {author} {\bibfnamefont {Y.}~\bibnamefont {Su}}, \bibinfo {author} {\bibfnamefont {D.~W.}\ \bibnamefont {Berry}}, \bibinfo {author} {\bibfnamefont {N.}~\bibnamefont {Wiebe}}, \bibinfo {author} {\bibfnamefont {N.}~\bibnamefont {Rubin}},\ and\ \bibinfo {author} {\bibfnamefont {R.}~\bibnamefont {Babbush}},\ }\bibfield  {title} {\bibinfo {title} {Fault-tolerant quantum simulations of chemistry in first quantization},\ }\href@noop {} {\bibfield  {journal} {\bibinfo  {journal} {PRX Quantum}\ }\textbf {\bibinfo {volume} {2}},\ \bibinfo {pages} {040332} (\bibinfo {year} {2021})}\BibitemShut {NoStop}%
\bibitem [{\citenamefont {Huggins}\ \emph {et~al.}(2024)\citenamefont {Huggins}, \citenamefont {Leimkuhler}, \citenamefont {Stetina},\ and\ \citenamefont {Whaley}}]{huggins2024efficient}%
  \BibitemOpen
  \bibfield  {author} {\bibinfo {author} {\bibfnamefont {W.~J.}\ \bibnamefont {Huggins}}, \bibinfo {author} {\bibfnamefont {O.}~\bibnamefont {Leimkuhler}}, \bibinfo {author} {\bibfnamefont {T.~F.}\ \bibnamefont {Stetina}},\ and\ \bibinfo {author} {\bibfnamefont {K.~B.}\ \bibnamefont {Whaley}},\ }\bibfield  {title} {\bibinfo {title} {Efficient state preparation for the quantum simulation of molecules in first quantization},\ }\href@noop {} {\bibfield  {journal} {\bibinfo  {journal} {arXiv preprint arXiv:2407.00249}\ } (\bibinfo {year} {2024})}\BibitemShut {NoStop}%
\bibitem [{\citenamefont {Berry}\ \emph {et~al.}(2015)\citenamefont {Berry}, \citenamefont {Childs}, \citenamefont {Cleve}, \citenamefont {Kothari},\ and\ \citenamefont {Somma}}]{Berry2015-ly}%
  \BibitemOpen
  \bibfield  {author} {\bibinfo {author} {\bibfnamefont {D.~W.}\ \bibnamefont {Berry}}, \bibinfo {author} {\bibfnamefont {A.~M.}\ \bibnamefont {Childs}}, \bibinfo {author} {\bibfnamefont {R.}~\bibnamefont {Cleve}}, \bibinfo {author} {\bibfnamefont {R.}~\bibnamefont {Kothari}},\ and\ \bibinfo {author} {\bibfnamefont {R.~D.}\ \bibnamefont {Somma}},\ }\bibfield  {title} {{\selectlanguage {en}\bibinfo {title} {Simulating hamiltonian dynamics with a truncated taylor series}},\ }\href {https://doi.org/10.1103/PhysRevLett.114.090502} {\bibfield  {journal} {\bibinfo  {journal} {Phys. Rev. Lett.}\ }\textbf {\bibinfo {volume} {114}},\ \bibinfo {pages} {090502} (\bibinfo {year} {2015})}\BibitemShut {NoStop}%
\bibitem [{\citenamefont {Gilboa}\ \emph {et~al.}(2023)\citenamefont {Gilboa}, \citenamefont {Michaeli}, \citenamefont {Soudry},\ and\ \citenamefont {McClean}}]{Gilboa2023-db}%
  \BibitemOpen
  \bibfield  {author} {\bibinfo {author} {\bibfnamefont {D.}~\bibnamefont {Gilboa}}, \bibinfo {author} {\bibfnamefont {H.}~\bibnamefont {Michaeli}}, \bibinfo {author} {\bibfnamefont {D.}~\bibnamefont {Soudry}},\ and\ \bibinfo {author} {\bibfnamefont {J.}~\bibnamefont {McClean}},\ }\bibfield  {title} {\bibinfo {title} {Exponential quantum communication advantage in distributed inference and learning},\ }\href {https://proceedings.neurips.cc/paper_files/paper/2024/hash/3639efedc51a522595372f76b91cbb25-Abstract-Conference.html} {\bibfield  {journal} {\bibinfo  {journal} {Neural Inf Process Syst}\ }\textbf {\bibinfo {volume} {37}},\ \bibinfo {pages} {30425} (\bibinfo {year} {2023})},\ \Eprint {https://arxiv.org/abs/2310.07136} {2310.07136} \BibitemShut {NoStop}%
\bibitem [{\citenamefont {Huggins}\ and\ \citenamefont {McClean}(2023)}]{Huggins2023-nx}%
  \BibitemOpen
  \bibfield  {author} {\bibinfo {author} {\bibfnamefont {W.~J.}\ \bibnamefont {Huggins}}\ and\ \bibinfo {author} {\bibfnamefont {J.~R.}\ \bibnamefont {McClean}},\ }\bibfield  {title} {\bibinfo {title} {Accelerating quantum algorithms with precomputation},\ }\Eprint {https://arxiv.org/abs/2305.09638} {arXiv:2305.09638 [quant-ph]}  (\bibinfo {year} {2023})\BibitemShut {NoStop}%
\bibitem [{\citenamefont {Childs}\ \emph {et~al.}(2018)\citenamefont {Childs}, \citenamefont {Maslov}, \citenamefont {Nam}, \citenamefont {Ross},\ and\ \citenamefont {Su}}]{Childs2018-fq}%
  \BibitemOpen
  \bibfield  {author} {\bibinfo {author} {\bibfnamefont {A.~M.}\ \bibnamefont {Childs}}, \bibinfo {author} {\bibfnamefont {D.}~\bibnamefont {Maslov}}, \bibinfo {author} {\bibfnamefont {Y.}~\bibnamefont {Nam}}, \bibinfo {author} {\bibfnamefont {N.~J.}\ \bibnamefont {Ross}},\ and\ \bibinfo {author} {\bibfnamefont {Y.}~\bibnamefont {Su}},\ }\bibfield  {title} {{\selectlanguage {en}\bibinfo {title} {Toward the first quantum simulation with quantum speedup}},\ }\href {https://doi.org/10.1073/pnas.1801723115} {\bibfield  {journal} {\bibinfo  {journal} {Proc. Natl. Acad. Sci. U. S. A.}\ }\textbf {\bibinfo {volume} {115}},\ \bibinfo {pages} {9456} (\bibinfo {year} {2018})}\BibitemShut {NoStop}%
\bibitem [{\citenamefont {Reiher}\ \emph {et~al.}(2017)\citenamefont {Reiher}, \citenamefont {Wiebe}, \citenamefont {Svore}, \citenamefont {Wecker},\ and\ \citenamefont {Troyer}}]{reiher2017elucidating}%
  \BibitemOpen
  \bibfield  {author} {\bibinfo {author} {\bibfnamefont {M.}~\bibnamefont {Reiher}}, \bibinfo {author} {\bibfnamefont {N.}~\bibnamefont {Wiebe}}, \bibinfo {author} {\bibfnamefont {K.~M.}\ \bibnamefont {Svore}}, \bibinfo {author} {\bibfnamefont {D.}~\bibnamefont {Wecker}},\ and\ \bibinfo {author} {\bibfnamefont {M.}~\bibnamefont {Troyer}},\ }\bibfield  {title} {\bibinfo {title} {Elucidating reaction mechanisms on quantum computers},\ }\href@noop {} {\bibfield  {journal} {\bibinfo  {journal} {Proceedings of the national academy of sciences}\ }\textbf {\bibinfo {volume} {114}},\ \bibinfo {pages} {7555} (\bibinfo {year} {2017})}\BibitemShut {NoStop}%
\bibitem [{\citenamefont {H{\"a}ner}\ \emph {et~al.}(2016)\citenamefont {H{\"a}ner}, \citenamefont {Roetteler},\ and\ \citenamefont {Svore}}]{haner2016factoring}%
  \BibitemOpen
  \bibfield  {author} {\bibinfo {author} {\bibfnamefont {T.}~\bibnamefont {H{\"a}ner}}, \bibinfo {author} {\bibfnamefont {M.}~\bibnamefont {Roetteler}},\ and\ \bibinfo {author} {\bibfnamefont {K.~M.}\ \bibnamefont {Svore}},\ }\bibfield  {title} {\bibinfo {title} {Factoring using 2n+ 2 qubits with toffoli based modular multiplication},\ }\href@noop {} {\bibfield  {journal} {\bibinfo  {journal} {arXiv preprint arXiv:1611.07995}\ } (\bibinfo {year} {2016})}\BibitemShut {NoStop}%
\bibitem [{\citenamefont {Sanders}\ \emph {et~al.}(2020)\citenamefont {Sanders}, \citenamefont {Berry}, \citenamefont {Costa}, \citenamefont {Tessler}, \citenamefont {Wiebe}, \citenamefont {Gidney}, \citenamefont {Neven},\ and\ \citenamefont {Babbush}}]{sanders2020compilation}%
  \BibitemOpen
  \bibfield  {author} {\bibinfo {author} {\bibfnamefont {Y.~R.}\ \bibnamefont {Sanders}}, \bibinfo {author} {\bibfnamefont {D.~W.}\ \bibnamefont {Berry}}, \bibinfo {author} {\bibfnamefont {P.~C.}\ \bibnamefont {Costa}}, \bibinfo {author} {\bibfnamefont {L.~W.}\ \bibnamefont {Tessler}}, \bibinfo {author} {\bibfnamefont {N.}~\bibnamefont {Wiebe}}, \bibinfo {author} {\bibfnamefont {C.}~\bibnamefont {Gidney}}, \bibinfo {author} {\bibfnamefont {H.}~\bibnamefont {Neven}},\ and\ \bibinfo {author} {\bibfnamefont {R.}~\bibnamefont {Babbush}},\ }\bibfield  {title} {\bibinfo {title} {Compilation of fault-tolerant quantum heuristics for combinatorial optimization},\ }\href@noop {} {\bibfield  {journal} {\bibinfo  {journal} {PRX quantum}\ }\textbf {\bibinfo {volume} {1}},\ \bibinfo {pages} {020312} (\bibinfo {year} {2020})}\BibitemShut {NoStop}%
\bibitem [{\citenamefont {O'Gorman}\ and\ \citenamefont {Campbell}(2017)}]{O-Gorman2017-ul}%
  \BibitemOpen
  \bibfield  {author} {\bibinfo {author} {\bibfnamefont {J.}~\bibnamefont {O'Gorman}}\ and\ \bibinfo {author} {\bibfnamefont {E.~T.}\ \bibnamefont {Campbell}},\ }\bibfield  {title} {\bibinfo {title} {Quantum computation with realistic magic-state factories},\ }\bibfield  {journal} {\bibinfo  {journal} {Physical Review A}\ }\href {https://doi.org/10.1103/PhysRevA.95.032338} {10.1103/PhysRevA.95.032338} (\bibinfo {year} {2017})\BibitemShut {NoStop}%
\bibitem [{\citenamefont {Gidney}\ and\ \citenamefont {Fowler}(2019)}]{Gidney2019-mv}%
  \BibitemOpen
  \bibfield  {author} {\bibinfo {author} {\bibfnamefont {C.}~\bibnamefont {Gidney}}\ and\ \bibinfo {author} {\bibfnamefont {A.~G.}\ \bibnamefont {Fowler}},\ }\bibfield  {title} {\bibinfo {title} {Efficient magic state factories with a catalyzed $|ccz\rangle$ to $2|t\rangle$ transformation},\ }\href {https://doi.org/10.22331/q-2019-04-30-135} {\bibfield  {journal} {\bibinfo  {journal} {Quantum}\ }\textbf {\bibinfo {volume} {3}},\ \bibinfo {pages} {135} (\bibinfo {year} {2019})},\ \Eprint {https://arxiv.org/abs/1812.01238v3} {1812.01238v3} \BibitemShut {NoStop}%
\bibitem [{\citenamefont {Evered}\ \emph {et~al.}(2023)\citenamefont {Evered}, \citenamefont {Bluvstein}, \citenamefont {Kalinowski}, \citenamefont {Ebadi}, \citenamefont {Manovitz}, \citenamefont {Zhou}, \citenamefont {Li}, \citenamefont {Geim}, \citenamefont {Wang}, \citenamefont {Maskara}, \citenamefont {Levine}, \citenamefont {Semeghini}, \citenamefont {Greiner}, \citenamefont {Vuletic},\ and\ \citenamefont {Lukin}}]{Evered2023-bz}%
  \BibitemOpen
  \bibfield  {author} {\bibinfo {author} {\bibfnamefont {S.~J.}\ \bibnamefont {Evered}}, \bibinfo {author} {\bibfnamefont {D.}~\bibnamefont {Bluvstein}}, \bibinfo {author} {\bibfnamefont {M.}~\bibnamefont {Kalinowski}}, \bibinfo {author} {\bibfnamefont {S.}~\bibnamefont {Ebadi}}, \bibinfo {author} {\bibfnamefont {T.}~\bibnamefont {Manovitz}}, \bibinfo {author} {\bibfnamefont {H.}~\bibnamefont {Zhou}}, \bibinfo {author} {\bibfnamefont {S.~H.}\ \bibnamefont {Li}}, \bibinfo {author} {\bibfnamefont {A.~A.}\ \bibnamefont {Geim}}, \bibinfo {author} {\bibfnamefont {T.~T.}\ \bibnamefont {Wang}}, \bibinfo {author} {\bibfnamefont {N.}~\bibnamefont {Maskara}}, \bibinfo {author} {\bibfnamefont {H.}~\bibnamefont {Levine}}, \bibinfo {author} {\bibfnamefont {G.}~\bibnamefont {Semeghini}}, \bibinfo {author} {\bibfnamefont {M.}~\bibnamefont {Greiner}}, \bibinfo {author} {\bibfnamefont {V.}~\bibnamefont {Vuletic}},\ and\ \bibinfo {author} {\bibfnamefont {M.~D.}\ \bibnamefont {Lukin}},\ }\bibfield  {title} {\bibinfo {title} {High-fidelity parallel entangling gates on a neutral atom quantum computer},\ }\Eprint {https://arxiv.org/abs/2304.05420} {arXiv:2304.05420 [quant-ph]}  (\bibinfo {year} {2023})\BibitemShut {NoStop}%
\bibitem [{\citenamefont {Xu}\ \emph {et~al.}(2024)\citenamefont {Xu}, \citenamefont {Bonilla~Ataides}, \citenamefont {Pattison}, \citenamefont {Raveendran}, \citenamefont {Bluvstein}, \citenamefont {Wurtz}, \citenamefont {Vasić}, \citenamefont {Lukin}, \citenamefont {Jiang},\ and\ \citenamefont {Zhou}}]{Xu2024-zd}%
  \BibitemOpen
  \bibfield  {author} {\bibinfo {author} {\bibfnamefont {Q.}~\bibnamefont {Xu}}, \bibinfo {author} {\bibfnamefont {J.~P.}\ \bibnamefont {Bonilla~Ataides}}, \bibinfo {author} {\bibfnamefont {C.~A.}\ \bibnamefont {Pattison}}, \bibinfo {author} {\bibfnamefont {N.}~\bibnamefont {Raveendran}}, \bibinfo {author} {\bibfnamefont {D.}~\bibnamefont {Bluvstein}}, \bibinfo {author} {\bibfnamefont {J.}~\bibnamefont {Wurtz}}, \bibinfo {author} {\bibfnamefont {B.}~\bibnamefont {Vasić}}, \bibinfo {author} {\bibfnamefont {M.~D.}\ \bibnamefont {Lukin}}, \bibinfo {author} {\bibfnamefont {L.}~\bibnamefont {Jiang}},\ and\ \bibinfo {author} {\bibfnamefont {H.}~\bibnamefont {Zhou}},\ }\bibfield  {title} {{\selectlanguage {en}\bibinfo {title} {Constant-overhead fault-tolerant quantum computation with reconfigurable atom arrays}},\ }\href {https://doi.org/10.1038/s41567-024-02479-z} {\bibfield  {journal} {\bibinfo  {journal} {Nat. Phys.}\ }\textbf {\bibinfo {volume} {20}},\ \bibinfo {pages} {1084} (\bibinfo {year} {2024})}\BibitemShut {NoStop}%
\bibitem [{\citenamefont {Grover}\ and\ \citenamefont {Rudolph}(2002)}]{Grover2002-st}%
  \BibitemOpen
  \bibfield  {author} {\bibinfo {author} {\bibfnamefont {L.}~\bibnamefont {Grover}}\ and\ \bibinfo {author} {\bibfnamefont {T.}~\bibnamefont {Rudolph}},\ }\bibfield  {title} {\bibinfo {title} {Creating superpositions that correspond to efficiently integrable probability distributions},\ }\Eprint {https://arxiv.org/abs/quant-ph/0208112} {arXiv:quant-ph/0208112 [quant-ph]}  (\bibinfo {year} {2002})\BibitemShut {NoStop}%
\bibitem [{\citenamefont {Gosset}\ \emph {et~al.}(2024)\citenamefont {Gosset}, \citenamefont {Kothari},\ and\ \citenamefont {Wu}}]{Gosset2024-xi}%
  \BibitemOpen
  \bibfield  {author} {\bibinfo {author} {\bibfnamefont {D.}~\bibnamefont {Gosset}}, \bibinfo {author} {\bibfnamefont {R.}~\bibnamefont {Kothari}},\ and\ \bibinfo {author} {\bibfnamefont {K.}~\bibnamefont {Wu}},\ }\bibfield  {title} {\bibinfo {title} {Quantum state preparation with optimal {T}-count},\ }\Eprint {https://arxiv.org/abs/2411.04790} {arXiv:2411.04790 [quant-ph]}  (\bibinfo {year} {2024})\BibitemShut {NoStop}%
\bibitem [{\citenamefont {Ross}\ and\ \citenamefont {Selinger}(2014)}]{ross2014optimal}%
  \BibitemOpen
  \bibfield  {author} {\bibinfo {author} {\bibfnamefont {N.~J.}\ \bibnamefont {Ross}}\ and\ \bibinfo {author} {\bibfnamefont {P.}~\bibnamefont {Selinger}},\ }\bibfield  {title} {\bibinfo {title} {Optimal ancilla-free clifford+ t approximation of z-rotations},\ }\href@noop {} {\bibfield  {journal} {\bibinfo  {journal} {arXiv preprint arXiv:1403.2975}\ } (\bibinfo {year} {2014})}\BibitemShut {NoStop}%
\bibitem [{\citenamefont {Gilyén}\ \emph {et~al.}(2019)\citenamefont {Gilyén}, \citenamefont {Su}, \citenamefont {Low},\ and\ \citenamefont {Wiebe}}]{Gilyen2019-jt}%
  \BibitemOpen
  \bibfield  {author} {\bibinfo {author} {\bibfnamefont {A.}~\bibnamefont {Gilyén}}, \bibinfo {author} {\bibfnamefont {Y.}~\bibnamefont {Su}}, \bibinfo {author} {\bibfnamefont {G.~H.}\ \bibnamefont {Low}},\ and\ \bibinfo {author} {\bibfnamefont {N.}~\bibnamefont {Wiebe}},\ }\bibfield  {title} {\bibinfo {title} {Quantum singular value transformation and beyond: exponential improvements for quantum matrix arithmetics},\ }in\ \href {https://doi.org/10.1145/3313276.3316366} {\emph {\bibinfo {booktitle} {Proceedings of the 51st Annual ACM SIGACT Symposium on Theory of Computing}}},\ \bibinfo {series and number} {STOC 2019}\ (\bibinfo  {publisher} {ACM},\ \bibinfo {address} {New York, NY, USA},\ \bibinfo {year} {2019})\ pp.\ \bibinfo {pages} {193--204}\BibitemShut {NoStop}%
\bibitem [{\citenamefont {Low}\ and\ \citenamefont {Chuang}(2019)}]{Low2019-oa}%
  \BibitemOpen
  \bibfield  {author} {\bibinfo {author} {\bibfnamefont {G.~H.}\ \bibnamefont {Low}}\ and\ \bibinfo {author} {\bibfnamefont {I.~L.}\ \bibnamefont {Chuang}},\ }\bibfield  {title} {{\selectlanguage {en}\bibinfo {title} {Hamiltonian simulation by qubitization}},\ }\href {https://doi.org/10.22331/q-2019-07-12-163} {\bibfield  {journal} {\bibinfo  {journal} {Quantum}\ }\textbf {\bibinfo {volume} {3}},\ \bibinfo {pages} {163} (\bibinfo {year} {2019})}\BibitemShut {NoStop}%
\bibitem [{\citenamefont {Martyn}\ \emph {et~al.}(2021)\citenamefont {Martyn}, \citenamefont {Rossi}, \citenamefont {Tan},\ and\ \citenamefont {Chuang}}]{Martyn2021-mf}%
  \BibitemOpen
  \bibfield  {author} {\bibinfo {author} {\bibfnamefont {J.~M.}\ \bibnamefont {Martyn}}, \bibinfo {author} {\bibfnamefont {Z.~M.}\ \bibnamefont {Rossi}}, \bibinfo {author} {\bibfnamefont {A.~K.}\ \bibnamefont {Tan}},\ and\ \bibinfo {author} {\bibfnamefont {I.~L.}\ \bibnamefont {Chuang}},\ }\bibfield  {title} {{\selectlanguage {en}\bibinfo {title} {Grand unification of quantum algorithms}},\ }\href {https://doi.org/10.1103/prxquantum.2.040203} {\bibfield  {journal} {\bibinfo  {journal} {PRX quantum}\ }\textbf {\bibinfo {volume} {2}},\ \bibinfo {pages} {040203} (\bibinfo {year} {2021})}\BibitemShut {NoStop}%
\bibitem [{\citenamefont {McClean}\ \emph {et~al.}(2014)\citenamefont {McClean}, \citenamefont {Babbush}, \citenamefont {Love},\ and\ \citenamefont {Aspuru-Guzik}}]{McClean2014-ll}%
  \BibitemOpen
  \bibfield  {author} {\bibinfo {author} {\bibfnamefont {J.~R.}\ \bibnamefont {McClean}}, \bibinfo {author} {\bibfnamefont {R.}~\bibnamefont {Babbush}}, \bibinfo {author} {\bibfnamefont {P.~J.}\ \bibnamefont {Love}},\ and\ \bibinfo {author} {\bibfnamefont {A.}~\bibnamefont {Aspuru-Guzik}},\ }\bibfield  {title} {{\selectlanguage {en}\bibinfo {title} {Exploiting locality in quantum computation for quantum chemistry}},\ }\Eprint {https://arxiv.org/abs/1407.7863} {arXiv:1407.7863 [quant-ph]}  (\bibinfo {year} {2014})\BibitemShut {NoStop}%
\bibitem [{\citenamefont {Svore}\ \emph {et~al.}(2013)\citenamefont {Svore}, \citenamefont {Hastings},\ and\ \citenamefont {Freedman}}]{svore2013faster}%
  \BibitemOpen
  \bibfield  {author} {\bibinfo {author} {\bibfnamefont {K.~M.}\ \bibnamefont {Svore}}, \bibinfo {author} {\bibfnamefont {M.~B.}\ \bibnamefont {Hastings}},\ and\ \bibinfo {author} {\bibfnamefont {M.}~\bibnamefont {Freedman}},\ }\bibfield  {title} {\bibinfo {title} {Faster phase estimation},\ }\href@noop {} {\bibfield  {journal} {\bibinfo  {journal} {arXiv preprint arXiv:1304.0741}\ } (\bibinfo {year} {2013})}\BibitemShut {NoStop}%
\bibitem [{\citenamefont {Perez-Garcia}\ \emph {et~al.}(2007)\citenamefont {Perez-Garcia}, \citenamefont {Verstraete}, \citenamefont {Wolf},\ and\ \citenamefont {Cirac}}]{Perez-Garcia2007-wt}%
  \BibitemOpen
  \bibfield  {author} {\bibinfo {author} {\bibfnamefont {D.}~\bibnamefont {Perez-Garcia}}, \bibinfo {author} {\bibfnamefont {F.}~\bibnamefont {Verstraete}}, \bibinfo {author} {\bibfnamefont {M.~M.}\ \bibnamefont {Wolf}},\ and\ \bibinfo {author} {\bibfnamefont {J.~I.}\ \bibnamefont {Cirac}},\ }\bibfield  {title} {\bibinfo {title} {Matrix product state representations},\ }\href {https://doi.org/10.26421/qic7.5-6-1} {\bibfield  {journal} {\bibinfo  {journal} {Quantum Inf. Comput.}\ }\textbf {\bibinfo {volume} {7}},\ \bibinfo {pages} {401} (\bibinfo {year} {2007})}\BibitemShut {NoStop}%
\bibitem [{\citenamefont {Oseledets}(2011)}]{Oseledets2011-xm}%
  \BibitemOpen
  \bibfield  {author} {\bibinfo {author} {\bibfnamefont {I.~V.}\ \bibnamefont {Oseledets}},\ }\bibfield  {title} {\bibinfo {title} {Tensor-train decomposition},\ }\href {https://doi.org/10.1137/090752286} {\bibfield  {journal} {\bibinfo  {journal} {SIAM J. Sci. Comput.}\ }\textbf {\bibinfo {volume} {33}},\ \bibinfo {pages} {2295} (\bibinfo {year} {2011})}\BibitemShut {NoStop}%
\bibitem [{\citenamefont {Fomichev}\ \emph {et~al.}(2024)\citenamefont {Fomichev}, \citenamefont {Hejazi}, \citenamefont {Zini}, \citenamefont {Kiser}, \citenamefont {Fraxanet}, \citenamefont {Casares}, \citenamefont {Delgado}, \citenamefont {Huh}, \citenamefont {Voigt}, \citenamefont {Mueller} \emph {et~al.}}]{fomichev2024initial}%
  \BibitemOpen
  \bibfield  {author} {\bibinfo {author} {\bibfnamefont {S.}~\bibnamefont {Fomichev}}, \bibinfo {author} {\bibfnamefont {K.}~\bibnamefont {Hejazi}}, \bibinfo {author} {\bibfnamefont {M.~S.}\ \bibnamefont {Zini}}, \bibinfo {author} {\bibfnamefont {M.}~\bibnamefont {Kiser}}, \bibinfo {author} {\bibfnamefont {J.}~\bibnamefont {Fraxanet}}, \bibinfo {author} {\bibfnamefont {P.~A.~M.}\ \bibnamefont {Casares}}, \bibinfo {author} {\bibfnamefont {A.}~\bibnamefont {Delgado}}, \bibinfo {author} {\bibfnamefont {J.}~\bibnamefont {Huh}}, \bibinfo {author} {\bibfnamefont {A.-C.}\ \bibnamefont {Voigt}}, \bibinfo {author} {\bibfnamefont {J.~E.}\ \bibnamefont {Mueller}}, \emph {et~al.},\ }\bibfield  {title} {\bibinfo {title} {Initial state preparation for quantum chemistry on quantum computers},\ }\href@noop {} {\bibfield  {journal} {\bibinfo  {journal} {PRX Quantum}\ }\textbf {\bibinfo {volume} {5}},\ \bibinfo {pages} {040339} (\bibinfo {year} {2024})}\BibitemShut {NoStop}%
\bibitem [{\citenamefont {Berry}\ \emph {et~al.}(2024)\citenamefont {Berry}, \citenamefont {Tong}, \citenamefont {Khattar}, \citenamefont {White}, \citenamefont {Kim}, \citenamefont {Boixo}, \citenamefont {Lin}, \citenamefont {Lee}, \citenamefont {Chan}, \citenamefont {Babbush},\ and\ \citenamefont {Rubin}}]{Berry2024-qe}%
  \BibitemOpen
  \bibfield  {author} {\bibinfo {author} {\bibfnamefont {D.~W.}\ \bibnamefont {Berry}}, \bibinfo {author} {\bibfnamefont {Y.}~\bibnamefont {Tong}}, \bibinfo {author} {\bibfnamefont {T.}~\bibnamefont {Khattar}}, \bibinfo {author} {\bibfnamefont {A.}~\bibnamefont {White}}, \bibinfo {author} {\bibfnamefont {T.~I.}\ \bibnamefont {Kim}}, \bibinfo {author} {\bibfnamefont {S.}~\bibnamefont {Boixo}}, \bibinfo {author} {\bibfnamefont {L.}~\bibnamefont {Lin}}, \bibinfo {author} {\bibfnamefont {S.}~\bibnamefont {Lee}}, \bibinfo {author} {\bibfnamefont {G.~K.-L.}\ \bibnamefont {Chan}}, \bibinfo {author} {\bibfnamefont {R.}~\bibnamefont {Babbush}},\ and\ \bibinfo {author} {\bibfnamefont {N.~C.}\ \bibnamefont {Rubin}},\ }\bibfield  {title} {\bibinfo {title} {Rapid initial state preparation for the quantum simulation of strongly correlated molecules},\ }\Eprint {https://arxiv.org/abs/2409.11748} {arXiv:2409.11748 [quant-ph]}  (\bibinfo {year} {2024})\BibitemShut {NoStop}%
\bibitem [{\citenamefont {Brassard}\ \emph {et~al.}(2000)\citenamefont {Brassard}, \citenamefont {Hoyer}, \citenamefont {Mosca},\ and\ \citenamefont {Tapp}}]{brassard2000quantum}%
  \BibitemOpen
  \bibfield  {author} {\bibinfo {author} {\bibfnamefont {G.}~\bibnamefont {Brassard}}, \bibinfo {author} {\bibfnamefont {P.}~\bibnamefont {Hoyer}}, \bibinfo {author} {\bibfnamefont {M.}~\bibnamefont {Mosca}},\ and\ \bibinfo {author} {\bibfnamefont {A.}~\bibnamefont {Tapp}},\ }\bibfield  {title} {\bibinfo {title} {Quantum amplitude amplification and estimation},\ }\href@noop {} {\bibfield  {journal} {\bibinfo  {journal} {arXiv preprint quant-ph/0005055}\ } (\bibinfo {year} {2000})}\BibitemShut {NoStop}%
\bibitem [{\citenamefont {Zalka}(1999)}]{Zalka1999-no}%
  \BibitemOpen
  \bibfield  {author} {\bibinfo {author} {\bibfnamefont {C.}~\bibnamefont {Zalka}},\ }\bibfield  {title} {\bibinfo {title} {Grover's quantum searching algorithm is optimal},\ }\href {https://doi.org/10.1103/PhysRevA.60.2746} {\bibfield  {journal} {\bibinfo  {journal} {Phys. Rev. A}\ }\textbf {\bibinfo {volume} {60}},\ \bibinfo {pages} {2746} (\bibinfo {year} {1999})}\BibitemShut {NoStop}%
\bibitem [{\citenamefont {Yoder}\ \emph {et~al.}(2014)\citenamefont {Yoder}, \citenamefont {Low},\ and\ \citenamefont {Chuang}}]{Yoder2014-ek}%
  \BibitemOpen
  \bibfield  {author} {\bibinfo {author} {\bibfnamefont {T.~J.}\ \bibnamefont {Yoder}}, \bibinfo {author} {\bibfnamefont {G.~H.}\ \bibnamefont {Low}},\ and\ \bibinfo {author} {\bibfnamefont {I.~L.}\ \bibnamefont {Chuang}},\ }\bibfield  {title} {{\selectlanguage {en}\bibinfo {title} {Fixed-point quantum search with an optimal number of queries}},\ }\href {https://doi.org/10.1103/PhysRevLett.113.210501} {\bibfield  {journal} {\bibinfo  {journal} {Phys. Rev. Lett.}\ }\textbf {\bibinfo {volume} {113}},\ \bibinfo {pages} {210501} (\bibinfo {year} {2014})}\BibitemShut {NoStop}%
\end{thebibliography}
\end{document}